\documentclass[review]{elsarticle}

\usepackage{hyperref}

\usepackage{amsthm}
\usepackage{amsmath}
\usepackage{graphicx}
\usepackage{comment}
\usepackage{subfigure}
\usepackage{xfrac}
\usepackage[usenames, dvipsnames]{color}
\usepackage{multicol}
\usepackage{bm}
\usepackage{ctable}
\usepackage{longtable}
\usepackage{lipsum}
\usepackage{float}
\usepackage{moreverb}
\usepackage[utf8]{inputenc}
\usepackage[english]{babel}
\usepackage{natbib}
\usepackage{multirow}
\usepackage{dcolumn}
\usepackage{siunitx}
\usepackage{dcolumn,booktabs}
\usepackage{array}
\usepackage{caption} 
\usepackage{bbm}
\usepackage{hhline}

\usepackage[normalem]{ulem}

\usepackage[margin=1in]{geometry}

\newcolumntype{d}[1]{D{.}{.}{#1}}

\newtheorem{theorem}{Theorem}[section]
\newtheorem{prop}{Proposition}[section]

\newtheorem{assumption}{Assumption}[section]
\newtheorem{app}{Approximation Procedure}[subsection]

\journal{Journal of Multivariate Analysis}

\begin{document}

\begin{frontmatter}

\title{Nonparametric Methods for Complex Multivariate Data: Asymptotics and Small Sample Approximations}

\author{Yue Cui\corref{cor1}
}
\ead{YueCui@missouristate.edu}
\address{Department of Mathematics, Missouri State University, Springfield, MO}


\author{Solomon W. Harrar
}
\ead{Solomon.Harrar@uky.edu}
\address{Department of Statistics, University of Kentucky, Lexington, Kentucky}



\cortext[cor1]{Corresponding author}

\begin{abstract}
Quality of Life (QOL) outcomes are important in the management of chronic illnesses. In studies of efficacies of treatments or intervention modalities, QOL scales --multi-dimensional constructs-- are routinely used as primary endpoints. The standard data analysis strategy computes composite (average) overall and domain scores, and conducts a mixed-model analysis for evaluating efficacy or monitoring medical conditions as if these scores were in continuous metric scale. However, assumptions of parametric models like continuity and homoscedastivity can be violated in many cases. Furthermore, it is even more challenging when there are missing values on some of the variables. In this paper, we propose a purely nonparametric approach in the sense that meaningful and, yet, nonparametric effect size measures are developed. We propose estimator for the effect size and develop the asymptotic properties. Our methods are shown to be particularly effective in the presence of some form of clustering and/or missing values. Inferential procedures are derived from the asymptotic theory. The \textit{Asthma Randomized Trial of Indoor Wood Smoke} data will be used to illustrate the applications of the proposed methods. The data was collected from a three-arm randomized trial which evaluated interventions targeting biomass smoke particulate matter from older model residential wood stoves in homes that have children with asthma.
\end{abstract}

\begin{keyword}
Quality of Life outcomes; Multivariate two-sample problem; Nonparametric effect size measure; Missing data; Rank test
\end{keyword}

\end{frontmatter}


\section{Introduction}
\label{sec:Moti-and-ntro}
Multivariate data commonly arise in medical, sociological and behavioural research. For example, in clinical trials the primary outcome is typically assessed at multiple time points and compared across treatment groups. Multivariate data also arise when the experimental units are subjected to multiple treatments and outcome of interest is measured by multiple variables. Such data are usually analyzed by parametric MANOVA tests such as Wilks' $\Lambda$. However, the corresponding assumptions such as multivariate normality and homoscedasticity are hard to meet in practice and therefore, may lead to wrong results. Moreover, if data are non-metric, such as ordinal or ordered categorical, the classic parametric models are not applicable since means no longer provide meaningful measures of effect size. To overcome these limitations, some progress has been made with rank-based methods for classical multivariate problems, e.g. Sen and Puri \cite{PuriSen-1971}. However, the body of research documented in this book assumes a continuous location family for data distributions.

A fully nonparametric univariate model was proposed by Brunner and Munzel \cite{baby-case-2000} for two independent samples. This nonparametric model abandons the assumptions of continuity or normality of distribution functions and provides a nonparametric measure of effect size. This model was further extended by Konietschke et al. \cite{matched-pair-2012} to the matched-pair data under the assumption of missing completely at random and, thus, observations in two samples are allowed to be dependent. 

In the multivariate setting, let observations be modeled by independent random vectors $\bm{X}_{ij}=(X_{ij}^{(1)},\cdots,$ $X_{ij}^{(d)})^\top\sim\bm{F}_{ij}$, $i=1,\cdots,a$, $j=1,\cdots,n_i$, with possibly dependent components and marginal distributions given by $X_{ij}^{(k)}\sim F_i^{(k)}$. 
Munzel and Brunner \cite{munzel2000nonparametric, munzel2000nonparametric-unbalanced} developed a nonparametric approach to the analysis of multivariate data. Their results were derived using the asymptotic theory for rank statistics developed in, e.g., Akritas and Arnold \cite{akritas1994fully}, Brunner and Denker \cite{brunner1994rank}, and Akritas et al. \cite{akritas1997nonparametric}. Specifically, based on the assumption that the sample size per treatment tends to infinity, whereas the number of treatments (factor levels, cells) remain fixed (large $n_i$ small $a$ case), they derived asymptotic results to test both the multivariate null hypothesis $H^F_0:\bm{F}_1=\cdots=\bm{F}_a$ and the marginal null hypothesis $\overline{H}^F_0:F_1^{(l)}=\cdots=F_a^{(l)}$ for $l=1,\cdots,d$. This problem has also been investigated by Bathke and Harrar \cite{bathke2008nonparametric} and Harrar and Bathke \cite{harrar2008nonparametric} in the asymptotic setting where the number of treatments tends to infinity whereas the sample size per treatment is fixed (large $a$ small $n_i$ case). In these two papers, multivariate factorial structures were considered  
in the balanced and unbalanced designs, respectively. Different from the aforementioned nonparametric tests on multivariate null hypothesis or marginal null hypothesis, Brunner et al. \cite{brunner2002multivariate} generalized Brunner and Munzel \cite{baby-case-2000} to the fully nonparametric model and developed inferential  methods for the so-called nonparametric effect sizes, $\bm{p}=(p_1,\cdots,p_d)^\top$, where $p_l$ is the nonparametric effect size on the $l^{th}$ response variable. 

The aim of the present paper is to generalize Brunner et al. \cite{brunner2002multivariate} to complex multivariate data structures. In Brunner et al. \cite{brunner2002multivariate}, each subject can receive only one of the treatments so that the two samples are independent. In this paper, we relax this condition by allowing some subjects to receive both treatments and, therefore, the two samples to be  dependent. The resulting data structure may be regarded as missing data problem in the sense that subjects that received only one treatment are considered as having missing data on the other treatment. Existing methods or strategies for handing missing data in multivariate models assume unrealistic missing patterns, e.g. monotone missing pattern. It is also the aim of this paper to derive inferential methods for multivariate data in the presence of missing data that occur at component (variable) level rather than treatment level. In addition to asymptotic solution for the above two problems, we also propose approximations for small samples along the ideas of Brunner and Dette \cite{brunner1997box} and Brunner et al. \cite{brunner2002multivariate}. 

The paper is organized as follows. Section \ref{sec:Motivation} describes \textit{Asthma Randomized Trial of Indoor Wood Smoke} (ARTIS), which is the motivation for methods developed in this paper. A precise formulation of the statistical model, nonparametric measure of effect size and hypothesis of no treatment effect are introduced in Section \ref{sec:ModelandHypothesis}. The asymptotic theory is developed in Section \ref{sec:EstimatorandDistribution}. More specifically, the asymptotic multivariate normal distribution of the estimator of nonparametric effect size is derived in a closed form. Estimator of the covariance matrix of the limiting distribution is also derived. The covariance matrix is estimated by manipulating the so-called overall and internal rankings, and it is proved to be $L_2$ consistent. Based on the main results in Section \ref{sec:EstimatorandDistribution}, two tests of treatment effects are proposed in Section \ref{sec:TestStat}. These tests are the rank versions of the Wald-type and ANOVA-type tests. In Section \ref{sec:flex-missing}, extensions to general missing patterns are presented. The numerical accuracy of both statistics under the two missing patterns are investigated in a simulation study in Section \ref{sec:SimuRes}. Application of the new methods is illustrated with the ARTIS data set in Section \ref{sec:RealData}. Discussion and some concluding remarks are provided  in Section \ref{sec:Discussion}. All proofs and technical details are placed in the Appendix.

\section{Motivating Examples}
\label{sec:Motivation}
The research in this paper is motivated by the \textit{Asthma Randomized Trial of Indoor Wood Smoke} (ARTIS) (Noonan et al. \citep{ARTIS-data}), which investigates the impact of home heating sources on quality of life for children with asthma. The study is a randomized placebo-controlled intervention trial with two intervention strategies for reducing in-home woodsmoke particulate matter (PM). Eligible participants included children with asthma, age 6-8 years old, residing in a non-tobacco-smoking household that used older-model wood stoves as primary heating source. All children with asthma in a household were included in the study. ARTIS was conducted over 5 years and each household participated in two consecutive winter periods during which they experienced household interventions. The households were assigned randomly to three treatments: placebo group receiving sham air-filtration devices, air-filter group receiving air-filtration devices and wood stove group receiving improved-technology wood-burning appliances. 

The primary health outcome was the score on Pediatric Asthma Quality of Life Questionnaire (PAQLQ) administered directly to children at each visit. PAQLQ is a 23-item asthma specific battery that provides an overall score as well as domain scores for activity limitations (5 items), symptoms (10 items), and emotional function (8 items), with each item rated in a 7 points scale. The higher the score, the better the clinical record. The design in the ARTIS has multiple level of nesting (clustering). Data is collected from each child multiple times in the pre- and post-intervention winters. In some houses, multiple children are enrolled. In Cui et al. \cite{cui2020}, the authors avoided clustering effect of children in the same house by randomly selecting one of the children from each household that has multiple participating children. For each selected child, multiple measurements from multiple visits are regarded as dependent replicates. However, the methods in Cui et al. \cite{cui2020} can only analyze one item or domain variable at a time. For multiple items or domain variables, multiplicity adjustment is needed to correct the p-values.  

The ideas of the nonparametric test proposed in Cui et al. \cite{cui2020} can be extended to the multivariate data so that the response variables can be tested jointly. Here, to facilitate a smooth development, all clustering effects are ignored. Therefore, we confine attention to one visit per child and one child per home. Details on missing patterns and sample size allocations based on the randomly chosen sample is provided in Section \ref{sec:RealData}. 

Boxplots for domain variables at each intervention period in the air-filter group are shown in Figure \ref{boxplot}. The objective of the study is to monitor the pattern of change in the domain scores before and after the air-filter intervention. Note that the outcomes are ordered categorical and, therefore, analyzing them using mean-based or other parametric effect size measures would be inappropriate. Also, it can be readily seen from the boxplots that the scores in post-intervention period tend to be greater than those in the pre-intervention period for all domains. Of major interest is estimating appropriate intervention effects and testing whether the air-filter intervention effect is significant, i.e. PAQLQ scores of post-intervention group are higher than those in pre-intervention group on at least one domain variable. 

\begin{figure}[htb]
	\caption{Box plot of the Quality of Life scores for each domain.}
	\label{boxplot}
	\centering
	\includegraphics[scale=0.4]{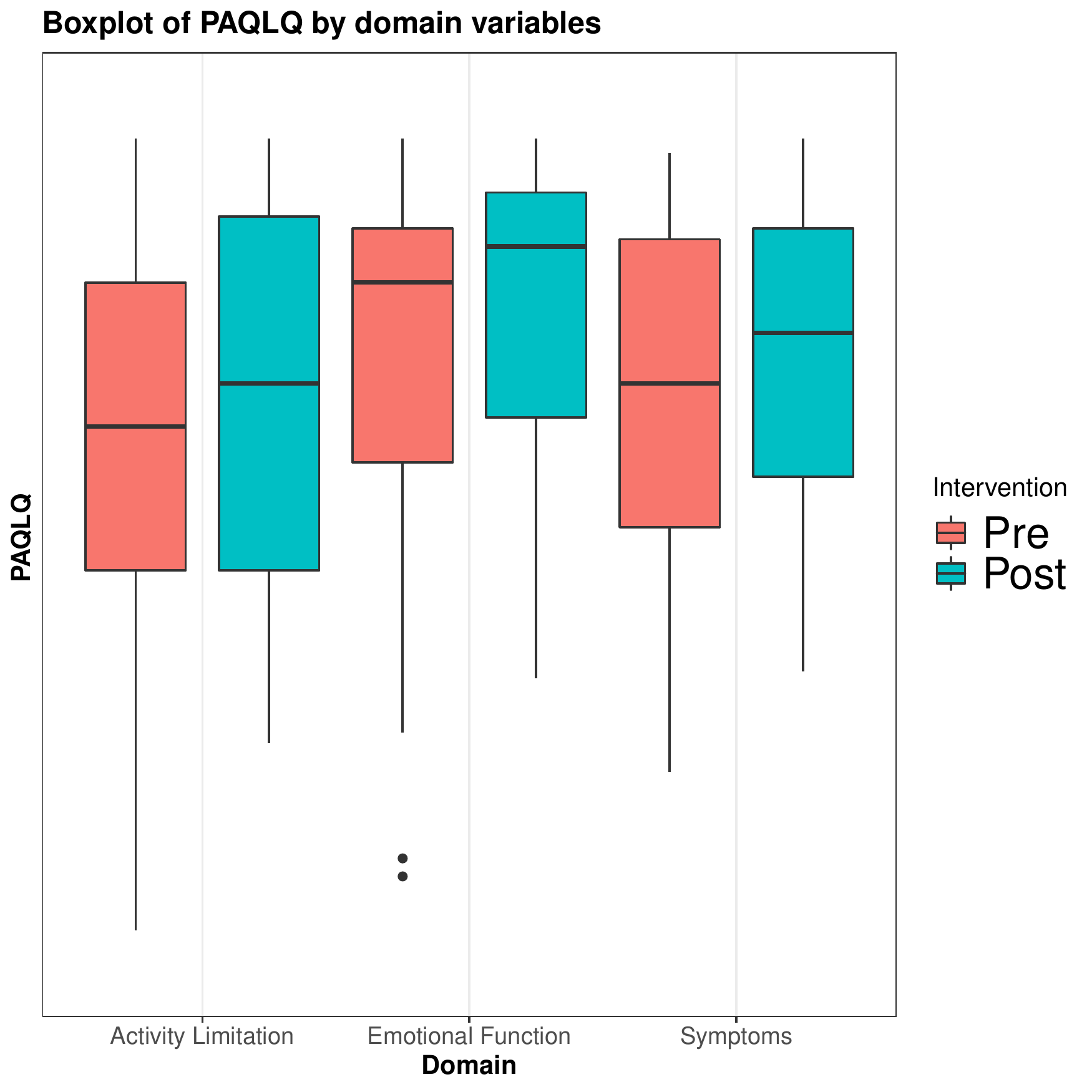}
\end{figure}

\section{Statistical Model and Hypothesis}
\label{sec:ModelandHypothesis}
We consider multivariate data from $n$ independent random vectors. Each vector contains observations of a subject on $d$ response variables and for two treatment groups, i.e. subjects may receive both treatments. This situation may arise by design or as a result of missing values. Define a subject as a \textit{complete} case on the $l^{th}$ component (variable) if it has measurements on the $l^{th}$ component for both treatment groups, and an \textit{incomplete} case if he/she has measurement on the $l^{th}$ component in the $g^{th}$ group only.

Denote $X_{gj}^{(c)(l)}$ as observation on the $l^{th}$ component of the $j^{th}$ complete case in the $g^{th}$ treatment group and $X_{gk}^{(i)(l)}$ as observation on the $l^{th}$ component of the $k^{th}$ incomplete case in the $g^{th}$ treatment group. Further denote $n_c^{(l)}$ as the number of complete cases on the $l^{th}$ component and $n_{g}^{(l)}$ as the number of incomplete cases on the $l^{th}$ component within the $g^{th}$ group, $g=1,2$.

To facilitate ease of presentation, we first focus on a simple structure, where a subject will have no data in a group if there is no data on any component within that group. Under this structure, $n_c^{(l)}$, $n_{1}^{(l)}$ and $n_{2}^{(l)}$ do not depend on $l$ and, hence, we refer to $n_c,n_1$ and $n_2$, respectively. Later in Section \ref{sec:flex-missing}, we present the generalizations to arbitrary structures (missing patterns).

Then data for complete cases are collected in vectors as
\[\textbf{X}_j^{(c)}=(\textbf{X}_{1j}^{(c)},\textbf{X}_{2j}^{(c)})\textrm{ where } \textbf{X}_{gj}^{(c)}=(X_{gj}^{(c)(1)},\cdots,X_{gj}^{(c)(d)}),
g=1,2,j=1,\cdots,n_c\]
and data for incomplete cases are similarly collected in vectors as 
\[\textbf{X}_{gk}^{(i)}=(X_{gk}^{(i)(1)},\cdots,X_{gk}^{(i)(d)}),g=1,2,k=1,\cdots,n_g.\]

Assume the joint distribution of data in the same treatment group to be identical for all subjects, i.e. $\bm{X}_{gj}^{(c)}$, $j=1,\cdots,n_c$ and $\bm{X}_{gj'}^{(i)}$, $j'=1,\cdots,n_g$, are \textit{i.i.d.} within treatment group $g=1,2$. Also, marginal distributions for complete and incomplete cases on the same component and in the same treatment group are assumed to be the same, i.e. $X_{gk}^{(A)(l)}\sim F_g^{(l)}\text{ for }A\in\{c,i\}$. The data scheme of the above model is displayed in Table \ref{Table:schematic} where each row contains observations of a subject and * indicates that data are not available for the subject in that group.
\begin{table}[H]
	\centering
	\caption{Schematic representation of the data. Stars denote missing observations.}
	\label{Table:schematic}
	\setlength\tabcolsep{4 pt}
	\begin{tabular}{|c|ccccc|}
		\hline
		\centering Subject& \multicolumn{2}{c}{TX=1} & \multicolumn{2}{c}{TX=2} &  \\
		\hline
		1 & $\textbf{X}_{11}^{(c)}=$ & $(X_{11}^{(c)(1)},\cdots,X_{11}^{(c)(d)})$ & $\textbf{X}_{21}^{(c)}=$ & $(X_{21}^{(c)(1)},\cdots,X_{21}^{(c)(d)})$ &  \\
		2 & $\textbf{X}_{12}^{(c)}=$ & $(X_{12}^{(c)(1)},\cdots,X_{12}^{(c)(d)})$ & $\textbf{X}_{22}^{(c)}=$ & $(X_{22}^{(c)(1)},\cdots,X_{22}^{(c)(d)})$ &  \\
		$\vdots$ &  \multicolumn{2}{c}{$\vdots$} &  \multicolumn{2}{c}{$\vdots$} &  \\
		$n_c$ & $\textbf{X}_{1n_c}^{(c)}=$ & $(X_{1n_c}^{(c)(1)},\cdots,X_{1n_c}^{(c)(d)})$ & $\textbf{X}_{2n_c}^{(c)}=$ & $(X_{2n_c}^{(c)(1)},\cdots,X_{2n_c}^{(c)(d)})$ & \\
		\hline
		$n_c+1$ & $\textbf{\textbf{X}}_{11}^{(i)}=$ & $(X_{11}^{(i)(1)},\cdots,X_{11}^{(i)(d)})$ &\multicolumn{2}{c}{*} &  \\
		$\vdots$ &  \multicolumn{2}{c}{$\vdots$} & \multicolumn{2}{c}{$\vdots$}  &   \\
		$n_c+n_1$ & $\textbf{X}_{1n_1}^{(i)}=$ & $(X_{1n_1}^{(i)(1)},\cdots,X_{1n_1}^{(i)(d)})$ & \multicolumn{2}{c}{*}&  \\
		\hline
		$n_c+n_1+1$ & \multicolumn{2}{c}{*} & $\textbf{X}_{21}^{(i)}=$ & $(X_{21}^{(i)(1)},\cdots,X_{21}^{(i)(d)})$ &  \\
		$\vdots$ &  \multicolumn{2}{c}{$\vdots$}  &  \multicolumn{2}{c}{$\vdots$} &  \\
		$n_c+n_1+n_2$ & \multicolumn{2}{c}{*} & $\textbf{X}_{2n_2}^{(i)}=$ & $(X_{2n_2}^{(i)(1)},\cdots,X_{2n_2}^{(i)(d)})$ &  \\
		\hline
	\end{tabular}
\end{table}                

In order to accommodate metric, discrete, dichotomous as well as ordinal data in a unified way, we use the normalized distribution function (Ruymgaart \citep{Ruymgaart-1980}), i.e. $F_g^{(l)}(x)=\frac{1}{2}[F_g^{(l)+}(x)+F_g^{(l)-}(x)]$, which is the average of left-continuous 
and right-continuous
distribution functions. That is,
\[F_g^{(l)}(x)=\frac{1}{2}[P(X_{gk}^{(A)(l)}\le x)+P(X_{gk}^{(A)(l)}< x)],\quad\textrm{for}\quad g=1,2,l=1,\cdots,d\quad\textrm{and}\quad A\in\{c,i\}.\]
Note that $F_g^{(l)}(x)$ is an arbitrary distribution function except the trivial case of one-point distribution. 

The statistical model considered here does not entail any parameter by which adequate treatment effects can be described. Therefore, treatment effects are defined via marginal distributions $F_1^{(l)}(x)$ and $F_2^{(l)}(x)$ as
\begin{equation}
	\label{p-l}
	p^{(l)}=\int F_1^{(l)}dF_2^{(l)}=p(X_{11}^{(c)(l)}\le X_{21}^{(c)(l)})+\frac{1}{2}P(X_{11}^{(c)(l)}= X_{21}^{(c)(l)}),\quad l=1,\cdots,d,
\end{equation}                                                                                                                                    
which is the so-called nonparametric relative treatment effect for the $l^{th}$ response, $l=1,\cdots,d$. Actually, $p^{(l)}$ is the generalized Wilcoxon-Mann-Whiteney (WMW) effect introduced in Brunner and Munzel \cite{baby-case-2000}. The tendency of observations from $F_1^{(l)}$ being smaller or larger than $F_2^{(l)}$ can be assessed by the comparisons $p^{(l)}>\frac{1}{2}$ or $p^{(l)}<\frac{1}{2}$. Especially, when $p^{(l)}=\frac{1}{2}$, observations from $F_1^{(l)}$ do not tend to be smaller or larger than observations from $F_2^{(l)}$, and this situation is regarded as "no treatment effect" for the $l^{th}$ component. In this respect, we can characterize the case of "no treatment effect" in the general multivariate model by the condition that $p^{(l)}=\frac{1}{2}$ for all $l=1,\cdots,d$. Now, let $\bm{p}=(p^{(1)},\cdots,p^{(d)})^\top$ be the vector of relative treatment effects. Then the hypothesis of no treatment effect in the multivariate model can be expressed as 
$$H_0:\bm{p}=\frac{1}{2}\mathbbm{1}_d,$$
where $\mathbbm{1}_d=(1,\cdots,1)^\top$ is a vector of all 1's. Equivalently, $H_0:p^{(l)}=\frac{1}{2}\quad\textrm{for all}\quad l=1,\cdots,d$.

One can easily show that $F_1^{(l)}=F_2^{(l)}$ implies $p^{(l)}=\frac{1}{2}$, but the converse is not necessarily true. Indeed, the inference method for $\bm{p}$ allows for heteroscedastic variances, skewness and higher moments of distributions in the two groups. 

To construct an appropriate inference procedure, we need to come up with a consistent estimator for the effect size vector and derive its asymptotic distribution. In the next section, we will derive an (asymptotically) unbiased and consistent estimator of $\bm{p}$ and establish its asymptotic normality. In order to derive the results, we need the following regularity assumption:

\begin{assumption}
	\label{assumption-3.1}
	$n_c+n_g\to\infty$ such that $\frac{n}{n_c+n_g}\le N_0<\infty$, $g=1,2$.
\end{assumption}
This assumption guarantees that either the number of complete cases or both of the incomplete cases are large. In particular, we do not want either $n_1$ or $n_2$ alone to dominate the total sample size. Therefore, it includes the most common practice-oriental sample size allocations:

\begin{enumerate}
	\item[(a)] $n_c\rightarrow \infty,n_1,n_2\le M<\infty,$ or
	\item[(b)] $n_c\rightarrow\infty,n_1\rightarrow\infty,n_2\le M<\infty,$ or
	\item[(c)] $n_c\rightarrow\infty,n_1\rightarrow\infty,n_2\rightarrow\infty,$ or
	\item[(d)] $n_c\le N_c,n_1\rightarrow\infty,n_2\rightarrow\infty$,
\end{enumerate}
where $M$ and $N_c$ are constants.

\section{Asymptotic Theory}
\label{sec:EstimatorandDistribution}
\subsection{Effect Size Estimator}
\label{sec:EffectSize}
To get the asymptotically unbiased and consistent estimator of the relative treatment effect $p^{(l)}$, we replace $F_1^{(l)}$ and $F_2^{(l)}$ with their empirical counterparts $\widehat{F}_1^{(l)}$ and $\widehat{F}_2^{(l)}$. For the $l^{th}$ component, let $\theta_g^{(l)}=\frac{n_c^{(l)}}{n_c^{(l)}+n_{g}^{(l)}}=\frac{n_c^{(l)}}{m_g^{(l)}}$ be the relative sample size of complete and incomplete data in the $g^{th}$ group. Notice that under the data structure shown in Table \ref{Table:schematic}, $\theta_g^{(l)}$'s are identical for all components within the same group. For simplicity, we drop the superscript $(l)$ in the subsequence section. Using relative sample size as weights, define the weighted empirical distribution function by
\begin{equation}
	\label{Fhat-gl}
	\widehat{F}_g^{(l)}(x) =  \theta_g\widehat{F}_g^{(c)(l)}(x)+(1-\theta_g)\widehat{F}_g^{(i)(l)}(x),\quad g=1,2,\quad\textrm{where}
\end{equation}
$$\widehat{F}_g^{(c)(l)}(x)=\frac{1}{n_c}\sum_{k=1}^{n_c}c(x-X_{gk}^{(c)(l)})\quad \textrm{and}\quad \widehat{F}_g^{(i)(l)}(x)=\frac{1}{n_g}\sum_{k=1}^{n_g}c(x-X_{gk}^{(i)(l)})$$
denote the empirical distribution functions of complete and incomplete cases, respectively. Here, the function $c(x)=0,\frac{1}{2},1$ according as $x<0,=0,>0$ is the normalized count function. 
A weighted estimator $\widehat{p}^{(l)}$ of $p^{(l)}$ can then be achieved by plugging in $\widehat{F}_1^{(l)}$ and $\widehat{F}_2^{(l)}$ into the integral representation in (\ref{p-l}) and the resulting estimator is
\begin{equation}
	\label{equation:phatl}
	\begin{split}
		\widehat{p}^{(l)} & = 
		\int \widehat{F}_1^{(l)}d\widehat{F}_2^{(l)} \\
		& = \frac{1}{m_1m_2}[n_c^2\int\widehat{F}_{1}^{(c)(l)}d\widehat{F}_{2}^{(c)(l)}+n_cn_2\int\widehat{F}_{1}^{(c)(l)}d\widehat{F}_{2}^{(i)(l)}\\
		&\qquad\quad+n_cn_1\int\widehat{F}_{1}^{(i)(l)}d\widehat{F}_{2}^{(c)(l)}+n_1n_2\int\widehat{F}_{1}^{(i)(l)}d\widehat{F}_{2}^{(i)(l)}].\\
	\end{split}
\end{equation}
One can verify that
\begin{equation}
	\label{equation:phatl_rank}
	\widehat{p}^{(l)}=\frac{1}{N}(\theta_2\overline{R}_{2\cdot}^{(c)(l)}-\theta_1\overline{R}_{1\cdot}^{(c)(l)}+(1-\theta_2)\overline{R}_{2\cdot}^{(i)(l)}-(1-\theta_1)\overline{R}_{1\cdot}^{(i)(l)})+\frac{1}{2},
\end{equation} 
where $\overline{R}_{g\cdot}^{(c)(l)}=\frac{1}{n_c}\sum_{k=1}^{n_c}R_{gk}^{(c)(l)}$ and $\overline{R}_{g\cdot}^{(i)(l)}=\frac{1}{n_g}\sum_{k=1}^{n_g}R_{gk}^{(i)(l)}$ are the means of ranks $R_{gk}^{(c)(l)}$ and $R_{gk}^{(i)(l)}$ of $X_{gk}^{(c)(l)}$ and $X_{gk}^{(i)(l)}$ among all $N=2n_c+n_1+n_2$ observations on the $l^{th}$ component among both groups. 
\begin{prop}
	\label{prop-phat}
	Let $\widehat{\bm{p}}=(\widehat{p}^{(1)},\cdots,\widehat{p}^{(d)})^\top$ be as defined in (\ref{equation:phatl_rank}). 
	Then, under Assumption \ref{assumption-3.1},
	\begin{description}
		\item[(1)] $E(\widehat{\bm{p}})\to\bm{p}$ and
		\item[(2)] $\widehat{\bm{p}}\overset{a.s.}{\to}\bm{p}$. 
	\end{description}                                                                                                                            
\end{prop}
\begin{proof}
	Since $d$ is finite, it suffices to show asymptotic unbiasedness and strong consistency of the relative effect size on each component, i.e. $E(\widehat{p}^{(l)})\to p^{(l)}$ and $\widehat{p}^{(l)}\overset{a.s.}{\to}p^{(l)}$. Proof of these statements are provided in Proposition 4.1 of Konietschke et al. \cite{matched-pair-2012} and is therefore omitted.
\end{proof}                                       
\subsection{Asymptotic Distribution}                             
\label{sec:AsymnptoticDist}
It is shown in Theorem \ref{thm:equivalence} that $\sqrt{n}(\widehat{\bm{p}}-\bm{p})$ is asymptotically equivalent to a random vector, whose components are sums of independent random variables. Define $Y_{gk}^{(A)(l)}=F_s^{(l)}(X_{gk}^{(A)(l)})$, $g,s\in\{1,2\}$, $g\ne s$, $A\in\{c,i\}$ and denote $Z_k^{(c)(l)}=\theta_2Y_{2k}^{(c)(l)}-\theta_1Y_{1k}^{(c)(l)}$. Then define
\begin{equation}
	\label{U_N}
	\sqrt{n}U^{(l)}=\sqrt{n}(\frac{1}{n_c}\sum_{k=1}^{n_c}Z_k^{(c)(l)}+\frac{1}{m_2}\sum_{k=1}^{n_2}Y_{2k}^{(i)(l)}-\frac{1}{m_1}\sum_{k=1}^{n_1}Y_{1k}^{(i)(l)})+\sqrt{n}(1-2p^{(l)}).
\end{equation}

\begin{theorem}
	\label{thm:equivalence}                             
	Let $\sqrt{n}\bm{U}=\sqrt{n}(U^{(1)},\cdots,U^{(d)})^\top$. Then under the Assumptions \ref{assumption-3.1}, 
	\begin{equation*}
		\parallel\sqrt{n}(\widehat{\bm{p}}-\bm{p})-\sqrt{n}\bm{U}\parallel_2^2\to0,
	\end{equation*}
	where $\parallel\cdot\parallel_2$ is $L_2$-norm of vector, i.e. for a $q$-dimensional vector, $\parallel \bm{x}\parallel_2=(\sum_{i=1}^{q}x_i^2)^{\frac{1}{2}}$. \end{theorem}                                                                                                                                    
\begin{proof}
	This statement can be proved by showing $\textrm{E}[\sqrt{n}(\widehat{p}^{(l)}-p^{(l)})-\sqrt{n}U^{(l)}]^2\to0$ for all $l=1,\cdots,d$. The proof is similar to Theorem 1 in Konietschke and Brunner \cite{konietschke2009nonparametric} and is therefore omitted.
\end{proof}
Based on Theorem \ref{thm:equivalence}, it suffices to consider the distribution of vector $\sqrt{n}\bm{U}$. Let $\lambda_1,\cdots,\lambda_d$ denote the eigenvalues of $\bm{V}=Cov(\sqrt{n}\bm{U})$ and let $\underset{1\le l\le d}{\min}\{\lambda_l\}$ be the smallest eigenvalue. To derive the asymptotic distribution of $\sqrt{n}\bm{U}$, we need the regularity condition below:
\begin{assumption}
	\label{assumption-lambda}
	Let $\underset{1\le l\le d}{\min}\{\lambda_l\}\ge\lambda_0>0$, where $\lambda_0$ is some constant.
\end{assumption}

\begin{theorem}
	\label{thm:normality}
	Under Assumptions \ref{assumption-3.1} and \ref{assumption-lambda}, the statistic $\sqrt{n}(\widehat{\bm{p}}-\bm{p})$ has asymptotically a multivariate normal distribution with mean $\bm{0}$ and covariance matrix $\bm{V}$.
\end{theorem}
\begin{proof}                  
	The asymptotic normality of $\sqrt{n}(\widehat{\bm{p}}-\bm{p})$ can be established from the asymptotic distribution of random vector $\sqrt{n}\bm{U}$. Apart from some constants, $\sqrt{n}U^{(l)}$ is the sum of three independent random variables. Since the random variables $Z_k^{(c)(l)}$ and ${Y}_{gk}^{(i)(l)}$, $g\in\{1,2\}$ are uniformly bounded by Assumption \ref{assumption-3.1}, asymptotic normality of $\sqrt{n}U^{(l)}$ can be completed by verifying the Lindeberg's condition. 
	Furthermore, the joint normality is verified by the Cramér–Wold device.
\end{proof}                                                                        

\subsection{Estimation of the Asymptotic Covariance Matrix}
\label{sec:EstofCov}
To apply the results in Theorem \ref{thm:normality}, a $L_2$-\textrm{consistent} estimator of $\bm{V}$ is derived in this section. Denote by $v^{(l,r)}$ as the $(l,r)^{th}$ entry of $\bm{V}$. Then, by independence of $Z_{k}^{(c)(l)}$, $Y_{1k}^{(i)(l)}$ and $Y_{2k}^{(i)(l)}$ and also independence between subjects, 
\begin{align}
	\label{v_n}
	v^{(l,r)}&=n\textrm{Cov}(U^{(l)},U^{(r)})\displaybreak[0]\nonumber\\
	&=n\big[\textrm{Cov}(\frac{1}{m_1}\sum_{k=1}^{n_1}Y_{1k}^{(i)(l)},\frac{1}{m_1}\sum_{k=1}^{n_1}Y_{1k}^{(i)(r)})
	+\textrm{Cov}(\frac{1}{m_2}\sum_{k=1}^{n_2}Y_{2k}^{(i)(l)},\frac{1}{m_2}\sum_{k=1}^{n_2,}Y_{2k}^{(i)(r)})+\nonumber\\
	&\qquad \textrm{Cov}(\frac{1}{n_c}\sum_{k=1}^{n_c}Z_k^{(c)(l)},\frac{1}{n_c}\sum_{k=1}^{n_c}Z_k^{(c)(r)})\big]\displaybreak[0]\nonumber\\
	&=n\big[\frac{n_1}{m_1^2}\textrm{Cov}(Y_{11}^{(i)(l)},Y_{11}^{(i)(r)})+\frac{n_2}{m_2^2}\textrm{Cov}(Y_{21}^{(i)(l)},Y_{21}^{(i)(r)})+
	\frac{1}{n_c}\textrm{Cov}(Z_1^{(c)(l)},Z_1^{(c)(r)})\big]\displaybreak[0]\nonumber\\
	&=n(\frac{n_1}{m_1^2}c_1^{(l,r)}+\frac{n_2}{m_2^2}c_2^{(l,r)}+\frac{1}{n_c}c_3^{(l,r)}),
\end{align}
where $c_g^{(l,r)}=\textrm{Cov}(Y_{g1}^{(i)(l)},Y_{g1}^{(i)(r)})$ for $g=1,2$ and $c_3^{(l,r)}=\textrm{Cov}(Z_1^{(c)(l)},Z_1^{(c)(r)})$. Suppose the random variables $Y_{gk}^{(A)(l)}$ were observable, then natural estimators of $c_g^{(l,r)}$ are the empirical covariance
\begin{equation*}
	\tilde{c}_{g}^{(l,r)}=\frac{1}{n_g-1}\sum_{k=1}^{n_g}(Y_{gk}^{(i)(l)}-\overline{Y}_{g\cdot}^{(i)(l)})(Y_{gk}^{(i)(r)}-\overline{Y}_{g\cdot}^{(i)(r)}),\quad g=1,2,
\end{equation*}
and
\begin{equation*}
	\tilde{c}_3^{(l,r)}=\frac{1}{n_c-1}\sum_{k=1}^{n_c}(Z_k^{(c)(l)}-\overline{Z}_\cdot^{(c)(l)})(Z_k^{(c)(r)}-\overline{Z}_\cdot^{(c)(r)}),
\end{equation*}
where $\overline{Y}_{g\cdot}^{(c)(l)}=\frac{1}{n_c}\sum_{k=1}^{n_c}Y_{gk}^{(c)(l)}$, $\overline{Y}_{g\cdot}^{(i)(l)}=\frac{1}{n_g}\sum_{k=1}^{n_g}Y_{gk}^{(i)(l)}$ and $\overline{Z}_{\cdot}^{(c)(l)}=\frac{1}{n_c}\sum_{k=1}^{n_c}Z_{k}^{(c)(l)}$. Consequently an estimator of $v^{(l,r)}$ is given by 
\begin{equation}                                                                             
	\label{tildevN}
	\widetilde{v}^{(l,r)}=n(\frac{n_1}{m_1^2}\tilde{c}_1^{(l,r)}+\frac{n_2}{m_2^2}\tilde{c}_2^{(l,r)}+\frac{1}{n_c}\tilde{c}_3^{(l,r)}).                                    
\end{equation}
Now, let $\bm{Y}_{gk}^{(A)}=(Y_{gk}^{(A)(1)},\cdots,Y_{gk}^{(A)(d)})^\top$, $\overline{\bm{Y}}_{g\cdot}^{(A)}=(\overline{Y}_{g\cdot}^{(A)(1)},\cdots,\overline{Y}_{g\cdot}^{(A)(d)})^\top$, $\bm{Z}_k^{(c)}=(Z_{k}^{(c)(1)}\cdots,Z_{k}^{(c)(d)})^\top$ and $\overline{\bm{Z}}_\cdot^{(c)}=(\overline{Z}_\cdot^{(c)(1)},\cdots,\overline{Z}_\cdot^{(c)(d)})^\top$. Then a natural estimator of $\bm{V}$ is given by $\widetilde{\bm{V}}=\widetilde{\bm{V}}_{c}+\widetilde{\bm{V}}_{1}+\widetilde{\bm{V}}_{2}$, where
\begin{equation*}
	\widetilde{\bm{V}}_{c}=\frac{n}{n_c(n_c-1)}\sum_{k=1}^{n_c}(\bm{Z}_k^{(c)}-\overline{\bm{Z}}_\cdot^{(c)})(\bm{Z}_k^{(c)}-\overline{\bm{Z}}_\cdot^{(c)})^\top\quad\textrm{and}
\end{equation*} 
\begin{equation*}
	\widetilde{\bm{V}}_{g}=n\frac{n_g}{m_g^2(n_g-1)}\sum_{k=1}^{n_g}(\bm{Y}_{gk}^{(i)}-\overline{\bm{Y}}_{g\cdot}^{(i)})(\bm{Y}_{gk}^{(i)}-\overline{\bm{Y}}_{g\cdot}^{(i)})^\top,\quad g=1,2.
\end{equation*}
However, the random variables $Y_{gk}^{(A)(l)}$ are not observable and must be replaced by observable ones, which are "close enough" in probability. One can verify that $\widehat{F}_{g}^{(l)}$ defined in (\ref{Fhat-gl}) can be represented as 
\begin{equation*} 
	\widehat{F}_{g}^{(l)}=\frac{1}{n_c+n_g}\big[\sum_{k=1}^{n_c}c(x-X_{gk}^{(c)(l)})+\sum_{k=1}^{n_g}c(x-X_{gk}^{(i)(l)})\big],\quad g=1,2.
\end{equation*}                                                          
Hence, the empirical counterparts 
$\widehat{Y}_{gk}^{(A)(l)}=\widehat{F}_s^{(l)}{(X_{gk}^{(A)(l)})}$, $s\ne g\in\{1,2\}$, $A\in\{c,i\}$ can be expressed in terms of the overall ranks $R_{gk}^{(A)(l)}$ from $X_{gk}^{(A)(l)}$ among all $N$ observations on the $l^{th}$ component, and the internal ranks $R_{gk}^{(A,g)(l)}$ of $X_{gk}^{(A)(l)}$ among all $n_c+n_g$ observations on the $l^{th}$ component in the $g^{th}$ treatment group. Specifically,
\begin{equation*}
	\widehat{Y}_{gk}^{(A)(l)}=\frac{1}{n_c+n_s}(R_{gk}^{(A)(l)}-R_{gk}^{(A,g)(l)}), g\ne s\in\{1,2\},A\in\{c,i\}\quad\textrm{and}\quad l=1,\cdots,d.
\end{equation*}
Further, let $\bm{R}_{gk}^{(A)}=(R_{gk}^{(A)(1)},\cdots,R_{gk}^{(A)(d)})^\top$ and $\bm{R}_{gk}^{(A,g)}=(R_{gk}^{(A,g)(1)},\cdots,R_{gk}^{(A,g)(d)})^\top$ denote the vector of the ranks $R_{gk}^{(A)(l)}$ and $R_{gk}^{(A,g)(l)}$, respectively, for $g=1,2$ and $l=1,\cdots,d$. Also, define 
\begin{equation*}
	\overline{\bm{R}}_{g\cdot}^{(c)}=\frac{1}{n_c}\sum_{k=1}^{n_c}\bm{R}_{gk}^{(c)},\quad\overline{\bm{R}}_{g\cdot}^{(i)}=\frac{1}{n_g}\sum_{k=1}^{n_g}\bm{R}_{gk}^{(i)},
\end{equation*}
\begin{equation*}
	\overline{\bm{R}}_{g\cdot}^{(c,g)}=\frac{1}{n_c}\sum_{k=1}^{n_c}\bm{R}_{gk}^{(c,g)}\text{ and }\quad\overline{\bm{R}}_{g\cdot}^{(i,g)}=\frac{1}{n_g}\sum_{k=1}^{n_g}\bm{R}_{gk}^{(i,g)}=\frac{n_g+1}{2}\mathbbm{1}_d,\quad g=1,2.
\end{equation*}
Let $\bm{B}_{gk}^{(A)}=\bm{R}_{gk}^{(A)}-\bm{R}_{gk}^{(A,g)}$, $\overline{\bm{B}}_{g\cdot}^{(A)}=\overline{\bm{R}}_{g\cdot}^{(A)}-\overline{\bm{R}}_{g\cdot}^{(A,g)}$for $A\in\{c,i\}$. As proposed before, an estimator of $\bm{V}$ is given by $\widehat{\bm{V}}=\widehat{\bm{V}}_{c}+\widehat{\bm{V}}_{1}+\widehat{\bm{V}}_{2}$, where
\begin{equation}
	\label{V-12}
	\widehat{\bm{V}}_{g}=n\frac{n_g}{m_1^2m_2^2(n_g-1)}\sum_{k=1}^{n_g}(\bm{B}_{gk}^{(i)}-\overline{\bm{B}}_{g\cdot}^{(i)})(\bm{B}_{gk}^{(i)}-\overline{\bm{B}}_{g\cdot}^{(i)})^\top, g=1,2\text{ and }
\end{equation}

\begin{equation}
	\label{V-3}
	\widehat{\bm{V}}_{c}=n\frac{n_c}{m_1^2m_2^2(n_c-1)}\sum_{k=1}^{n_c}\big[\bm{B}_{2k}^{(c)}-\overline{\bm{B}}_{2\cdot}^{(c)}-(\bm{B}_{1k}^{(c)}-\overline{\bm{B}}_{1\cdot}^{(c)})\big]\big[\bm{B}_{2k}^{(c)}-\overline{\bm{B}}_{2\cdot}^{(c)}-(\bm{B}_{1k}^{(c)}-\overline{\bm{B}}_{1\cdot}^{(c)})\big]^\top.
\end{equation}                                   
The consistency of the estimator $\widehat{\bm{V}}$ is established in Theorem \ref{thm:consistency}.
\begin{theorem}
	\label{thm:consistency}
	Under Assumption \ref{assumption-3.1} and \ref{assumption-lambda},
	\[\parallel\widehat{\bm{V}}-\bm{V}\parallel_2^2\rightarrow0,\]
	where $\parallel\cdot\parallel_2$ is $L_2$-norm of matrix, i.e.  for any $s\times t$ matrix $A$, $\parallel A\parallel_2=(\sum_{i=1}^{s}\sum_{j=1}^{t}|a_{ij}|^2)^{1/2}$.
\end{theorem}
\begin{proof}                                                                               
	In Appendix.
\end{proof} 
\section{Test Statistics}
\label{sec:TestStat}
In this section, we will introduce two test statistics for $H_0:\widehat{\bm{p}}=\frac{1}{2}\mathbbm{1}_d$. The first one is Wald-type statistic, which is commonly used for nonparametric model in the multivariate structure. The second one is ANOVA-type statistic, which is first introduced in Brunner and Dette \cite{brunner1997box} for univariate factorial designs and further extended to the multivariate structure in Brunner et al.  \cite{brunner2002multivariate}.
\label{sec:StatandCI}     
\subsection{Wald-type Statistic}
From Theorem \ref{thm:normality}, it follows under $H_0:\bm{p}=\frac{1}{2}\mathbbm{1}_d$ that the statistic $\sqrt{n}\bm{V}^{-1/2}(\widehat{\bm{p}}-\frac{1}{2}\mathbbm{1}_d)$ has asymptotic multivariate standard normal distribution $\bm{\textbf{N}}_d(\bm{0},\bm{I}_d)$. By the Continuous Mapping Theorem, distribution of the quadratic form 
\begin{equation}
	Q_n^*=n\cdot(\widehat{\bm{p}}-\frac{1}{2}\mathbbm{1}_d)^\top\bm{V}^{-1}(\widehat{\bm{p}}-\frac{1}{2}\mathbbm{1}_d)
\end{equation} 
tends to a central $\chi^2$-distribution with $d$ degrees of freedom. However, the covariance matrix $\bm{V}$ is unknown and must be replaced by a consistent estimator. From Theorem \ref{thm:consistency}, it follows that $\parallel\widehat{\bm{V}}^{-1}-\bm{V}^{-1}\parallel_2^2\rightarrow0$ under Assumption \ref{assumption-lambda}. Therefore, under $H_0$, the quadratic form 
\begin{equation}
	\label{Q_N}
	Q_n=n\cdot(\widehat{\bm{p}}-\frac{1}{2}\mathbbm{1}_d)^\top\widehat{\bm{V}}^{-1}(\widehat{\bm{p}}-\frac{1}{2}\mathbbm{1}_d)
\end{equation}
has the same asymptotic distribution as $Q_n^*$, which is also a $\chi^2$-distribution with $d$ degrees of freedom. This can be verified by  
\begin{equation*}
	\parallel Q_n^*-Q_n\parallel_2^{2}\le n\cdot\parallel\widehat{\bm{p}}-\frac{1}{2}\mathbbm{1}_d\parallel_2^2\cdot\parallel\bm{V}^{-1}-\widehat{\bm{V}}^{-1}\parallel_2^2\rightarrow0,
\end{equation*}                                                                       
which follows by Proposition \ref{prop-phat} (2).  Therefore, the null hypothesis $H_0:\bm{p}=\frac{1}{2}\mathbbm{1}_d$ will be rejected at significance level $\alpha$ if $Q_n\ge\chi^2_\alpha$.

The quadratic form $Q_n$ is called the nonparametric Wald-type statistic. It is well-known that this statistic has slow convergence to the limiting $\chi^2$-distribution. Thus $Q_N$ may become extremely liberal unless a very large sample size is available. Hence, Wald-type statistics should only be applied in the case of large sample sizes.

\subsection{ANOVA-type Statistic}
In practice, data sets may not have large enough sample size and the Wald-type statistic may not provide accurate results. Therefore, it is necessary to develop a test with small sample approximation which maintains the pre-assigned level with satisfactory accuracy. The idea for developing such a procedure is to replace the inverse $\widehat{\bm{V}}^{-1}$ of the estimated covariance matrix in (\ref{Q_N}) with $1/tr(\widehat{\bm{V}})$ much in the same way the ANOVA tests do and consider the asymptotic distribution of the resulting quadratic form.  
Details of the approximation process can be found in Brunner et al. \cite{brunner2002multivariate} and is summarized below.
\begin{app}                                                                                                                                  
	Under Assumption \ref{assumption-3.1} and Assumption \ref{assumption-lambda}, the statistic
	\begin{equation}
		\label{ANOVA}
		F_n=\frac{n}{tr(\widehat{\bm{V}})}(\widehat{\bm{p}}-\frac{1}{2}\mathbbm{1}_d)^\top(\widehat{\bm{p}}-\frac{1}{2}\mathbbm{1}_d)=\frac{n}{tr(\widehat{\bm{V}})}\cdot\sum_{l=1}^{d}(\widehat{p}^{(l)}-\frac{1}{2})^2
	\end{equation}
	has, approximately, a central $F(\hat{\nu},\infty)$ distribution under $H_0:\bm{p}=\frac{1}{2}\mathbbm{1}_d$, where $\hat{\nu}=\frac{[tr(\widehat{\bm{V}})]^2}{tr(\widehat{\bm{V}}^2)}$.
\end{app}
Test statistic $F_n$ is called ANOVA-type statistic. For small sample sizes, the null hypothesis $H_0:\bm{p}=\frac{1}{2}\mathbbm{1}_d$ will be rejected at significance level $\alpha$ if $F_n\ge F_\alpha(\hat{\nu},\infty)$.

\section{Extension to General Missing Pattern}
\label{sec:flex-missing}                                                                 
In Sections \ref{sec:ModelandHypothesis}-\ref{sec:TestStat}, we assumed that incomplete data occur at treatment level as, for example, a subject not showing up for a scheduled treatment. 
That means this subject will have no observation on all variables for that treatment group. However, in practical applications, missing data may occur for any component (variable) within a treatment group. In this situation, the total number of possible missing patterns is $2^{2d}-1$, not counting completely missing cases. Also note that the order of the missing patterns is completely arbitrary and is irrelevant to the methodology to be developed. The central idea of our approach is to combine estimates from cases with different missing patterns by weighing them appropriately. Let the $2d\times n$ data matrix be denoted as $\bm{X}=(\bm{X}_1,\cdots,\bm{X}_n)$, where the first $d$ rows are measurements on the $d$ variables in the first group and the remaining $d$ rows are measurements on the $d$ variables in the second group. The columns represent subjects. Suppose there are $K$ missing patterns  in the data set and let $\bm{X}$ be partitioned as $(\bm{X}^{(1)},\cdots,\bm{X}^{(K)})$ where $K\le2^{2d}-1$ and $\bm{X}^{(k)}$ is a $2d\times n_k$ matrix of observations for the $k^{th}$ missing pattern. That is, $\bm{X}^{(k)}$ contains data from $n_k$ subjects for whom there are observations on the same $p_k$ components, say $(j_1,\cdots,j_{p_k})$, among all $2d$ components, and no observations for the remaining $2d-p_k$ components. For the sake of brevity, we assume that all possible patterns are represented in the data matrix, i.e. $K=1,2,\cdots,2^{2d}-1$. To further display missing patterns by means of component allocations, we define a $2d\times2d$ matrix $B_k$ such that for $s=1,\cdots,2d$, the $s^{th}$ row of $B_k$ has unity at the $j_s^{th}$ position and zero elsewhere. With the assumption that at least one subject has complete data on all components in both groups (i.e. $n_1\ge1$), $B_1=I_{2d}$. 

Based on the notations defined above, complete and incomplete sample sizes on the $l^{th}$ component can be represented as 
$$n_c^{(l)}=\sum_{k=1}^{K}n_kI_{\{B_{k,l}=B_{k,l+d}=1\}},\quad n_{1}^{(l)}=\sum_{k=1}^{K}n_kI_{\{B_{k,l}=1,B_{k,l+d}=0\}}$$
$$\textrm{and}\quad n_{2}^{(l)}=\sum_{k=1}^{K}n_kI_{\{B_{k,l}=0,B_{k,l+d}=1\}},$$
where $B_{k,l}$ refer to the $l^{th}$ diagonal entry of matrix $B_k$. In addition, denote $m_g^{(l)}=n_c^{(l)}+n_{g}^{(l)}$ and, $n=\sum_{k=1}^{K}n_k$. Similar to Assumption \ref{assumption-3.1}, we need an assumption on the sample size allocations.
\begin{assumption}
	\label{assumption-flexible}
	$\underset{g,l}{\min}\{n_c^{(l)}+n_{g}^{(l)}\}\rightarrow\infty$ such that $\underset{g,l}{\max}\{\frac{n}{n_c^{(l)}+n_{g}^{(l)}}\}\le N_0<\infty$ for $l=1,\cdots,d$ and $g=1,2$. 
\end{assumption}
Let $\bigtriangleup_{gk}^{(A)(l)}=1$ if $X_{gk}^{(A)(l)}$ is observed, and $\bigtriangleup_{gk}^{(A)(l)}=0$ otherwise. Define the index set $S^{(c,l)}=\{k: \bigtriangleup_{1k}^{(c)(l)}\cdot \bigtriangleup_{2k}^{(c)(l)}=1\}$ of all complete cases on the $l^{th}$ component. Similarly, define $S^{(g,l)}=\{k:\bigtriangleup_{gk}^{(i)(l)}=1\wedge\bigtriangleup_{g'k}^{(i)(l)}=0\}$, $g\ne g'\in\{1,2\}$ as the index set of all incomplete cases on the $l^{th}$ component in the $g^{th}$ group. In this setup, we define 
\[\widehat{F}_g^{(c)(l)}(x)=\frac{1}{n_c^{(l)}}\sum_{k=1}^{n_c^{(l)}}c(x-X_{gk}^{(c)(l)})\quad \textrm{and}\quad \widehat{F}_g^{(i)(l)}(x)=\frac{1}{n_g^{(l)}}\sum_{k=1}^{n_g^{(l)}}c(x-X_{gk}^{(i)(l)})\]
and plug them in equation (\ref{Fhat-gl}) to define the empirical distribution functions $\widehat{F}_g^{(l)}(x),g=1,2$. This estimate will be used in equation (\ref{equation:phatl}) to define an estimator $\widehat{\bm{p}}$ of the vector of relative treatment effects $\bm{p}$. Furthermore, write
\begin{equation*} 
	\begin{split}                 
		\sqrt{n}U^{(l)}&=\sqrt{n}\big[\frac{1}{m_2^{(l)}}\sum_{k\in S^{(c,l)}}F_1^{(l)}(X_{2k}^{(c)(l)})-\frac{1}{m_1^{(l)}}\sum_{k\in S^{(c,l)}}F_2^{(l)}(X_{1k}^{(c)(l)})+\frac{1}{m_{2}^{(l)}}\sum_{k\in S^{(2,l)}}F_1^{(l)}(X_{2k}^{(i)(l)})\\
		&\quad-\frac{1}{m_{1}^{(l)}}\sum_{k\in S^{(1,l)}}F_2^{(l)}(X_{1k}^{(i)(l)})\big]+\sqrt{n}(1-2p^{(l)})\\
		&=\sqrt{n}(\frac{1}{n_c^{(l)}}\sum_{k\in S^{(c,l)}}Z_k^{(c)(l)}+\frac{1}{m_2^{(l)}}\sum_{k\in S^{(2,l)}}Y_{2k}^{(i)(l)}-\frac{1}{m_1^{(l)}}\sum_{k\in S^{(1,l)}}Y_{1k}^{(i)(l)})+\sqrt{n}(1-2p^{(l)}),
	\end{split}
\end{equation*}
where $Z_k^{(c)(l)}=\theta_2^{(l)}Y_{2k}^{(c)(l)}-\theta_1^{(l)}Y_{1k}^{(c)(l)}$. Let  $\sqrt{n}\bm{U}=\sqrt{n}(U^{(1)},\cdots,U^{(d)})^\top$ and $\bm{V}=Cov(\sqrt{n}\bm{U})$. At this point, we can prove that Theorem \ref{thm:equivalence} still holds by verifying $\textrm{E}[\sqrt{n}(\widehat{p}^{(l)}-p^{(l)})-\sqrt{n}U^{(l)}]^2\to0$ for all $l=1,\cdots,d$. Therefore, under Assumptions \ref{assumption-lambda} and \ref{assumption-flexible}, $\sqrt{n}(\widehat{\bm{p}}-\bm{p})$ has asymptotically a multivariate normal distribution with mean $\bm{0}$ and covariance $\bm{\bm{V}}$, as proved in Theorem \ref{thm:normality}. The $(l,r)^{th}$ entry of the covariance matrix $\bm{V}$ can be computed as 
\begin{align}
	\label{vir_flex}
	v^{(l,r)}&=n\textrm{Cov}(U^{(l)},U^{(r)})\nonumber\\
	&=n\textrm{Cov}(\frac{1}{n_c^{(l)}}\sum_{k\in S^{(c,l)}}Z_k^{(c)(l)}+\frac{1}{m_2^{(l)}}\sum_{k\in S^{(2,l)}}Y_{2k}^{(i)(l)}-\frac{1}{m_1^{(l)}}\sum_{k\in S^{(1,l)}}Y_{1k}^{(i)(l)},\displaybreak[0]\nonumber\\
	&\qquad\qquad\frac{1}{n_c^{(r)}}\sum_{k\in S^{(c,r)}}Z_k^{(c)(r)}+\frac{1}{m_2^{(r)}}\sum_{k\in S^{(2,r)}}Y_{2k}^{(i)(r)}-\frac{1}{m_1^{(r)}}\sum_{k\in S^{(1,r)}}Y_{1k}^{(i)(r)})\displaybreak[0]\nonumber\\
	&=n\big[\underset{C_1}{\frac{1}{n_c^{(l)}n_c^{(r)}}\underbrace{\textrm{Cov}(\sum_{k\in S^{(c,l)}}Z_k^{(c)(l)},\sum_{k\in S^{(c,r)}}Z_k^{(c)(r)})}}
	+\underset{C_2}{\frac{1}{n_c^{(l)}m_2^{(r)}}\underbrace{\textrm{Cov}(\sum_{k\in S^{(c,l)}}Z_k^{(c)(l)},\sum_{k\in S^{(2,r)}}Y_{2k}^{(i)(r)})}}\displaybreak[0]\nonumber\\
	&\qquad -\underset{C_3}{\frac{1}{n_c^{(l)}m_1^{(r)}}\underbrace{\textrm{Cov}(\sum_{k\in S^{(c,l)}}Z_k^{(c)(l)},\sum_{k\in S^{(1,r)}}Y_{1k}^{(i)(r)})}}+\underset{C_4}{\frac{1}{m_2^{(l)}n_c^{(r)}}\underbrace{\textrm{Cov}(\sum_{k\in S^{(2,l)}}Y_{2k}^{(i)(l)},\sum_{k\in S^{(c,r)}}Z_k^{(c)(r)})}}\displaybreak[0]\nonumber\\
	&\qquad+\underset{C_5}{\frac{1}{m_2^{(l)}m_2^{(r)}}\underbrace{\textrm{Cov}(\sum_{k\in S^{(2,l)}}Y_{2k}^{(i)(l)},\sum_{k\in S^{(2,r)}}Y_{2k}^{(i)(r)})}}-\underset{C_6}{\frac{1}{m_2^{(l)}m_1^{(r)}}\underbrace{\textrm{Cov}(\sum_{k\in S^{(2,l)}}Y_{2k}^{(i)(l)},\sum_{k\in S^{(1,r)}}Y_{1k}^{(i)(r)})}}\displaybreak[0]\nonumber\\
	&\qquad-\underset{C_7}{\frac{1}{m_1^{(l)}n_c^{(r)}}\underbrace{\textrm{Cov}(\sum_{k\in S^{(1,l)}}Y_{1k}^{(i)(l)},\sum_{k\in S^{(c,r)}}Z_k^{(c)(r)})}}-\underset{C_8}{\frac{1}{m_1^{(l)}m_2^{(r)}}\underbrace{\textrm{Cov}(\sum_{k\in S^{(1,l)}}Y_{1k}^{(i)(l)},\sum_{k\in S^{(2,r)}}Y_{2k}^{(i)(r)})}}\displaybreak[0]\nonumber\\
	&\qquad+\underset{C_9}{\frac{1}{m_1^{(l)}m_1^{(r)}}\underbrace{\textrm{Cov}(\sum_{k\in S^{(1,l)}}Y_{1k}^{(i)(l)},\sum_{k\in S^{(1,r)}}Y_{1k}^{(i)(r)})}}\big].\displaybreak[0]
\end{align}

A consistent estimator $\widehat{\bm{V}}$ of $\bm{V}$ is achieved by separately estimating terms involving $C_1$-$C_9$. Details are shown in the Appendix. Test procedures described in Section \ref{sec:TestStat} can be directly applied by using the newly constructed covariance matrix estimator $\widehat{\bm{V}}$ in (\ref{Q_N}) and (\ref{ANOVA}). 

\section{Simulation Study}
\label{sec:SimuRes}
In this section, we examine the performance of the proposed Wald-type statistic $Q_n$ in (\ref{Q_N}) and that of the ANOVA-type statistic $\widehat{F}_n$ in (\ref{ANOVA}). The evaluations focus on (a) control of the preassigned Type-I error level ($\alpha=0.05$) under  $H_0:\bm{p}=\frac{1}{2}\mathbbm{1}_d$ and (b) achieved powers to detect specific alternatives. Also, a simulation study will be conducted to investigate the accuracy of these test procedures for datasets with general missing patterns. 
\subsection{Simulation Settings}
\label{Sec:SimulationSetting}
The simulation seeks to generate evidence on the performance of the tests along the two criteria in various scenarios that cover a wide-spectrum of reasonable models. The study involves multivariate data with strong/weak and positive/negative correlations for small and moderate sample sizes. 
The observations $\bm{X}_j^{(c)}$ and $\bm{X}_{gk}^{(i)}$, $j=1,\cdots,n_c$, $g=1,2$, $k=1,\cdots,n_g$ will be generated from Discretized Multivariate Normal, Multivariate Log-Normal and Multivariate Cauchy distributions, which represent discrete, skewed and heavily tailed data, respectively. Covariance or scale matrix of the data will be set to
\[\bm{\Sigma}=
\begin{bmatrix}
	\sigma_1^2\textrm{I}_d+\rho_1\sigma_1^2(\textrm{J}_d-\textrm{I}_d) & \rho_{12}\sigma_1\sigma_2\textrm{J}_d&\\
	&&\\
	\rho_{12}\sigma_1\sigma_2\textrm{J}_d& \sigma_2^2\textrm{I}_d+\rho_2\sigma_2^2(\textrm{J}_d-\textrm{I}_d) &\\
\end{bmatrix},\]
where $J_d$ is $d$ dimensional square matrix of all ones. The impact of between and within treatment group correlations can be investigated by varying values of $\rho_1$, $\rho_2$ and $\rho_{12}$. Further, homoscedastic and heteroscedastic scenarios are covered by setting $\sigma_1^2=\sigma_2^2$ and $\sigma_1^2\ne\sigma_2^2$, respectively. For the correlations and variances, we investigate for $(\rho_1,\rho_2,\rho_{12})\in\{(0.1,0.1,0.1),(-0.1,-0.1,-0.1),(0.1,0.9,0.5)\}$ and $(\sigma_1^2,\sigma_2^2)\in\{(1,1),(1,5)\}$. With this covariance matrix $\bm{\Sigma}$, we consider three multivariate distributions:
\begin{itemize}
	\item Discretized Multivariate Normal: Data for the complete as well as incomplete cases are generated from multivariate normal distributions and then each component is rounded to the nearest integer. More precisely, defining $[\cdot]$ as the rounding operator, data for the complete cases are generated as  $\textbf{X}_j^{(c)}=(X_{1j}^{(c)(1)},\cdots,X_{1j}^{(c)(d)},X_{2j}^{(c)(1)},\cdots,X_{2j}^{(c)(d)})$, where $X_{gj}^{(c)(d)}=[W_{gk}^{(c)(d)}]$ and $\bm{W}_j^{(c)}=(W_{1j}^{(c)(1)},\cdots,W_{1j}^{(c)(d)},W_{2j}^{(c)(1)},\cdots,W_{2j}^{(c)(d)})\sim N(\bm{0},\bm{\Sigma})$. Data for incomplete cases are generated in the same manner.
	\item Multivariate Log-Normal: Data for the complete as well as incomplete cases are generated from multivariate normal distributions and then each component is exponentiated. Specifically, data for complete cases are generated as $\textbf{X}_j^{(c)}$=$(X_{1j}^{(c)(1)},\cdots,X_{1j}^{(c)(d)},X_{2j}^{(c)(1)},\cdots,X_{2j}^{(c)(d)})$, where $X_{gj}^{(c)(d)}={\rm exp}(W_{gj}^{(c)(d)})$ and $\bm{W}_{j}^{(c)}$ $\sim \bm{N}(\bm{0},\bm{\Sigma})$. Data for incomplete cases are generated similarly.
	\item Multivariate Cauchy: Both the complete and incomplete cases are generated from the multivariate Cauchy distributions
	$\bm{C}(\textbf{0},\bm{\Sigma})$.
\end{itemize}
Four combinations of sample sizes listed in Table \ref{Table:SampleComb} will be considered.
\begin{table}[!htb]
	\caption{Complete and Incomplete Sample Sizes Combinations}
	\label{Table:SampleComb}
	\centering 
	\begin{tabular}{|c|c|c|c|}
		\hline
		Setting & $n_c$ & $n_1$ & $n_2$ \\ \hline
		1 & 10 & 30 & 30 \\ \hline
		2 & 30 & 10 & 10 \\ \hline                                                                                        
		3 & 30 & 30 & 10 \\ \hline
		4 & 10 & 10 & 30 \\ \hline                                                                                                                                                                                                                                                     
	\end{tabular}
\end{table}
Dimensions of the multivariate data will be set to $d = 2 ,3, 5$. For every combination of the sample size, dimension, covariance matrices and distributions, 1000 simulations are performed. The empirical sizes or powers are calculated from these replications. All the computations are done in R (version 3.6.0) \cite{RStudio}.

Two different methods for handling missing values in nonparametric multivariate analysis are considered for comparison: 
\begin{itemize}
	\item the method of Brunner et al. \cite{brunner2002multivariate} which is designed for independent multivariate samples. We keep data from incomplete cases only for each treatment group so that the problem reduces to that of two independent samples. This method is a special case of the methods developed in this paper by setting $n_c=0$.
	\item A special case of the methods derived in this manuscript where $n_1=n_2=0$. In this case, only data from complete cases will be used. 
\end{itemize} 

The goal of the simulation is to obtain information on whether the test procedures proposed in this manuscript, which use all available data, have superior performance over the tests that use only partial data.
For brevity of notations, the alternative methods will be referred to as Incomplete and Complete. In the simulation tables, test procedures for Incomplete and Complete are denoted by $Q_n^{(j)}$ and $F_n^{(j)}$, where $j=1,2$ represents Wald-type and ANOVA-type statistics respectively.

\subsection{Type-I Error Rate}
\label{sec:TypeIErrorRate}
Achieved Type-I error rates are presented in Tables \ref{Table:Typeone-normal}-\ref{Table:Typeone-cauchy}. From these tables, we note that dependence structures as well as distributions have minor effects on Type-I error rates. The achieved Type-I error rates for each method are close to each other no matter data are generated from discrete, skewed or heavily-tailed distributions. We can also see that the ANOVA-type statistics $F_n$, $F_n^{(1)}$ and $F_n^{(2)}$ are quite stable for all settings and $F_n$ has an advantage over $F_n^{(1)}$ and $F_n^{(2)}$ in preserving the preassigned significance level of $\alpha=0.05$ in most of the cases. However, although the Wald-type statistics $Q_n$, $Q_n^{(1)}$ and $Q_n^{(2)}$ are all too liberal, $Q_n$ performs better than $Q_n^{(1)}$ and $Q_n^{(2)}$ generally. Further, performance of Incomplete and Complete methods are slightly affected by sample size settings, i.e. $Q_n^{(1)}$ and $F_n^{(1)}$ have better performance under Setting 1 where more samples are allocated to incomplete cases, while $Q_n^{(2)}$ and $F_n^{(2)}$ perform better under Setting 4 where more samples are allocated to complete cases. 

There is no obvious pattern showing that the dimension will affect the performance for the ANOVA-type statistics, but it does affect the performance for the Wald-type statistics. The achieved Type-I error rates for $Q_n$, $Q_n^{(1)}$ and $Q_n^{(2)}$ increase with $d$, and the greatest increase is observed for $Q_n^{(2)}$ at sample size Settings 1 and 4. Variances of samples also make a difference. For fixed sample sizes and correlations, data that are generated from heterogeneous distributions tend to have more liberal performance too. 

\begin{table}[!htb]
	\centering
	\caption{Achieved Type-I error rate ($\times100$) for data generated from Discretized Multivariate Normal distribution with missing pattern as in Table \ref{Table:schematic} for $d=2,3,5$. Here, $Q_n$ is the Wald-type statistic proposed in (\ref{Q_N}); $F_n$ is ANOVA-type statistic for small sample approximation proposed in (\ref{ANOVA}); $Q_n^{(1)}$ and $F_n^{(1)}$ are Wald-type and ANOVA-type tests for the Incomplete method; $Q_n^{(2)}$ and $F_n^{(2)}$ are Wald-type and ANOVA-type tests for the Complete method. The nominal Type-I error rate is $\alpha=0.05$.}
	\begin{tabular}{ccllc|c|c|c|c|c|c|c|c|c|c|c|c|}
		\cline{6-17}
		& \multicolumn{3}{c}{}                                &     & \multicolumn{6}{c|}{$(\sigma_1^2,\sigma_2^2)=(1,1)$} & \multicolumn{6}{c|}{$(\sigma_1^2,\sigma_2^2)=(1,5)$} \\ \hline
		\multicolumn{1}{|c|}{\begin{tabular}[c]{@{}c@{}}Sample\\ Size\\ Setting\end{tabular}} & \multicolumn{3}{c|}{$(\rho_1,\rho_2,\rho_{12})$}    & $d$ & $Q_n$   & $F_n$  & $Q_n^{(1)}$  & $F_n^{(1)}$  & $Q_n^{(2)}$  & $F_n^{(2)}$  & $Q_n$   & $F_n$  & $Q_n^{(1)}$  & $F_n^{(1)}$  & $Q_n^{(2)}$  & $F_n^{(2)}$  \\ \hline
		\multicolumn{1}{|c|}{\multirow{6}{*}{1}}                                              & \multicolumn{3}{c|}{\multirow{3}{*}{(0.1,0.1,0.1)}} & 2   & 6.3     & 4.7    & 6.7    & 5.9    & 11.8   & 6.5    & 6.5     & 5.7    & 7.2    & 6.3    & 13     & 7.0    \\ \cline{5-17} 
		\multicolumn{1}{|c|}{}                                                                & \multicolumn{3}{c|}{}                               & 3   & 7.1     & 5.6    & 8.1    & 5.6    & 20.6   & 5.0    & 7.9     & 5.8    & 6.4    & 4      & 19.4   & 6.1    \\ \cline{5-17} 
		\multicolumn{1}{|c|}{}                                                                & \multicolumn{3}{c|}{}                               & 5   & 8.7     & 4.6    & 9.0    & 5.6    & 40.2   & 6.1    & 10.4    & 5.6    & 10.8   & 5.1    & 40.3   & 5.4    \\ \cline{2-17} 
		\multicolumn{1}{|c|}{}                                                                & \multicolumn{3}{c|}{\multirow{3}{*}{(0.1,0.9,0.5)}} & 2   & 5.8     & 6.3    & 5.6    & 6.3    & 9.6    & 5.2    & 5.7     & 8.4    & 6.5    & 7.5    & 10.1   & 7.2    \\ \cline{5-17} 
		\multicolumn{1}{|c|}{}                                                                & \multicolumn{3}{c|}{}                               & 3   & 7.8     & 7.6    & 8.3    & 7.5    & 15.3   & 4.9    & 5.3     & 8.3    & 5.5    & 7.9    & 15.0   & 9.5    \\ \cline{5-17} 
		\multicolumn{1}{|c|}{}                                                                & \multicolumn{3}{c|}{}                               & 5   & 7.3     & 5.9    & 8.7    & 8.0    & 36     & 3.8    & 5.5     & 10.6   & 6.6    & 11.1   & 30.7   & 10.8   \\ \hline
		\multicolumn{1}{|c|}{\multirow{6}{*}{2}}                                              & \multicolumn{3}{c|}{\multirow{3}{*}{(0.1,0.1,0.1)}} & 2   & 6.5     & 4.9    & 10.2   & 5.4    & 7.7    & 4.9    & 6.8     & 5.8    & 11.9   & 7.5    & 7.9    & 5.2    \\ \cline{5-17} 
		\multicolumn{1}{|c|}{}                                                                & \multicolumn{3}{c|}{}                               & 3   & 8.5     & 6.3    & 13.2   & 7.4    & 8.4    & 4.7    & 7.7     & 4.7    & 15.4   & 7.0    & 8.6    & 4.6    \\ \cline{5-17} 
		\multicolumn{1}{|c|}{}                                                                & \multicolumn{3}{c|}{}                               & 5   & 9.2     & 4.4    & 25.3   & 6.9    & 12.6   & 5.2    & 10.7    & 5.7    & 26.9   & 6.2    & 13.9   & 6.3    \\ \cline{2-17} 
		\multicolumn{1}{|c|}{}                                                                & \multicolumn{3}{c|}{\multirow{3}{*}{(0.1,0.9,0.5)}} & 2   & 6.0     & 4.7    & 10.8   & 7.7    & 7.1    & 5.8    & 5.8     & 6.6    & 9.8    & 10.2   & 7.0    & 6.2    \\ \cline{5-17} 
		\multicolumn{1}{|c|}{}                                                                & \multicolumn{3}{c|}{}                               & 3   & 6.8     & 5.2    & 13.7   & 9.0    & 8.0    & 4.5    & 8.1     & 8.5    & 11.3   & 11.0   & 7.9    & 8.6    \\ \cline{5-17} 
		\multicolumn{1}{|c|}{}                                                                & \multicolumn{3}{c|}{}                               & 5   & 9.4     & 5.9    & 25     & 9.9    & 11.9   & 5.0    & 6.3     & 8.6    & 12.0   & 12.2   & 9.7    & 9.1    \\ \hline
		\multicolumn{1}{|c|}{\multirow{6}{*}{3}}                                              & \multicolumn{3}{c|}{\multirow{3}{*}{(0.1,0.1,0.1)}} & 2   & 5.6     & 4.8    & 9.5    & 6.0    & 6.9    & 5.5    & 6.7     & 5.7    & 13.8   & 8.8    & 7.2    & 5.6    \\ \cline{5-17} 
		\multicolumn{1}{|c|}{}                                                                & \multicolumn{3}{c|}{}                               & 3   & 6.9     & 5.4    & 13.4   & 5.7    & 9.2    & 5.8    & 7.9     & 6.3    & 16.6   & 5.2    & 9.9    & 6.8    \\ \cline{5-17} 
		\multicolumn{1}{|c|}{}                                                                & \multicolumn{3}{c|}{}                               & 5   & 8.5     & 5.0    & 19.5   & 6.7    & 13.4   & 5.7    & 11.5    & 6.1    & 30.9   & 6.7    & 14.0   & 5.3    \\ \cline{2-17} 
		\multicolumn{1}{|c|}{}                                                                & \multicolumn{3}{c|}{\multirow{3}{*}{(0.1,0.9,0.5)}} & 2   & 5.2     & 6.3    & 9.6    & 9.2    & 5.0    & 4.1    & 7.0     & 7.7    & 9.6    & 11.1   & 7.3    & 6.8    \\ \cline{5-17} 
		\multicolumn{1}{|c|}{}                                                                & \multicolumn{3}{c|}{}                               & 3   & 5.8     & 6.5    & 9.8    & 9.4    & 8.4    & 4.6    & 6.1     & 9.7    & 10.9   & 11.8   & 6.1    & 7.3    \\ \cline{5-17} 
		\multicolumn{1}{|c|}{}                                                                & \multicolumn{3}{c|}{}                               & 5   & 7.9     & 9.0    & 15.1   & 13.8   & 11.3   & 4.6    & 5.7     & 9.6    & 13.4   & 13.9   & 9.7    & 9.4    \\ \hline
		\multicolumn{1}{|c|}{\multirow{6}{*}{4}}                                              & \multicolumn{3}{c|}{\multirow{3}{*}{(0.1,0.1,0.1)}} & 2   & 9.0     & 6.8    & 9.6    & 6.3    & 14.2   & 8.0    & 7.0     & 5.5    & 7.5    & 5.7    & 11.9   & 7.4    \\ \cline{5-17} 
		\multicolumn{1}{|c|}{}                                                                & \multicolumn{3}{c|}{}                               & 3   & 8.0     & 4.9    & 11.4   & 6.6    & 20.8   & 6.9    & 8.7     & 5.8    & 8.7    & 6.3    & 19.6   & 6.6    \\ \cline{5-17} 
		\multicolumn{1}{|c|}{}                                                                & \multicolumn{3}{c|}{}                               & 5   & 13.5    & 5.1    & 23.9   & 7.2    & 40.5   & 4.7    & 10.1    & 5.5    & 11.4   & 4.9    & 40.2   & 6.3    \\ \cline{2-17} 
		\multicolumn{1}{|c|}{}                                                                & \multicolumn{3}{c|}{\multirow{3}{*}{(0.1,0.9,0.5)}} & 2   & 8.6     & 7.2    & 10.6   & 8.1    & 11.8   & 6.2    & 6.1     & 6.0    & 5.8    & 6.3    & 10.9   & 7.5    \\ \cline{5-17} 
		\multicolumn{1}{|c|}{}                                                                & \multicolumn{3}{c|}{}                               & 3   & 8.7     & 5.5    & 14.3   & 6.0    & 16.0   & 4.2    & 8.7     & 8.3    & 11.0    & 8.4    & 14.8   & 7.4    \\ \cline{5-17} 
		\multicolumn{1}{|c|}{}                                                                & \multicolumn{3}{c|}{}                               & 5   & 14.1    & 6.9    & 27.6   & 7.2    & 35.3   & 5.3    & 7.7     & 8.9    & 13.5   & 9.4    & 30.5   & 9.9    \\ \hline
	\end{tabular}
\label{Table:Typeone-normal}
\end{table}

\begin{table}[!htb]
	\centering
	\caption{Achieved Type-I error rate ($\times100$) for data generated from Multivariate Log-Normal distribution with missing pattern as in Table \ref{Table:schematic} for $d=2,3,5$. Here, $Q_n$ is the Wald-type statistic proposed in (\ref{Q_N}); $F_n$ is ANOVA-type statistic for small sample approximation proposed in (\ref{ANOVA}); $Q_n^{(1)}$ and $F_n^{(1)}$ are Wald-type and ANOVA-type tests for the Incomplete method; $Q_n^{(2)}$ and $F_n^{(2)}$ are Wald-type and ANOVA-type tests for the Complete method. The nominal Type-I error rate is $\alpha=0.05$.}
	\begin{tabular}{ccllc|c|c|c|c|c|c|c|c|c|c|c|c|}
		\cline{6-17}
		\multicolumn{5}{c}{}                                                                                                                                 & \multicolumn{6}{|c|}{$(\sigma_1^2,\sigma_2^2)=(1,1)$} & \multicolumn{6}{c|}{$(\sigma_1^2,\sigma_2^2)=(1,5)$} \\ \hline
		\multicolumn{1}{|c|}{\begin{tabular}[c]{@{}c@{}}Sample\\ Size\\ Setting\end{tabular}} & \multicolumn{3}{c|}{$(\rho_1,\rho_2,\rho_{12})$}       & $d$ & $Q_n$   & $F_n$  & $Q_n^{(1)}$  & $F_n^{(1)}$  & $Q_n^{(2)}$  & $F_n^{(2)}$  & $Q_n$   & $F_n$  & $Q_n^{(1)}$  & $F_n^{(1)}$  & $Q_n^{(2)}$  & $F_n^{(2)}$  \\ \hline
		\multicolumn{1}{|c|}{\multirow{6}{*}{1}}                                              & \multicolumn{3}{c|}{\multirow{3}{*}{(0.1,0.1,0.1)}}    & 2   & 7.1     & 6.5    & 6.9    & 5.5    & 11.7   & 6      & 5.6     & 4.9    & 5.9    & 4.8    & 12.5   & 7.7    \\ \cline{5-17} 
		\multicolumn{1}{|c|}{}                                                                & \multicolumn{3}{c|}{}                                  & 3   & 7.2     & 4.7    & 7.1    & 5.8    & 22.2   & 6.5    & 7.1     & 4.7    & 6.3    & 4.1    & 16.9   & 6.2    \\ \cline{5-17} 
		\multicolumn{1}{|c|}{}                                                                & \multicolumn{3}{c|}{}                                  & 5   & 6.6     & 3.9    & 7.9    & 4.8    & 37.2   & 4.5    & 10.7    & 7      & 11.1   & 5.8    & 41.5   & 6.2    \\ \cline{2-17} 
		\multicolumn{1}{|c|}{}                                                                & \multicolumn{3}{c|}{\multirow{3}{*}{(-0.1,-0.1,-0.1)}} & 2   & 5.6     & 5      & 7.1    & 6.3    & 13.8   & 8.1    & 7.6     & 5.7    & 7.1    & 5.9    & 13.9   & 6.9    \\ \cline{5-17} 
		\multicolumn{1}{|c|}{}                                                                & \multicolumn{3}{c|}{}                                  & 3   & 6.3     & 4.9    & 6.1    & 4.4    & 20.3   & 6.6    & 7       & 4.6    & 7.2    & 5.3    & 21.5   & 6.1    \\ \cline{5-17} 
		\multicolumn{1}{|c|}{}                                                                & \multicolumn{3}{c|}{}                                  & 5   & 11.4    & 6.1    & 10.9   & 6.5    & 40.7   & 5.5    & 9       & 5.7    & 11.5   & 5.3    & 40.1   & 6.4    \\ \hline
		\multicolumn{1}{|c|}{\multirow{6}{*}{2}}                                              & \multicolumn{3}{c|}{\multirow{3}{*}{(0.1,0.1,0.1)}}    & 2   & 6       & 5.4    & 9.5    & 7      & 6.3    & 4.7    & 7       & 6.5    & 11.8   & 6.5    & 7.4    & 6.3    \\ \cline{5-17} 
		\multicolumn{1}{|c|}{}                                                                & \multicolumn{3}{c|}{}                                  & 3   & 5.6     & 4.6    & 12.5   & 5.5    & 8.4    & 4.8    & 7.9     & 4.7    & 15.9   & 6.2    & 7.8    & 5.6    \\ \cline{5-17} 
		\multicolumn{1}{|c|}{}                                                                & \multicolumn{3}{c|}{}                                  & 5   & 9.1     & 5.8    & 20.3   & 5.4    & 12.8   & 5.8    & 9.8     & 5.4    & 28     & 5.3    & 13.1   & 6.8    \\ \cline{2-17} 
		\multicolumn{1}{|c|}{}                                                                & \multicolumn{3}{c|}{\multirow{3}{*}{(-0.1,-0.1,-0.1)}} & 2   & 6.7     & 6      & 9.9    & 6.2    & 8      & 6.8    & 8       & 5.9    & 11.8   & 7.8    & 8.6    & 6.8    \\ \cline{5-17} 
		\multicolumn{1}{|c|}{}                                                                & \multicolumn{3}{c|}{}                                  & 3   & 7.1     & 5.6    & 15.7   & 6.6    & 9.5    & 4.8    & 8       & 6.4    & 14.8   & 6.4    & 9.7    & 5.6    \\ \cline{5-17} 
		\multicolumn{1}{|c|}{}                                                                & \multicolumn{3}{c|}{}                                  & 5   & 11      & 4.9    & 19.5   & 4.9    & 13.7   & 6.2    & 12.5    & 6.5    & 27.1   & 6      & 15.3   & 6.2    \\ \hline
		\multicolumn{1}{|c|}{\multirow{6}{*}{3}}                                              & \multicolumn{3}{c|}{\multirow{3}{*}{(0.1,0.1,0.1)}}    & 2   & 6.6     & 5.8    & 8.1    & 6.1    & 7.1    & 5.7    & 6.5     & 5.6    & 11.8   & 6.6    & 8.1    & 5.8    \\ \cline{5-17} 
		\multicolumn{1}{|c|}{}                                                                & \multicolumn{3}{c|}{}                                  & 3   & 7.2     & 5.5    & 14.9   & 6.9    & 9.1    & 5.6    & 7.1     & 5.1    & 16.5   & 6.9    & 9.8    & 5.2    \\ \cline{5-17} 
		\multicolumn{1}{|c|}{}                                                                & \multicolumn{3}{c|}{}                                  & 5   & 8.4     & 5.7    & 18.8   & 5.2    & 13.7   & 5.7    & 10.5    & 5.1    & 33.3   & 6.5    & 13.6   & 4.9    \\ \cline{2-17} 
		\multicolumn{1}{|c|}{}                                                                & \multicolumn{3}{c|}{\multirow{3}{*}{(-0.1,-0.1,-0.1)}} & 2   & 6.7     & 5.6    & 9.4    & 5.9    & 8.4    & 6.5    & 6.6     & 5.5    & 13.2   & 8.2    & 7.1    & 5.6    \\ \cline{5-17} 
		\multicolumn{1}{|c|}{}                                                                & \multicolumn{3}{c|}{}                                  & 3   & 7.2     & 4.3    & 10.4   & 5.4    & 9.7    & 5      & 7.3     & 4.9    & 19.1   & 7      & 8.5    & 5.2    \\ \cline{5-17} 
		\multicolumn{1}{|c|}{}                                                                & \multicolumn{3}{c|}{}                                  & 5   & 9.1     & 5.6    & 19.4   & 6.1    & 13.4   & 5.6    & 8.6     & 4.7    & 35.2   & 6.5    & 12.6   & 5.7    \\ \hline
		\multicolumn{1}{|c|}{\multirow{6}{*}{4}}                                              & \multicolumn{3}{c|}{\multirow{3}{*}{(0.1,0.1,0.1)}}    & 2   & 7.8     & 5.4    & 10     & 6.3    & 13.5   & 6.2    & 7       & 5.7    & 7.7    & 6      & 14.1   & 7.5    \\ \cline{5-17} 
		\multicolumn{1}{|c|}{}                                                                & \multicolumn{3}{c|}{}                                  & 3   & 9.4     & 5.9    & 12.3   & 6.1    & 19.5   & 6.6    & 6.1     & 4.4    & 7.9    & 5.8    & 18.4   & 5      \\ \cline{5-17} 
		\multicolumn{1}{|c|}{}                                                                & \multicolumn{3}{c|}{}                                  & 5   & 12.6    & 5.1    & 20.4   & 6.2    & 40.3   & 5.3    & 8.8     & 5      & 10.4   & 4.5    & 42.3   & 5.6    \\ \cline{2-17} 
		\multicolumn{1}{|c|}{}                                                                & \multicolumn{3}{c|}{\multirow{3}{*}{(-0.1,-0.1,-0.1)}} & 2   & 7.3     & 5.4    & 9.8    & 6.1    & 14.7   & 6.5    & 6.8     & 5.3    & 7.3    & 6.6    & 12.9   & 7.6    \\ \cline{5-17} 
		\multicolumn{1}{|c|}{}                                                                & \multicolumn{3}{c|}{}                                  & 3   & 9.8     & 6.9    & 13     & 6.4    & 20.3   & 6.9    & 7.4     & 4.8    & 9      & 5.9    & 20.1   & 5.6    \\ \cline{5-17} 
		\multicolumn{1}{|c|}{}                                                                & \multicolumn{3}{c|}{}                                  & 5   & 11.3    & 5.2    & 19.1   & 5.9    & 40.6   & 5.6    & 11.8    & 6.2    & 12.5   & 6.2    & 42.1   & 7      \\ \hline
	\end{tabular}
\label{Table:Typeone-exp}
\end{table}

\begin{table}[!htb]
	\centering
	\caption{Achieved Type-I error rate ($\times100$) for data generated from Multivariate Cauchy distribution with missing pattern as in Table \ref{Table:schematic} for $d=2,3,5$. Here, $Q_n$ is the Wald-type statistic proposed in (\ref{Q_N}); $F_n$ is ANOVA-type statistic for small sample approximation proposed in (\ref{ANOVA}); $Q_n^{(1)}$ and $F_n^{(1)}$ are Wald-type and ANOVA-type tests for the Incomplete method; $Q_n^{(2)}$ and $F_n^{(2)}$ are Wald-type and ANOVA-type tests for the Complete method. The nominal Type-I error rate is $\alpha=0.05$.}
	\begin{tabular}{ccllc|c|c|c|c|c|c|c|c|c|c|c|c|}
		\cline{6-17}
		\multicolumn{5}{c}{}                                                                                                                                 & \multicolumn{6}{|c|}{$(\sigma_1^2,\sigma_2^2)=(1,1)$} & \multicolumn{6}{c|}{$(\sigma_1^2,\sigma_2^2)=(1,5)$} \\ \hline
		\multicolumn{1}{|c|}{\begin{tabular}[c]{@{}c@{}}Sample\\ Size\\ Setting\end{tabular}} & \multicolumn{3}{c|}{$(\rho_1,\rho_2,\rho_{12})$}       & $d$ & $Q_n$   & $F_n$  & $Q_n^{(1)}$  & $F_n^{(1)}$  & $Q_n^{(2)}$  & $F_n^{(2)}$  & $Q_n$   & $F_n$  & $Q_n^{(1)}$  & $F_n^{(1)}$  & $Q_n^{(2)}$  & $F_n^{(2)}$  \\ \hline
		\multicolumn{1}{|c|}{\multirow{6}{*}{1}}                                              & \multicolumn{3}{c|}{\multirow{3}{*}{(0.1,0.1,0.1)}}    & 2   & 6.9     & 5.8    & 7.6    & 6.5    & 11.9   & 6.3    & 6.2     & 5.2    & 5.8    & 4.9    & 13.8   & 6.1    \\ \cline{5-17} 
		\multicolumn{1}{|c|}{}                                                                & \multicolumn{3}{c|}{}                                  & 3   & 7.6     & 5.8    & 6.5    & 4.3    & 20.8   & 6.6    & 6.4     & 4.4    & 6.7    & 4.3    & 20.8   & 7      \\ \cline{5-17} 
		\multicolumn{1}{|c|}{}                                                                & \multicolumn{3}{c|}{}                                  & 5   & 8       & 4.5    & 8.5    & 4.9    & 47.5   & 6.6    & 9.1     & 5.3    & 10.5   & 5.9    & 43.9   & 5.9    \\ \cline{2-17} 
		\multicolumn{1}{|c|}{}                                                                & \multicolumn{3}{c|}{\multirow{3}{*}{(-0.1,-0.1,-0.1)}} & 2   & 7.5     & 6.2    & 8.1    & 6.5    & 13.2   & 7.7    & 6.4     & 5.2    & 6.1    & 5.3    & 15.3   & 8.6    \\ \cline{5-17} 
		\multicolumn{1}{|c|}{}                                                                & \multicolumn{3}{c|}{}                                  & 3   & 7.3     & 5.2    & 5.9    & 4.1    & 22.9   & 6.6    & 7       & 4.8    & 8.4    & 5.2    & 22.9   & 7.4    \\ \cline{5-17} 
		\multicolumn{1}{|c|}{}                                                                & \multicolumn{3}{c|}{}                                  & 5   & 10.9    & 6.2    & 9.2    & 5.6    & 48     & 7.3    & 9.6     & 4.3    & 9.5    & 4.9    & 46.7   & 8.3    \\ \hline
		\multicolumn{1}{|c|}{\multirow{6}{*}{2}}                                              & \multicolumn{3}{c|}{\multirow{3}{*}{(0.1,0.1,0.1)}}    & 2   & 6.2     & 4.7    & 12.5   & 6.8    & 6.5    & 4.3    & 7.8     & 6.4    & 10     & 5.5    & 7.4    & 5.7    \\ \cline{5-17} 
		\multicolumn{1}{|c|}{}                                                                & \multicolumn{3}{c|}{}                                  & 3   & 7.2     & 5.3    & 15.3   & 6.4    & 9.2    & 5.3    & 8       & 5.8    & 14.7   & 6.7    & 10.2   & 5.2    \\ \cline{5-17} 
		\multicolumn{1}{|c|}{}                                                                & \multicolumn{3}{c|}{}                                  & 5   & 10.2    & 4.4    & 23     & 4.8    & 16     & 4.6    & 11.2    & 6.3    & 26.2   & 6.4    & 15.8   & 6.1    \\ \cline{2-17} 
		\multicolumn{1}{|c|}{}                                                                & \multicolumn{3}{c|}{\multirow{3}{*}{(-0.1,-0.1,-0.1)}} & 2   & 6.4     & 5.4    & 10.9   & 7.5    & 8.3    & 6.3    & 6       & 4.9    & 9.8    & 5.8    & 7.3    & 4.7    \\ \cline{5-17} 
		\multicolumn{1}{|c|}{}                                                                & \multicolumn{3}{c|}{}                                  & 3   & 9.8     & 7.4    & 14.2   & 6      & 10.1   & 6.8    & 7.4     & 5.1    & 13.7   & 4.5    & 9.4    & 4.9    \\ \cline{5-17} 
		\multicolumn{1}{|c|}{}                                                                & \multicolumn{3}{c|}{}                                  & 5   & 11.2    & 4.2    & 23.7   & 5.6    & 15.8   & 5      & 12.1    & 5.3    & 25.4   & 6.2    & 14.8   & 5.9    \\ \hline
		\multicolumn{1}{|c|}{\multirow{6}{*}{3}}                                              & \multicolumn{3}{c|}{\multirow{3}{*}{(0.1,0.1,0.1)}}    & 2   & 8       & 6.1    & 10.2   & 5.9    & 9.1    & 7.3    & 6.6     & 5.8    & 10.8   & 6.9    & 7.3    & 5.9    \\ \cline{5-17} 
		\multicolumn{1}{|c|}{}                                                                & \multicolumn{3}{c|}{}                                  & 3   & 7.9     & 7      & 15.2   & 7.2    & 9.3    & 5.4    & 7.5     & 5.1    & 14     & 5.7    & 10.1   & 6.4    \\ \cline{5-17} 
		\multicolumn{1}{|c|}{}                                                                & \multicolumn{3}{c|}{}                                  & 5   & 10.5    & 5.2    & 19.4   & 5.7    & 17     & 5.5    & 10.7    & 5.9    & 29.4   & 5.3    & 16.3   & 6.3    \\ \cline{2-17} 
		\multicolumn{1}{|c|}{}                                                                & \multicolumn{3}{c|}{\multirow{3}{*}{(-0.1,-0.1,-0.1)}} & 2   & 8.9     & 7      & 9.7    & 5.5    & 10.5   & 7.4    & 5       & 4.5    & 10.8   & 7.3    & 6.9    & 5.7    \\ \cline{5-17} 
		\multicolumn{1}{|c|}{}                                                                & \multicolumn{3}{c|}{}                                  & 3   & 6.3     & 4.1    & 10.9   & 4.8    & 10.6   & 6.1    & 6.8     & 5      & 16.3   & 6.6    & 8.1    & 5.3    \\ \cline{5-17} 
		\multicolumn{1}{|c|}{}                                                                & \multicolumn{3}{c|}{}                                  & 5   & 7.4     & 3.3    & 17.1   & 3.8    & 14.6   & 5.1    & 10.6    & 6.3    & 28.5   & 8      & 13.7   & 4.7    \\ \hline
		\multicolumn{1}{|c|}{\multirow{6}{*}{4}}                                              & \multicolumn{3}{c|}{\multirow{3}{*}{(0.1,0.1,0.1)}}    & 2   & 7       & 5.5    & 10     & 6.6    & 13.6   & 8.2    & 6.7     & 4.8    & 9.2    & 6      & 12.1   & 6.9    \\ \cline{5-17} 
		\multicolumn{1}{|c|}{}                                                                & \multicolumn{3}{c|}{}                                  & 3   & 7.7     & 4.5    & 12.2   & 5.4    & 19.8   & 5.2    & 9.3     & 5.4    & 9.6    & 5.5    & 19.9   & 6.3    \\ \cline{5-17} 
		\multicolumn{1}{|c|}{}                                                                & \multicolumn{3}{c|}{}                                  & 5   & 12.2    & 6.1    & 19.6   & 5.9    & 45.5   & 5.9    & 10.3    & 4.6    & 13     & 4.5    & 43.6   & 6.3    \\ \cline{2-17} 
		\multicolumn{1}{|c|}{}                                                                & \multicolumn{3}{c|}{\multirow{3}{*}{(-0.1,-0.1,-0.1)}} & 2   & 7.5     & 5.6    & 9.7    & 6.8    & 16.7   & 8.3    & 6.2     & 5.7    & 7.4    & 5.4    & 14.6   & 8.5    \\ \cline{5-17} 
		\multicolumn{1}{|c|}{}                                                                & \multicolumn{3}{c|}{}                                  & 3   & 9.5     & 5.1    & 12.4   & 5.5    & 24.3   & 7.2    & 10.3    & 6.4    & 11.9   & 6.8    & 22.6   & 6.4    \\ \cline{5-17} 
		\multicolumn{1}{|c|}{}                                                                & \multicolumn{3}{c|}{}                                  & 5   & 12.1    & 4.6    & 18.8   & 5.6    & 44.3   & 6.6    & 12.5    & 5.1    & 16.3   & 4.4    & 44.6   & 6.2    \\ \hline
	\end{tabular}
\label{Table:Typeone-cauchy}
\end{table}

\subsection{Power Study}
\label{sec:PowerStudy}
To investigate power of the two tests based on $Q_n$ and $F_n$, bivariate situation ($d=2$) is considered with three multivariate distributions. In this situation, the first sample is drawn with mean $\bm{\mu}=(0,0)$ and the second one is drawn with mean $\bm{\mu}=(\delta_1,\delta_2)$. Three types of location shift alternatives $(\delta_1,\delta_2)=\delta(1,k)$ for $k=0,1,2$ are considered, where $\delta\in\{0.3,0.6,0.9\}$. Power simulation results are displayed in Tables \ref{Table:Power-normal} and \ref{Table:Power-cauchy}.

Looking at the power results in these three tables, it is clear to see that data generated from homogeneous distribution yield higher power than the heterogeneous ones for fixed sample size and location shift. Also, mean shift $\bm{\mu}_2=(\delta,2\delta)$ yields higher power than $\bm{\mu}_2=(\delta,\delta)$, which in turn yields higher power than $\bm{\mu}_2=(\delta,0)$. This is expected because $\bm{\mu}_2=(\delta,2\delta)$ is a stronger departure from the null compared to $\bm{\mu}_2=(\delta,\delta)$ and $\bm{\mu}_2=(\delta,0)$ in the sense that $\parallel(\delta,2\delta)\parallel_2^2>\parallel(\delta,\delta)\parallel_2^2>\parallel(\delta,0)\parallel_2^2$.

Furthermore, larger values $\delta$ produce higher power for a given location alternative. $Q_n$ achieves the highest power among all Wald-type statistics, and similarly, $F_n$ attains higher power compared to $F_n^{(1)}$ and $F_{(2)}$. We also note that powers for Incomplete and Complete methods are related to sample size allocations. More specifically, $Q_n^{(1)}$ and $F_n^{(1)}$ perform better in Setting 1 and 4 since incomplete data make up higher proportion in the data set. For the same reason, $Q_n^{(2)}$ and $F_n^{(2)}$ perform better in Setting 2 and 3 as complete data weight more in this case. Different from the Type-I error simulation results, the data distributions greatly affect the powers. 

While the achieved power for data from Multivariate Discretized Normal distribution are generally close to those from the Multivariate Log-Normal distribution, they both are greater than the powers for data from Multivariate Cauchy distribution.

\begin{table}[!htb]
	\centering
	\caption{Achieved power ($\times100$) for data generated from Discretized Multivariate Normal distribution with missing pattern as in Table \ref{Table:schematic} for $d=2$. Here, $Q_n$ is the Wald-type statistic proposed in (\ref{Q_N}); $F_n$ is ANOVA-type statistic for small sample approximation proposed in (\ref{ANOVA}); $Q_n^{(1)}$ and $F_n^{(1)}$ are Wald-type and ANOVA-type tests in the Incomplete method; $Q_n^{(2)}$ and $F_n^{(2)}$ are Wald-type and ANOVA-type tests in the Complete method. The correlation coefficients are $(\rho_1,\rho_2,\rho_{12})=(0.1,0.1,0.1)$.}
	\begin{tabular}{ccc|c|c|c|c|c|c|c|c|c|c|c|c|}
		\cline{4-15}
		&                                 &            & \multicolumn{6}{c|}{$(\sigma_1^2,\sigma_2^2)=(1,1)$}                  & \multicolumn{6}{c|}{$(\sigma_1^2,\sigma_2^2)=(1,5)$}                  \\ \hline
		\multicolumn{1}{|c|}{\begin{tabular}[c]{@{}c@{}}Sample\\ Size\\ Setting\end{tabular}} & \multicolumn{1}{c|}{$\delta_1$} & $\delta_2$ & $Q_n$ & $F_n$ & $Q_n^{(1)}$ & $F_n^{(1)}$ & $Q_n^{(2)}$ & $F_n^{(2)}$ & $Q_n$ & $F_n$ & $Q_n^{(1)}$ & $F_n^{(1)}$ & $Q_n^{(2)}$ & $F_n^{(2)}$ \\ \hline
		\multicolumn{1}{|c|}{\multirow{6}{*}{1}}                                              & \multicolumn{1}{c|}{0}          & 0.3        & 21.6  & 20.2  & 18.7        & 17.3        & 20.5        & 11.3        & 11.3  & 10    & 10.9        & 9           & 13.4        & 7.3         \\ \cline{2-15} 
		\multicolumn{1}{|c|}{}                                                                & \multicolumn{1}{c|}{0.3}        & 0.3        & 35.9  & 36.5  & 28.6        & 28.5        & 26.6        & 17.3        & 15.2  & 14    & 13.5        & 12.1        & 16.9        & 8.9         \\ \cline{2-15} 
		\multicolumn{1}{|c|}{}                                                                & \multicolumn{1}{c|}{0.6}        & 0.6        & 91.8  & 92.4  & 79.3        & 81          & 49.3        & 38.6        & 44.8  & 44.8  & 37          & 36          & 25.7        & 18          \\ \cline{2-15} 
		\multicolumn{1}{|c|}{}                                                                & \multicolumn{1}{c|}{0.9}        & 0.9        & 100   & 100   & 99.3        & 99.6        & 77.9        & 72.2        & 79.2  & 80.3  & 66.1        & 65.9        & 40.2        & 32.3        \\ \cline{2-15} 
		\multicolumn{1}{|c|}{}                                                                & \multicolumn{1}{c|}{0.3}        & 0.6        & 77.3  & 76.6  & 61.7        & 62.7        & 37.8        & 26.5        & 32.6  & 30.6  & 25.3        & 23.9        & 21.9        & 15.5        \\ \cline{2-15} 
		\multicolumn{1}{|c|}{}                                                                & \multicolumn{1}{c|}{0.3}        & 0.9        & 96.4  & 96.4  & 89.5        & 89.8        & 57.9        & 45.8        & 53.9  & 52.5  & 39.9        & 38          & 29.9        & 20.9        \\ \hline
		\multicolumn{1}{|c|}{\multirow{6}{*}{2}}                                              & \multicolumn{1}{c|}{0}          & 0.3        & 22.8  & 20.7  & 15.2        & 10.4        & 21          & 17.4        & 11.9  & 10.7  & 14.2        & 9.5         & 11.6        & 9.4         \\ \cline{2-15} 
		\multicolumn{1}{|c|}{}                                                                & \multicolumn{1}{c|}{0.3}        & 0.3        & 41.9  & 38.8  & 18.7        & 13.5        & 34.9        & 30.4        & 17.4  & 15    & 13.3        & 8.9         & 16.3        & 14.2        \\ \cline{2-15} 
		\multicolumn{1}{|c|}{}                                                                & \multicolumn{1}{c|}{0.6}        & 0.6        & 94.6  & 94.8  & 42.6        & 38.9        & 86.9        & 86.1        & 47    & 45.5  & 20          & 14.2        & 40.9        & 38.3        \\ \cline{2-15} 
		\multicolumn{1}{|c|}{}                                                                & \multicolumn{1}{c|}{0.9}        & 0.9        & 100   & 100   & 72.5        & 70.4        & 99.6        & 99.5        & 80.5  & 80.5  & 34.7        & 28.3        & 69.2        & 68.4        \\ \cline{2-15} 
		\multicolumn{1}{|c|}{}                                                                & \multicolumn{1}{c|}{0.3}        & 0.6        & 77.4  & 76.7  & 31.2        & 25.9        & 68.8        & 65.2        & 31.3  & 29.4  & 16.9        & 11.5        & 27.2        & 24.9        \\ \cline{2-15} 
		\multicolumn{1}{|c|}{}                                                                & \multicolumn{1}{c|}{0.3}        & 0.9        & 97.3  & 97.2  & 51.3        & 44.7        & 93.2        & 92.8        & 56.1  & 54.3  & 25.4        & 19.9        & 46.4        & 43.1        \\ \hline
		\multicolumn{1}{|c|}{\multirow{6}{*}{3}}                                              & \multicolumn{1}{c|}{0}          & 0.3        & 26    & 23.2  & 15.9        & 11.3        & 20.8        & 17.9        & 12.5  & 10.5  & 13.8        & 8           & 12.6        & 9.5         \\ \cline{2-15} 
		\multicolumn{1}{|c|}{}                                                                & \multicolumn{1}{c|}{0.3}        & 0.3        & 47.9  & 47.3  & 22.3        & 18          & 37.2        & 32.6        & 17.1  & 15.8  & 15.1        & 9.3         & 13.7        & 11.9        \\ \cline{2-15} 
		\multicolumn{1}{|c|}{}                                                                & \multicolumn{1}{c|}{0.6}        & 0.6        & 96.9  & 97.1  & 58          & 54          & 87.5        & 87          & 46.4  & 46    & 23.6        & 17.2        & 37.9        & 35.8        \\ \cline{2-15} 
		\multicolumn{1}{|c|}{}                                                                & \multicolumn{1}{c|}{0.9}        & 0.9        & 100   & 100   & 87          & 86.5        & 99.5        & 99.6        & 81.5  & 81.6  & 36.4        & 30.1        & 70.2        & 68.4        \\ \cline{2-15} 
		\multicolumn{1}{|c|}{}                                                                & \multicolumn{1}{c|}{0.3}        & 0.6        & 83.5  & 83.3  & 36          & 32.5        & 68.7        & 64.9        & 31.6  & 30.9  & 19.9        & 14.2        & 26.1        & 23.8        \\ \cline{2-15} 
		\multicolumn{1}{|c|}{}                                                                & \multicolumn{1}{c|}{0.3}        & 0.9        & 99    & 98.9  & 68.4        & 62.4        & 93.9        & 92          & 59.6  & 56.9  & 27.7        & 20.1        & 47.4        & 44.3        \\ \hline
		\multicolumn{1}{|c|}{\multirow{6}{*}{4}}                                              & \multicolumn{1}{c|}{0}          & 0.3        & 15.7  & 12.6  & 14.8        & 10.1        & 17          & 9.5         & 10    & 8.1   & 8.6         & 7.2         & 13.4        & 7.5         \\ \cline{2-15} 
		\multicolumn{1}{|c|}{}                                                                & \multicolumn{1}{c|}{0.3}        & 0.3        & 29.6  & 27.8  & 22.5        & 18.9        & 22.8        & 13.6        & 13    & 12.5  & 11.6        & 9.4         & 16          & 9.3         \\ \cline{2-15} 
		\multicolumn{1}{|c|}{}                                                                & \multicolumn{1}{c|}{0.6}        & 0.6        & 77.6  & 78.5  & 53.8        & 49.9        & 51          & 41.2        & 42.9  & 42.6  & 29.9        & 29.1        & 25          & 17.2        \\ \cline{2-15} 
		\multicolumn{1}{|c|}{}                                                                & \multicolumn{1}{c|}{0.9}        & 0.9        & 98.4  & 98.4  & 86.4        & 84.5        & 78.7        & 72          & 71    & 71.4  & 54.3        & 54.3        & 39.3        & 30.9        \\ \cline{2-15} 
		\multicolumn{1}{|c|}{}                                                                & \multicolumn{1}{c|}{0.3}        & 0.6        & 59.8  & 57.7  & 38.8        & 31.8        & 38          & 28.4        & 26.7  & 24.2  & 18.5        & 15.1        & 24.7        & 15.9        \\ \cline{2-15} 
		\multicolumn{1}{|c|}{}                                                                & \multicolumn{1}{c|}{0.3}        & 0.9        & 87.6  & 86.1  & 66.9        & 62.5        & 58.1        & 46.2        & 47.9  & 46    & 35.3        & 32.7        & 27.3        & 18.8        \\ \hline
	\end{tabular}
\label{Table:Power-normal}
\end{table}

\begin{table}[!htb]
	\centering
	\caption{Achieved power ($\times100$) for data generated from Multivariate Log-Normal distribution with missing pattern as in Table \ref{Table:schematic} for $d=2$. Here, $Q_n$ is the Wald-type statistic proposed in (\ref{Q_N}); $F_n$ is ANOVA-type statistic for small sample approximation proposed in (\ref{ANOVA}); $Q_n^{(1)}$ and $F_n^{(1)}$ are Wald-type and ANOVA-type tests in the Incomplete method; $Q_n^{(2)}$ and $F_n^{(2)}$ are Wald-type and ANOVA-type tests in the Complete method. The correlation coefficients are $(\rho_1,\rho_2,\rho_{12})=(-0.1,-0.1,-0.1)$.}
	\begin{tabular}{ccc|c|c|c|c|c|c|c|c|c|c|c|c|}
		\cline{4-15}
		\multicolumn{3}{c}{}                                                                                                                 & \multicolumn{6}{|c|}{$(\sigma_1^2,\sigma_2^2)=(1,1)$}                  & \multicolumn{6}{c|}{$(\sigma_1^2,\sigma_2^2)=(1,5)$}                  \\ \hline
		\multicolumn{1}{|c|}{\begin{tabular}[c]{@{}c@{}}Sample\\ Size\\ Setting\end{tabular}} & \multicolumn{1}{c|}{$\delta_1$} & $\delta_2$ & $Q_n$ & $F_n$ & $Q_n^{(1)}$ & $F_n^{(1)}$ & $Q_n^{(2)}$ & $F_n^{(2)}$ & $Q_n$ & $F_n$ & $Q_n^{(1)}$ & $F_n^{(1)}$ & $Q_n^{(2)}$ & $F_n^{(2)}$ \\ \hline
		\multicolumn{1}{|c|}{\multirow{6}{*}{1}}                                              & \multicolumn{1}{c|}{0}          & 0.3        & 18.3  & 17.1  & 17.4        & 14.8        & 17.4        & 10.6        & 11.4  & 9.5   & 10.7        & 7.7         & 16          & 9.7         \\ \cline{2-15} 
		\multicolumn{1}{|c|}{}                                                                & \multicolumn{1}{c|}{0.3}        & 0.3        & 41.6  & 37.1  & 32.4        & 27.6        & 27.7        & 17.8        & 16.9  & 14.1  & 15.5        & 13.3        & 16.9        & 8.1         \\ \cline{2-15} 
		\multicolumn{1}{|c|}{}                                                                & \multicolumn{1}{c|}{0.6}        & 0.6        & 93.9  & 92.4  & 86.8        & 83.6        & 44.9        & 34.1        & 48.4  & 44.4  & 40.3        & 35.3        & 22.8        & 15.8        \\ \cline{2-15} 
		\multicolumn{1}{|c|}{}                                                                & \multicolumn{1}{c|}{0.9}        & 0.9        & 100   & 100   & 99.6        & 99.3        & 76.9        & 68          & 82.5  & 79.3  & 72.8        & 68.9        & 34.5        & 24.6        \\ \cline{2-15} 
		\multicolumn{1}{|c|}{}                                                                & \multicolumn{1}{c|}{0.3}        & 0.6        & 77.7  & 74.5  & 69.2        & 64.1        & 34.7        & 24          & 30.6  & 26.8  & 25.4        & 22.3        & 20.3        & 13.5        \\ \cline{2-15} 
		\multicolumn{1}{|c|}{}                                                                & \multicolumn{1}{c|}{0.3}        & 0.9        & 96.7  & 95.6  & 92.4        & 90.5        & 54.6        & 44.5        & 55.9  & 52.9  & 45.7        & 41.5        & 28          & 18.6        \\ \hline
		\multicolumn{1}{|c|}{\multirow{6}{*}{2}}                                              & \multicolumn{1}{c|}{0}          & 0.3        & 21.4  & 19.4  & 13.7        & 9.9         & 18.4        & 15.9        & 11.6  & 9.5   & 14.1        & 9.6         & 11.4        & 9.3         \\ \cline{2-15} 
		\multicolumn{1}{|c|}{}                                                                & \multicolumn{1}{c|}{0.3}        & 0.3        & 38.1  & 33.8  & 19.5        & 13.6        & 30.4        & 28.2        & 16    & 13.7  & 14.1        & 7.5         & 14.9        & 13          \\ \cline{2-15} 
		\multicolumn{1}{|c|}{}                                                                & \multicolumn{1}{c|}{0.6}        & 0.6        & 92.3  & 92.1  & 49.1        & 39.2        & 82.6        & 80          & 46.6  & 43.5  & 24          & 16.2        & 38.2        & 34.8        \\ \cline{2-15} 
		\multicolumn{1}{|c|}{}                                                                & \multicolumn{1}{c|}{0.9}        & 0.9        & 100   & 99.9  & 78.5        & 72.8        & 98.7        & 98.7        & 78.9  & 76.1  & 39.1        & 28.3        & 67.9        & 63.8        \\ \cline{2-15} 
		\multicolumn{1}{|c|}{}                                                                & \multicolumn{1}{c|}{0.3}        & 0.6        & 71.5  & 69.6  & 33.7        & 25.1        & 59.6        & 55.8        & 31.5  & 27.6  & 17.7        & 11.4        & 26.2        & 22.1        \\ \cline{2-15} 
		\multicolumn{1}{|c|}{}                                                                & \multicolumn{1}{c|}{0.3}        & 0.9        & 95.4  & 94.8  & 53.8        & 43.4        & 88.2        & 86          & 53.9  & 51.7  & 26.5        & 17.7        & 44.7        & 40.4        \\ \hline
		\multicolumn{1}{|c|}{\multirow{6}{*}{3}}                                              & \multicolumn{1}{c|}{0}          & 0.3        & 27.2  & 25.3  & 16.3        & 12.3        & 20.8        & 16.9        & 11.4  & 10.7  & 14.2        & 9.2         & 11          & 8.4         \\ \cline{2-15} 
		\multicolumn{1}{|c|}{}                                                                & \multicolumn{1}{c|}{0.3}        & 0.3        & 44    & 41    & 24.5        & 18.7        & 30.9        & 26.7        & 16.2  & 13.1  & 16.6        & 9.9         & 14.4        & 11.8        \\ \cline{2-15} 
		\multicolumn{1}{|c|}{}                                                                & \multicolumn{1}{c|}{0.6}        & 0.6        & 95.4  & 95    & 59.9        & 51.7        & 80.3        & 78.4        & 46.9  & 41.4  & 26.3        & 17.6        & 32.8        & 28.5        \\ \cline{2-15} 
		\multicolumn{1}{|c|}{}                                                                & \multicolumn{1}{c|}{0.9}        & 0.9        & 100   & 100   & 91          & 87.4        & 98.8        & 98.3        & 82.5  & 79.3  & 41.3        & 33          & 65.9        & 63.3        \\ \cline{2-15} 
		\multicolumn{1}{|c|}{}                                                                & \multicolumn{1}{c|}{0.3}        & 0.6        & 81.3  & 78.5  & 46.5        & 37.1        & 57.5        & 54.8        & 34.6  & 31.8  & 22          & 14          & 26.3        & 22.1        \\ \cline{2-15} 
		\multicolumn{1}{|c|}{}                                                                & \multicolumn{1}{c|}{0.3}        & 0.9        & 98.2  & 98.2  & 68.5        & 58.9        & 89.8        & 88.1        & 57.3  & 52.6  & 29.8        & 21.7        & 45          & 39.9        \\ \hline
		\multicolumn{1}{|c|}{\multirow{6}{*}{4}}                                              & \multicolumn{1}{c|}{0}          & 0.3        & 17.9  & 15.4  & 17.2        & 12.9        & 19.2        & 11          & 9.2   & 8.5   & 11.1        & 8.5         & 11.9        & 7.1         \\ \cline{2-15} 
		\multicolumn{1}{|c|}{}                                                                & \multicolumn{1}{c|}{0.3}        & 0.3        & 31.7  & 26.8  & 25          & 18          & 24.6        & 15.2        & 17.5  & 14.8  & 15.2        & 11.1        & 16.3        & 9.5         \\ \cline{2-15} 
		\multicolumn{1}{|c|}{}                                                                & \multicolumn{1}{c|}{0.6}        & 0.6        & 80.5  & 76.3  & 62.1        & 53.2        & 44.2        & 34.5        & 42.9  & 38.9  & 36.2        & 30.2        & 25.1        & 17.2        \\ \cline{2-15} 
		\multicolumn{1}{|c|}{}                                                                & \multicolumn{1}{c|}{0.9}        & 0.9        & 99.4  & 99    & 92.5        & 89.9        & 75.2        & 68.5        & 74.4  & 70.5  & 58.6        & 52.8        & 38.9        & 27.7        \\ \cline{2-15} 
		\multicolumn{1}{|c|}{}                                                                & \multicolumn{1}{c|}{0.3}        & 0.6        & 60.1  & 55.1  & 46.2        & 36.5        & 35.7        & 23.6        & 28.3  & 24.4  & 22.1        & 18.5        & 21.3        & 12.3        \\ \cline{2-15} 
		\multicolumn{1}{|c|}{}                                                                & \multicolumn{1}{c|}{0.3}        & 0.9        & 88.6  & 86.3  & 71.1        & 62.4        & 54          & 40.5        & 51.7  & 46.7  & 40.6        & 34.8        & 27.3        & 18.4        \\ \hline
	\end{tabular}
\label{Table:Power-exp}
\end{table}

\begin{table}[!htb]
	\centering
	\caption{Achieved power ($\times100$) for data generated from Multivariate Cauchy distribution with missing pattern as in Table \ref{Table:schematic} for $d=2$. Here, $Q_n$ is the Wald-type statistic proposed in (\ref{Q_N}); $F_n$ is ANOVA-type statistic for small sample approximation proposed in (\ref{ANOVA}); $Q_n^{(1)}$ and $F_n^{(1)}$ are Wald-type and ANOVA-type tests in the Incomplete method; $Q_n^{(2)}$ and $F_n^{(2)}$ are Wald-type and ANOVA-type tests in the Complete method. The correlation coefficients are $(\rho_1,\rho_2,\rho_{12})=(0.1,0.9,0.5)$.}
	\begin{tabular}{ccc|c|c|c|c|c|c|c|c|c|c|c|c|}
		\cline{4-15}
		&                                 &            & \multicolumn{6}{c|}{$(\sigma_1^2,\sigma_2^2)=(1,1)$}                  & \multicolumn{6}{c|}{$(\sigma_1^2,\sigma_2^2)=(1,5)$}                  \\ \hline
		\multicolumn{1}{|c|}{\begin{tabular}[c]{@{}c@{}}Sample\\ Size\\ Setting\end{tabular}} & \multicolumn{1}{c|}{$\delta_1$} & $\delta_2$ & $Q_n$ & $F_n$ & $Q_n^{(1)}$ & $F_n^{(1)}$ & $Q_n^{(2)}$ & $F_n^{(2)}$ & $Q_n$ & $F_n$ & $Q_n^{(1)}$ & $F_n^{(1)}$ & $Q_n^{(2)}$ & $F_n^{(2)}$ \\ \hline
		\multicolumn{1}{|c|}{\multirow{6}{*}{1}}                                              & \multicolumn{1}{c|}{0}          & 0.3        & 12.7  & 10.3  & 11          & 9.3         & 14.9        & 9.7         & 8.8   & 7.2   & 8.2         & 7.2         & 10.9        & 5.9         \\ \cline{2-15} 
		\multicolumn{1}{|c|}{}                                                                & \multicolumn{1}{c|}{0.3}        & 0.3        & 12.8  & 14.9  & 10.1        & 11.6        & 15.3        & 10.1        & 7.5   & 9.8   & 7.8         & 8.7         & 10.8        & 7.2         \\ \cline{2-15} 
		\multicolumn{1}{|c|}{}                                                                & \multicolumn{1}{c|}{0.6}        & 0.6        & 38.4  & 45.2  & 28.6        & 35.4        & 30.3        & 20.8        & 15.5  & 20.9  & 12.5        & 15.5        & 14          & 11.4        \\ \cline{2-15} 
		\multicolumn{1}{|c|}{}                                                                & \multicolumn{1}{c|}{0.9}        & 0.9        & 63.8  & 71.2  & 47.4        & 54.7        & 45.7        & 36.8        & 29.2  & 37.1  & 19          & 28.1        & 21.3        & 17.8        \\ \cline{2-15} 
		\multicolumn{1}{|c|}{}                                                                & \multicolumn{1}{c|}{0.3}        & 0.6        & 28.7  & 29.4  & 20.9        & 21.9        & 24.4        & 15.5        & 13    & 14.7  & 11.4        & 11.6        & 13.3        & 10.1        \\ \cline{2-15} 
		\multicolumn{1}{|c|}{}                                                                & \multicolumn{1}{c|}{0.3}        & 0.9        & 49.2  & 48.7  & 35.3        & 34.9        & 34.3        & 24          & 25.7  & 24.4  & 20.2        & 18.8        & 19.8        & 12.4        \\ \hline
		\multicolumn{1}{|c|}{\multirow{6}{*}{2}}                                              & \multicolumn{1}{c|}{0}          & 0.3        & 15.2  & 12.2  & 10.5        & 6.7         & 13.7        & 11.2        & 10    & 8.1   & 9.2         & 8.5         & 8.8         & 6.5         \\ \cline{2-15} 
		\multicolumn{1}{|c|}{}                                                                & \multicolumn{1}{c|}{0.3}        & 0.3        & 17.5  & 17.9  & 9.7         & 8.1         & 20.6        & 17.2        & 8.4   & 9.5   & 8.7         & 7.3         & 8.4         & 9.2         \\ \cline{2-15} 
		\multicolumn{1}{|c|}{}                                                                & \multicolumn{1}{c|}{0.6}        & 0.6        & 49.5  & 53.9  & 17.9        & 16.4        & 52.2        & 50.7        & 19.9  & 22.8  & 11.9        & 12.1        & 19.8        & 21.6        \\ \cline{2-15} 
		\multicolumn{1}{|c|}{}                                                                & \multicolumn{1}{c|}{0.9}        & 0.9        & 78.7  & 83.4  & 24.6        & 25.7        & 76.9        & 76.9        & 35.8  & 42.9  & 15.1        & 15.5        & 32.6        & 37.8        \\ \cline{2-15} 
		\multicolumn{1}{|c|}{}                                                                & \multicolumn{1}{c|}{0.3}        & 0.6        & 38.6  & 40.7  & 16.5        & 13.9        & 40.8        & 37.8        & 14.8  & 17.4  & 9.8         & 9.1         & 15.3        & 14.8        \\ \cline{2-15} 
		\multicolumn{1}{|c|}{}                                                                & \multicolumn{1}{c|}{0.3}        & 0.9        & 63.6  & 62.9  & 23.7        & 19          & 61.5        & 59          & 31    & 30.6  & 14.3        & 12.6        & 27.4        & 25.8        \\ \hline
		\multicolumn{1}{|c|}{\multirow{6}{*}{3}}                                              & \multicolumn{1}{c|}{0}          & 0.3        & 15.3  & 11.6  & 13          & 10.9        & 13.9        & 11.6        & 8.3   & 7.1   & 9.4         & 8.5         & 9.5         & 6.7         \\ \cline{2-15} 
		\multicolumn{1}{|c|}{}                                                                & \multicolumn{1}{c|}{0.3}        & 0.3        & 17    & 21.1  & 11.8        & 12.7        & 19.9        & 17.4        & 7.8   & 9.6   & 10.1        & 10.3        & 9.9         & 9.1         \\ \cline{2-15} 
		\multicolumn{1}{|c|}{}                                                                & \multicolumn{1}{c|}{0.6}        & 0.6        & 48.7  & 57.7  & 18.7        & 23.3        & 53          & 52.9        & 16.5  & 22    & 13.2        & 13.9        & 18.9        & 20.6        \\ \cline{2-15} 
		\multicolumn{1}{|c|}{}                                                                & \multicolumn{1}{c|}{0.9}        & 0.9        & 79    & 87.2  & 30.1        & 34.3        & 75.8        & 77.5        & 34.3  & 42.9  & 15.8        & 17.9        & 33.9        & 39          \\ \cline{2-15} 
		\multicolumn{1}{|c|}{}                                                                & \multicolumn{1}{c|}{0.3}        & 0.6        & 38.5  & 40    & 15.9        & 15.9        & 41.6        & 38.5        & 16.6  & 18    & 12.2        & 11.3        & 16.6        & 16.1        \\ \cline{2-15} 
		\multicolumn{1}{|c|}{}                                                                & \multicolumn{1}{c|}{0.3}        & 0.9        & 62.4  & 63.2  & 28.2        & 23          & 58.4        & 56.6        & 30.5  & 27.4  & 16          & 12.7        & 28.3        & 26.7        \\ \hline
		\multicolumn{1}{|c|}{\multirow{6}{*}{4}}                                              & \multicolumn{1}{c|}{0}          & 0.3        & 11.5  & 8.6   & 11.5        & 7.7         & 14.6        & 8.4         & 9.1   & 7.1   & 9           & 6.6         & 10          & 6.7         \\ \cline{2-15} 
		\multicolumn{1}{|c|}{}                                                                & \multicolumn{1}{c|}{0.3}        & 0.3        & 15.3  & 14.1  & 16.1        & 11.7        & 16.6        & 9.4         & 8.9   & 9     & 10.2        & 8.4         & 11.2        & 8.9         \\ \cline{2-15} 
		\multicolumn{1}{|c|}{}                                                                & \multicolumn{1}{c|}{0.6}        & 0.6        & 34.4  & 36.5  & 25          & 22.5        & 32.3        & 23.6        & 15.7  & 16.4  & 13.1        & 11.7        & 13.8        & 10.2        \\ \cline{2-15} 
		\multicolumn{1}{|c|}{}                                                                & \multicolumn{1}{c|}{0.9}        & 0.9        & 54.1  & 58.4  & 37.3        & 33.7        & 48.2        & 38.8        & 27.8  & 31.1  & 21.2        & 20.6        & 21.2        & 17.9        \\ \cline{2-15} 
		\multicolumn{1}{|c|}{}                                                                & \multicolumn{1}{c|}{0.3}        & 0.6        & 26.4  & 24.9  & 20.7        & 16.4        & 23.3        & 15.1        & 13.3  & 13.4  & 11.2        & 10.1        & 14          & 10.3        \\ \cline{2-15} 
		\multicolumn{1}{|c|}{}                                                                & \multicolumn{1}{c|}{0.3}        & 0.9        & 40.7  & 38.7  & 28.8        & 22.5        & 35.8        & 25.8        & 23.6  & 22.8  & 17.1        & 14.9        & 20.6        & 15.8        \\ \hline
	\end{tabular}
\label{Table:Power-cauchy}
\end{table}

\subsection{Multiple Imputation}
\label{sec:MultipleImputation}
Another common method for handling missing data is Multiple Imputation introduced in Rubin \cite{Rubin:1987}. In the context of our problem, missing or deficient values are replaced with  two or more acceptable values generated from a predictive distribution. Hotelling's two sample $T^2$ test will be conducted on each of the completed data to test equality of mean vectors in the two groups, i.e. $H_0:\bm{\mu}_1=\bm{\mu}_2$ vs $H_1:\bm{\mu}_1\ne\bm{\mu}_2$. We refer to this test procedure as Imputation. However, due to the computational cost  of Imputation, only a small-scale simulation is conducted to compare its performance with the methods introduced in this paper. In the covariance matrix, $(\rho_1,\rho_2,\rho_{12})=(-0.1,-0.1,-0.1)$ or $(0.1,0.1,0.1)$ and $(\sigma_1,\sigma_2)=(1,1)$ or $(1,5)$ are used.        	
The achieved Type-I error rates and powers from Imputation are displayed in Tables \ref{Table:Multiple-Typeone} and \ref{Table:Multiple-power}. In Table \ref{Table:Multiple-Typeone}, for discretized multivariate normal distribution or multivariate log-normal distribution with homoscedasticity, $F_n$ and Imputation achieve similar Type-I error rates. However, when data are generated from multivariate log-normal distribution with heteroscedasticity, the achieved Type-I error rates are too liberal and nearly 1. Also, for data generated from multivariate Cauchy distribution, the achieved Type-I error rates are too conservative compared to $F_n$. Overall, Imputation achieves smaller powers compared with $Q_n$ and $F_n$. Specifically, for data that are generated from multivariate Cauchy distribution, the achieved powers are way too small.

\begin{table}[!htb]
	\centering
	\caption{Achieved Type-I error rate ($\times100$) for data generated from Discretized Multivariate Normal, Multivariate Log-Normal and Multivariate Cauchy distributions with missing pattern as in Table \ref{Table:schematic} for $d=2,3,5$. Here, $Q_n$ is the Wald-type statistic proposed in (\ref{Q_N}); $F_n$ is ANOVA-type statistic for small sample approximation proposed in (\ref{ANOVA}); Multiple Impute is the test procedure in Imputation alternative method with 5 chains. The nominal Type-I error rate is $\alpha=0.05$ and sample sizes are as in Setting 1.}
	\begin{tabular}{cllclc|c|c|c|c|c|c|c|c|c|}
		\cline{7-15}
		\multicolumn{6}{c|}{} & \multicolumn{3}{c|}{\begin{tabular}[c]{@{}c@{}}Discretized\\ Multivariate\\ Normal\end{tabular}} & \multicolumn{3}{c|}{\begin{tabular}[c]{@{}c@{}}Multivariate\\ Log-Normal\end{tabular}} & \multicolumn{3}{c|}{\begin{tabular}[c]{@{}c@{}}Multivariate\\ Cauchy\end{tabular}} \\ \hline
		\multicolumn{3}{|c|}{$(\rho_1,\rho_2,\rho_{12})$} & \multicolumn{2}{c|}{$(\sigma_1^2,\sigma_2^2)$} & $d$ & $Q_n$ & $F_n$ & \begin{tabular}[c]{@{}c@{}}Multiple\\ Impute\end{tabular} & $Q_n$ & $F_n$ & \begin{tabular}[c]{@{}c@{}}Multiple\\ Impute\end{tabular} & $Q_n$ & $F_n$ & \begin{tabular}[c]{@{}c@{}}Multiple\\ Impute\end{tabular} \\ \hline
		\multicolumn{3}{|c|}{\multirow{6}{*}{(-0.1,-0.1,-0.1)}} & \multicolumn{2}{c|}{\multirow{3}{*}{(1,1)}} & 2 & 5.2 & 4.2 & 4.6 & 7.2 & 6.3 & 5.6  & 5.4  & 4.4 & 2.4 \\ \cline{6-15} 
		\multicolumn{3}{|c|}{}                                  & \multicolumn{2}{c|}{}                       & 3 & 7.5 & 5.7 & 5.9 & 7.1 & 5.3 & 5.8  & 7.3  & 5.3 & 2.1 \\ \cline{6-15} 
		\multicolumn{3}{|c|}{}                                  & \multicolumn{2}{c|}{}                       & 5 & 12  & 7.1 & 5.3 & 8.0   & 4.8 & 5.0    & 7.9  & 4.0   & 1.1 \\ \cline{4-15} 
		\multicolumn{3}{|c|}{}                                  & \multicolumn{2}{c|}{\multirow{3}{*}{(1,5)}} & 2 & 7.2 & 6.8 & 6.1 & 6.5 & 5.3 & 73.4 & 7.0    & 5.8 & 2.9 \\ \cline{6-15} 
		\multicolumn{3}{|c|}{}                                  & \multicolumn{2}{c|}{}                       & 3 & 7.1 & 5.7 & 4.8 & 5.4 & 4.3 & 85.7 & 7.8  & 5.8 & 2.7 \\ \cline{6-15} 
		\multicolumn{3}{|c|}{}                                  & \multicolumn{2}{c|}{}                       & 5 & 6.6 & 3.8 & 4.5 & 9.2 & 5.0   & 97.3 & 9.1  & 5.4 & 1.0   \\ \hline
		\multicolumn{3}{|c|}{\multirow{6}{*}{(0.1,0.1,0.1)}}    & \multicolumn{2}{c|}{\multirow{3}{*}{(1,1)}} & 2 & 5.5 & 5.3 & 4.7 & 7.1 & 6.3 & 6.0    & 5.5  & 4.3 & 3.2 \\ \cline{6-15} 
		\multicolumn{3}{|c|}{}                                  & \multicolumn{2}{c|}{}                       & 3 & 7.2 & 4.9 & 3.6 & 6.4 & 4.3 & 6.0    & 6.6  & 4.5 & 2.1 \\ \cline{6-15} 
		\multicolumn{3}{|c|}{}                                  & \multicolumn{2}{c|}{}                       & 5 & 9.7 & 6.0   & 4.0   & 9.6 & 5.4 & 6.3  & 10.2 & 5.3 & 0.8 \\ \cline{4-15} 
		\multicolumn{3}{|c|}{}                                  & \multicolumn{2}{c|}{\multirow{3}{*}{(1,5)}} & 2 & 5.1 & 4.7 & 5.6 & 7.2 & 6.3 & 72.2 & 6.4  & 5.4 & 2.8 \\ \cline{6-15} 
		\multicolumn{3}{|c|}{}                                  & \multicolumn{2}{c|}{}                       & 3 & 8.0   & 6.1 & 6.3 & 6.3 & 3.7 & 87   & 8.2  & 6.7 & 2.4 \\ \cline{6-15} 
		\multicolumn{3}{|c|}{}                                  & \multicolumn{2}{c|}{}                       & 5 & 8.8 & 5.9 & 4.8 & 8.9 & 5.0   & 95   & 8.9  & 4.8 & 1.1 \\ \hline
	\end{tabular}
	\label{Table:Multiple-Typeone}
\end{table}

\begin{table}[!htb]
	\centering
	\caption{Achieved power ($\times100$) for data generated from Discretized Multivariate Normal, Multivariate Log-Normal and Multivariate Cauchy distributions with missing pattern as in Table \ref{Table:schematic} for $d=2,3,5$. Here, $Q_n$ is the Wald-type statistic proposed in (\ref{Q_N}); $F_n$ is ANOVA-type statistic for small sample approximation proposed in (\ref{ANOVA}); Multiple Impute is the test procedure in Imputation alternative method with 5 chains. The nominal Type-I error rate is $\alpha=0.05$ and sample sizes are as in Setting 1.}
	\begin{tabular}{ccccccc|c|c|c|c|c|c|c|c|c|}
		\cline{8-16}
		&                  &                  &                        &                       &                                 &            & \multicolumn{3}{c|}{\begin{tabular}[c]{@{}c@{}}Discretized\\ Multivariate\\ Normal\end{tabular}} & \multicolumn{3}{c|}{\begin{tabular}[c]{@{}c@{}}Multivariate\\ Log-Normal\end{tabular}} & \multicolumn{3}{c|}{\begin{tabular}[c]{@{}c@{}}Multivariate\\ Cauchy\end{tabular}} \\ \hline
		\multicolumn{3}{|c|}{$(\rho_1,\rho_2,\rho_{12})$}        & \multicolumn{2}{c|}{$(\sigma_1^2,\sigma_2^2)$} & \multicolumn{1}{c|}{$\delta_1$} & $\delta_2$ & $Q_n$         & $F_n$         & \begin{tabular}[c]{@{}c@{}}Multiple\\ Impute\end{tabular}        & $Q_n$      & $F_n$     & \begin{tabular}[c]{@{}c@{}}Multiple\\ Impute\end{tabular}     & $Q_n$    & $F_n$    & \begin{tabular}[c]{@{}c@{}}Multiple\\ Impute\end{tabular}    \\ \hline
		\multicolumn{3}{|c|}{\multirow{12}{*}{(-0.1,-0.1,-0.1)}} & \multicolumn{2}{c|}{\multirow{6}{*}{(1,1)}}    & \multicolumn{1}{c|}{0}          & 0.3        & 20            & 18.1          & 9                                                                & 19.9       & 17        & Impute                                                        & 10       & 8.9      & 2.7                                                          \\ \cline{6-16} 
		\multicolumn{3}{|c|}{}                                   & \multicolumn{2}{c|}{}                          & \multicolumn{1}{c|}{0.3}        & 0.3        & 40            & 35            & 16.1                                                             & 40.3       & 35.9      & 17.8                                                          & 17.3     & 14.9     & 3.5                                                          \\ \cline{6-16} 
		\multicolumn{3}{|c|}{}                                   & \multicolumn{2}{c|}{}                          & \multicolumn{1}{c|}{0.6}        & 0.6        & 93            & 91.1          & 58.7                                                             & 94.2       & 93        & 45.5                                                          & 44.8     & 42.5     & 3.9                                                          \\ \cline{6-16} 
		\multicolumn{3}{|c|}{}                                   & \multicolumn{2}{c|}{}                          & \multicolumn{1}{c|}{0.9}        & 0.9        & 99.9          & 99.9          & 89.4                                                             & 100        & 100       & 72.9                                                          & 75.6     & 73.3     & 7.2                                                          \\ \cline{6-16} 
		\multicolumn{3}{|c|}{}                                   & \multicolumn{2}{c|}{}                          & \multicolumn{1}{c|}{0.3}        & 0.6        & 77.7          & 75            & 40.2                                                             & 77.3       & 73.8      & 30.8                                                          & 31.9     & 28.6     & 3.2                                                          \\ \cline{6-16} 
		\multicolumn{3}{|c|}{}                                   & \multicolumn{2}{c|}{}                          & \multicolumn{1}{c|}{0.3}        & 0.9        & 97.2          & 96.7          & 68.9                                                             & 97.5       & 97.1      & 48.4                                                          & 52.7     & 48.9     & 5.7                                                          \\ \cline{4-16} 
		\multicolumn{3}{|c|}{}                                   & \multicolumn{2}{c|}{\multirow{6}{*}{(1,5)}}    & \multicolumn{1}{c|}{0}          & 0.3        & 10.9          & 9.1           & 6.4                                                              & 13.2       & 12        & 31.1                                                          & 7.5      & 6.1      & 2.2                                                          \\ \cline{6-16} 
		\multicolumn{3}{|c|}{}                                   & \multicolumn{2}{c|}{}                          & \multicolumn{1}{c|}{0.3}        & 0.3        & 17.4          & 14            & 9.6                                                              & 15.3       & 13        & 31.5                                                          & 9.7      & 8.7      & 3                                                            \\ \cline{6-16} 
		\multicolumn{3}{|c|}{}                                   & \multicolumn{2}{c|}{}                          & \multicolumn{1}{c|}{0.6}        & 0.6        & 47.3          & 43.2          & 21.9                                                             & 47.6       & 44.8      & 37.4                                                          & 21.4     & 18.4     & 3.2                                                          \\ \cline{6-16} 
		\multicolumn{3}{|c|}{}                                   & \multicolumn{2}{c|}{}                          & \multicolumn{1}{c|}{0.9}        & 0.9        & 80.1          & 76.8          & 46.5                                                             & 81.8       & 78.2      & 44.4                                                          & 40.8     & 36.2     & 4.6                                                          \\ \cline{6-16} 
		\multicolumn{3}{|c|}{}                                   & \multicolumn{2}{c|}{}                          & \multicolumn{1}{c|}{0.3}        & 0.6        & 32.2          & 28            & 13.3                                                             & 31.2       & 27.9      & 33.5                                                          & 15.8     & 13.2     & 3.1                                                          \\ \cline{6-16} 
		\multicolumn{3}{|c|}{}                                   & \multicolumn{2}{c|}{}                          & \multicolumn{1}{c|}{0.3}        & 0.9        & 55.9          & 52.9          & 27.2                                                             & 58.4       & 54.5      & 37.8                                                          & 27.3     & 22.7     & 4.3                                                          \\ \hline
		\multicolumn{3}{|c|}{\multirow{12}{*}{(0.1,0.1,0.1)}}    & \multicolumn{2}{c|}{\multirow{6}{*}{(1,1)}}    & \multicolumn{1}{c|}{0}          & 0.3        & 22.2          & 20.3          & 10.8                                                             & 21.8       & 20        & 11.3                                                          & 10       & 8.7      & 2.3                                                          \\ \cline{6-16} 
		\multicolumn{3}{|c|}{}                                   & \multicolumn{2}{c|}{}                          & \multicolumn{1}{c|}{0.3}        & 0.3        & 36            & 35            & 15.7                                                             & 38.1       & 38.1      & 16.1                                                          & 16.3     & 15.5     & 2.2                                                          \\ \cline{6-16} 
		\multicolumn{3}{|c|}{}                                   & \multicolumn{2}{c|}{}                          & \multicolumn{1}{c|}{0.6}        & 0.6        & 90            & 90.6          & 51.6                                                             & 90.3       & 90.4      & 41.2                                                          & 41.9     & 43       & 4                                                            \\ \cline{6-16} 
		\multicolumn{3}{|c|}{}                                   & \multicolumn{2}{c|}{}                          & \multicolumn{1}{c|}{0.9}        & 0.9        & 100           & 100           & 89                                                               & 99.9       & 99.9      & 72.1                                                          & 70.1     & 72       & 7.4                                                          \\ \cline{6-16} 
		\multicolumn{3}{|c|}{}                                   & \multicolumn{2}{c|}{}                          & \multicolumn{1}{c|}{0.3}        & 0.6        & 70            & 69            & 35.8                                                             & 75.8       & 76        & 30.1                                                          & 30.6     & 30.5     & 4.3                                                          \\ \cline{6-16} 
		\multicolumn{3}{|c|}{}                                   & \multicolumn{2}{c|}{}                          & \multicolumn{1}{c|}{0.3}        & 0.9        & 95.3          & 95.2          & 63.6                                                             & 97.1       & 97.1      & 50.8                                                          & 49.8     & 49.2     & 4.6                                                          \\ \cline{4-16} 
		\multicolumn{3}{|c|}{}                                   & \multicolumn{2}{c|}{\multirow{6}{*}{(1,5)}}    & \multicolumn{1}{c|}{0}          & 0.3        & 12.4          & 10.9          & 6.6                                                              & 8.6        & 7.2       & 30.8                                                          & 8.9      & 7.9      & 1.8                                                          \\ \cline{6-16} 
		\multicolumn{3}{|c|}{}                                   & \multicolumn{2}{c|}{}                          & \multicolumn{1}{c|}{0.3}        & 0.3        & 16.3          & 16.2          & 9.1                                                              & 15.5       & 14.3      & 31.3                                                          & 10.2     & 9.1      & 2.7                                                          \\ \cline{6-16} 
		\multicolumn{3}{|c|}{}                                   & \multicolumn{2}{c|}{}                          & \multicolumn{1}{c|}{0.6}        & 0.6        & 45.9          & 45.1          & 20.3                                                             & 44.5       & 44.9      & 38.4                                                          & 20.3     & 19.2     & 2.3                                                          \\ \cline{6-16} 
		\multicolumn{3}{|c|}{}                                   & \multicolumn{2}{c|}{}                          & \multicolumn{1}{c|}{0.9}        & 0.9        &    76.9      &         77.3      &                             41.5                          & 78.5       & 78.9      & 42.7                                                          & 36.6     & 37.7     & 4.3                                                          \\ \cline{6-16} 
		\multicolumn{3}{|c|}{}                                   & \multicolumn{2}{c|}{}                          & \multicolumn{1}{c|}{0.3}        & 0.6        &       30.4        &       28.8        &                     13.6                                 & 30.9       & 30.5      & 34.1                                                          & 17.2     & 15.6     & 2.7                                                          \\ \cline{6-16} 
		\multicolumn{3}{|c|}{}                                   & \multicolumn{2}{c|}{}                          & \multicolumn{1}{c|}{0.3}        & 0.9        &      53.8         &       52.4        &                        23.5                               & 51.1       & 50.5      & 37.8                                                          & 25.6     & 24.1     & 3.5                                                          \\ \hline
	\end{tabular}
\label{Table:Multiple-power}
\end{table}

\subsection{General Missing Pattern}
\label{sec:GeneralMissingPattern}
{For multivariate data that have general missing structures, we anticipate performance to be affected by sample size allocations. We consider three allocations which are aimed to cover practical situations, and we refer to them as Design 1-3.
	\begin{description} 
		\item [Design 1:] Fix the total sample size $n$ and assign $n/15$ subjects per missing pattern.
		\item [Design 2:] Fix the total sample size $n$ and vary the proportion of complete cases $a$ such that $n_1=na$ (complete sample size) and $n_2=\cdots=n_{15}=n(1-a)/14$.
		\item [Design 3:] Fix $n_2,\cdots,n_{15}=100$ (i.e. large number of incomplete cases) and vary the complete sample size $n_1$.
	\end{description}
	
The achieved Type-I error rates and powers are shown in Tables \ref{Table:flex-design1-typeone}-\ref{Table:flex-design3-typeone}  and Tables \ref{Table:flex-design1-power}-\ref{Table:flex-design3-power}, respectively. The performance of the methods under Discretized Multivariate Normal distribution and Multivariate Log-Normal distribution are similar.

In Table \ref{Table:flex-design1-typeone}, the ANOVA-type statistic preserves preassigned significance level very well in all settings, while Wald-type statistic tends to be slightly liberal. In Table \ref{Table:flex-design2-typeone} and \ref{Table:flex-design3-typeone}, performance of Wald-type and ANOVA-type statistics are very close and they both control preassigned significance level well. Furthermore, their performances are not affected by the distributions, covariance structures or sample sizes allocations.

The power reacts are consistent with Section \ref{sec:PowerStudy}. The ANOVA-type statistics tend to achieve higher powers than the Wald-type statistics and, further, homogeneous data yield more powers compared to heterogeneous data. Furthermore, we see that powers generally increase with the total sample size $n$, the proportion of complete cases (in Table \ref{Table:flex-design2-power}) and the complete sample size (in Table \ref{Table:flex-design3-power}). It also worth to mention that the results in Table \ref{Table:flex-design3-power} are quite favorable in the sense that most of them approach unity even when $\delta$ is small. }

\begin{table}[!htb]
	\caption{Achieved Type-I error rate ($\times100$) for data generated from Discretized Multivariate Normal, Multivariate Log-Normal and Multivariate Cauchy distributions with general missing pattern for $d=2$. Here, $n$ is the total sample size in Design 1; $Q_n$ is the Wald-type statistic proposed in (\ref{Q_N}); $F_n$ is ANOVA-type statistic for small sample approximation proposed in (\ref{ANOVA}). The nominal Type-I error rate is $\alpha=0.05$.}
	\centering
	\begin{tabular}{|c|c|l|l|c|l|c|c|c|c|c|c|}
		\hline
		\multicolumn{6}{|c|}{\textbf{Design 1}} & \multicolumn{2}{c|}{\begin{tabular}[c]{@{}c@{}}Discretized\\ Multivariate\\ Normal\end{tabular}} & \multicolumn{2}{c|}{\begin{tabular}[c]{@{}c@{}}Multivariate\\ Log-Normal\end{tabular}} & \multicolumn{2}{c|}{\begin{tabular}[c]{@{}c@{}}Multivariate\\ Cauchy\end{tabular}} \\ \hline
		$n$ & \multicolumn{3}{c|}{$(\rho_1,\rho_2,\rho_{12})$} & \multicolumn{2}{c|}{$(\sigma_1,\sigma_2)$} & $Q_n$ & $F_n$ & $Q_n$ & $F_n$ & $Q_n$ & $F_n$ \\ \hline
		\multirow{4}{*}{75} & \multicolumn{3}{c|}{\multirow{2}{*}{(-0.1,-0.1,-0.1)}} & \multicolumn{2}{c|}{(1,1)} & 6.2 & 4.9 & 6.4 & 5.6 & 6.1 & 4.6 \\ \cline{5-12} 
		& \multicolumn{3}{c|}{} & \multicolumn{2}{c|}{(1,5)} & 5.9 & 5.5 & 5.4 & 5.0 & 6.5 & 5.6 \\ \cline{2-12} 
		& \multicolumn{3}{c|}{\multirow{2}{*}{(0.1,0.1,0.1)}} & \multicolumn{2}{c|}{(1,1)} & 7.7 & 6.1 & 5.5 & 4.4 & 6.7 & 6.0 \\ \cline{5-12} 
		& \multicolumn{3}{c|}{} & \multicolumn{2}{c|}{(1,5)} & 7.6 & 6.1 & 6.5 & 5.2 & 6.7 & 6.3 \\ \hline
		\multirow{4}{*}{150} & \multicolumn{3}{c|}{\multirow{2}{*}{(-0.1,-0.1,-0.1)}} & \multicolumn{2}{c|}{(1,1)} & 7.1 & 6.5 & 5.4 & 5.0   & 6.2 & 5.6 \\ \cline{5-12} 
		& \multicolumn{3}{c|}{} & \multicolumn{2}{c|}{(1,5)} & 6.3 & 5.7 & 5.9 & 5.5 & 5.8 & 5.4 \\ \cline{2-12} 
		& \multicolumn{3}{c|}{\multirow{2}{*}{(0.1,0.1,0.1)}} & \multicolumn{2}{c|}{(1,1)} & 5.4 & 4.7 & 5.5 & 5.3 & 4.9 & 3.4\\ \cline{5-12} 
		& \multicolumn{3}{c|}{} & \multicolumn{2}{c|}{(1,5)} & 5.7 & 5.0   & 5.6 & 4.8 & 4.9 & 4.6 \\ \hline
		\multirow{4}{*}{300} & \multicolumn{3}{c|}{\multirow{2}{*}{(-0.1,-0.1,-0.1)}} & \multicolumn{2}{c|}{(1,1)} & 5.5 & 5.5 & 5.9 & 5.4 & 5.1 & 4.8 \\ \cline{5-12} 
		& \multicolumn{3}{c|}{} & \multicolumn{2}{c|}{(1,5)} & 7.1 & 7.2 & 4.7 & 4.5 & 4.9 & 4.3 \\ \cline{2-12} 
		& \multicolumn{3}{c|}{\multirow{2}{*}{(0.1,0.1,0.1)}} & \multicolumn{2}{c|}{(1,1)} & 5.6 & 5.4 & 4.9 & 4.5 & 4.7 & 4.5 \\ \cline{5-12} 
		& \multicolumn{3}{c|}{} & \multicolumn{2}{c|}{(1,5)} & 3.6 & 3.6 & 4.5 & 3.9 & 8.0   & 7.5 \\ \hline
	\end{tabular}
	\label{Table:flex-design1-typeone}
\end{table}

\begin{table}[htb]
	\centering
	\caption{Achieved Type-I error rate ($\times100$) for data generated from Discretized Multivariate Normal, Multivariate Log-Normal and Multivariate Cauchy distributions with general missing pattern for $d=2$. Here, $n$ is the total sample size in Design 2; $Q_n$ is the Wald-type statistic proposed in (\ref{Q_N}); $F_n$ is ANOVA-type statistic for small sample approximation proposed in (\ref{ANOVA}). The nominal Type-I error rate is $\alpha=0.05$.}
	\begin{tabular}{|c|c|c|l|c|c|c|c|c|c|c|}
		\hline
		\multicolumn{5}{|c|}{\textbf{Design 2}} & \multicolumn{2}{c|}{\begin{tabular}[c]{@{}c@{}}Discretized\\ Multivariate\\ Normal\end{tabular}} & \multicolumn{2}{c|}{\begin{tabular}[c]{@{}c@{}}Multivariate\\ Log-Normal\end{tabular}} & \multicolumn{2}{c|}{\begin{tabular}[c]{@{}c@{}}Multivariate\\ Cauchy\end{tabular}} \\ \hline
		$n$ & $(\rho_1,\rho_2,\rho_{12})$ & \multicolumn{2}{c|}{$(\sigma ^2_1,\sigma^ 2_2)$} & $a$ & $Q_n$ & $F_n$ & $Q_n$ & $F_n$ & $Q_n$ & $F_n$ \\ \hline
		\multirow{16}{*}{210} & \multirow{8}{*}{(-0.1,-0.1-0.1)} & \multicolumn{2}{c|}{\multirow{4}{*}{(1,1)}} & 0.2 & 5.6 & 4.9 & 5.1 & 4.9 & 4.4 & 4.6 \\ \cline{5-11} 
		&                                  & \multicolumn{2}{c|}{}                       & 0.4 & 5.1 & 4.9 & 6.2 & 6.2 & 4.8 & 5.0   \\ \cline{5-11} 
		&                                  & \multicolumn{2}{c|}{}                       & 0.6 & 5.8 & 5.4 & 5.1 & 5.0   & 5.7 & 5.7 \\ \cline{5-11} 
		&                                  & \multicolumn{2}{c|}{}                       & 0.8 & 5.5 & 5.1 & 4.6 & 4.1 & 5.8 & 5.6 \\ \cline{3-11} 
		&                                  & \multicolumn{2}{c|}{\multirow{4}{*}{(1,5)}} & 0.2 & 5.2 & 5.3 & 4.3 & 4.1 & 7.1 & 6.3 \\ \cline{5-11} 
		&                                  & \multicolumn{2}{c|}{}                       & 0.4 & 5   & 4.5 & 4.6 & 4.0   & 5.3 & 5.2 \\ \cline{5-11} 
		&                                  & \multicolumn{2}{c|}{}                       & 0.6 & 6.2 & 6.2 & 6.2 & 5.6 & 4.4 & 4.4 \\ \cline{5-11} 
		&                                  & \multicolumn{2}{c|}{}                       & 0.8 & 6.4 & 6.2 & 5.9 & 6.0   & 5.6 & 5.2 \\ \cline{2-11} 
		& \multirow{8}{*}{(0.1,0.1,0.1)}   & \multicolumn{2}{c|}{\multirow{4}{*}{(1,1)}} & 0.2 & 4.9 & 4.7 & 5.2 & 4.8 & 4.9 & 4.8 \\ \cline{5-11} 
		&                                  & \multicolumn{2}{c|}{}                       & 0.4 & 5.4 & 5.0   & 5.9 & 4.9 & 5.9 & 5.5 \\ \cline{5-11} 
		&                                  & \multicolumn{2}{c|}{}                       & 0.6 & 5.5 & 5.4 & 7.0   & 6.4 & 6.4 & 5.9 \\ \cline{5-11} 
		&                                  & \multicolumn{2}{c|}{}                       & 0.8 & 5.4 & 4.9 & 4.8 & 4.5 & 5.3 & 4.7 \\ \cline{3-11} 
		&                                  & \multicolumn{2}{c|}{\multirow{4}{*}{(1,5)}} & 0.2 & 5.0   & 4.8 & 5.9 & 5.5 & 5.4 & 5.3 \\ \cline{5-11} 
		&                                  & \multicolumn{2}{c|}{}                       & 0.4 & 6.5 & 5.9 & 6.7 & 5.6 & 4.5 & 4.4 \\ \cline{5-11} 
		&                                  & \multicolumn{2}{c|}{}                       & 0.6 & 5.0   & 4.8 & 4.6 & 4.4 & 6.9 & 6.1 \\ \cline{5-11} 
		&                                  & \multicolumn{2}{c|}{}                       & 0.8 & 4.8 & 5.1 & 4.5 & 4.7 & 4.9 & 4.4 \\ \hline
		\end{tabular}
	\label{Table:flex-design2-typeone}
\end{table}

\begin{table}[htb]
	\caption{Achieved Type-I error rate ($\times100$) for data generated from Discretized Multivariate Normal, Multivariate Log-Normal and Multivariate Cauchy distributions with general missing pattern for $d=2$. Here, $n_1$ is the complete sample size in Design 3; $Q_n$ is the Wald-type statistic proposed in (\ref{Q_N}); $F_n$ is ANOVA-type statistic for small sample approximation proposed in (\ref{ANOVA}). The nominal Type-I error rate is $\alpha=0.05$.}
	\centering
	\begin{tabular}{|c|l|l|c|l|c|c|c|c|c|c|c|}
		\hline
		\multicolumn{6}{|c|}{\textbf{Design 3}} &\multicolumn{2}{c|}{\begin{tabular}[c]{@{}c@{}}Discretized\\ Multivariate\\ Normal\end{tabular}} & \multicolumn{2}{c|}{\begin{tabular}[c]{@{}c@{}}Multivariate\\ Log-Normal\end{tabular}} & \multicolumn{2}{c|}{\begin{tabular}[c]{@{}c@{}}Multivariate\\ Cauchy\end{tabular}} \\ \hline
		\multicolumn{3}{|c|}{($\rho_1$,$\rho_2$,$\rho_{12}$)} & \multicolumn{2}{c|}{($\sigma_1$,$\sigma_2$)} & $n_1$ & $Q_n$ & $F_n$ & $Q_n$ & $F_n$ & $Q_n$ & $F_n$ \\ \hline
		\multicolumn{3}{|c|}{\multirow{6}{*}{(-0.1,-0.1,-0.1)}} & \multicolumn{2}{c|}{\multirow{3}{*}{(1,1)}} & 5 & 6.6 & 6.3 & 4.8 & 4.6 & 5.4 & 5.6 \\ \cline{6-12} 
		\multicolumn{3}{|c|}{} & \multicolumn{2}{c|}{} & 10 & 6.6 & 6.6 & 6.7 & 6.7 & 4.8 & 4.6 \\ \cline{6-12} 
		\multicolumn{3}{|c|}{} & \multicolumn{2}{c|}{} & 20 & 4.3 & 4.1 & 4.3 & 4.3 & 5.6 & 5.6 \\ \cline{4-12} 
		\multicolumn{3}{|c|}{} & \multicolumn{2}{c|}{\multirow{3}{*}{(1,5)}} & 5   & 5.0 & 5.0   & 4.9 & 5.0   & 4.9 & 4.8 \\ \cline{6-12} 
		\multicolumn{3}{|c|}{} & \multicolumn{2}{c|}{} & 10 & 4.9 & 4.4 & 5.1 & 5.1 & 5.1 & 5.6 \\ \cline{6-12} 
		\multicolumn{3}{|c|}{} & \multicolumn{2}{c|}{} & 20 & 5.8 & 5.6 & 5.8 & 5.6 & 5.6 & 5.9 \\ \hline
		\multicolumn{3}{|c|}{\multirow{6}{*}{(0.1,0.1,0.1)}} & \multicolumn{2}{c|}{\multirow{3}{*}{(1,1)}} & 5 & 5.1 & 4.8 & 6.7 & 6.7 & 4.5 & 4.3 \\ \cline{6-12} 
		\multicolumn{3}{|c|}{} & \multicolumn{2}{c|}{} & 10 & 4.9 & 4.8 & 5.2 & 5.2 & 4.7 & 4.7 \\ \cline{6-12} 
		\multicolumn{3}{|c|}{} & \multicolumn{2}{c|}{} & 20 & 5.9 & 5.8 & 3.8 & 3.9 & 4.1 & 4.0 \\ \cline{4-12} 
		\multicolumn{3}{|c|}{} & \multicolumn{2}{c|}{\multirow{3}{*}{(1,5)}} & 5 & 5.2 & 4.6 & 6.4 & 6.1 & 5.2 & 5.2 \\ \cline{6-12} 
		\multicolumn{3}{|c|}{} & \multicolumn{2}{c|}{} & 10 & 6.4 & 6.6 & 5.8 & 5.5 & 4.2 & 4.4 \\ \cline{6-12} 
		\multicolumn{3}{|c|}{} & \multicolumn{2}{c|}{} & 20 & 4.2 & 4.1 & 6.3 & 6.2 & 5.4 & 5.2 \\ \hline
	\end{tabular}
	\label{Table:flex-design3-typeone}
\end{table}

\begin{table}[!htb]
		\caption{Achieved power ($\times100$) for data generated from Discretized Multivariate Normal, Multivariate Log-Normal and Multivariate Cauchy distribution with general missing pattern for $d=2$. Here, $n$ is the total sample size in Design 1; $Q_n$ is the Wald-type statistic proposed in (\ref{Q_N}); $F_n$ is ANOVA-type statistic for small sample approximation proposed in (\ref{ANOVA}). The correlation coefficients are $(\rho_1,\rho_2,\rho_{12})=(0.1,0.1,0.1)$. The nominal Type-I error rate is $\alpha=0.05$.}
	\centering
	\begin{tabular}{|c|c|l|c|c|c|c|c|c|c|c|}
		\hline
		\multicolumn{5}{|c|}{\textbf{Design 1}}                                                                                                                                     & \multicolumn{2}{c|}{$n=75$} & \multicolumn{2}{c|}{$n=150$} & \multicolumn{2}{c|}{$n=300$} \\ \hline
		Distribution                                                                                     & \multicolumn{2}{c|}{$(\sigma_1^2,\sigma_2^2)$} & $\delta_1$ & $\delta_2$ & $Q_n$        & $F_n$        & $Q_n$         & $F_n$        & $Q_n$         & $F_n$        \\ \hline
		\multirow{12}{*}{\begin{tabular}[c]{@{}c@{}}Discretized\\ Multivariate\\ Normal\end{tabular}} & \multicolumn{2}{c|}{\multirow{6}{*}{(1,1)}}    & 0          & 0.3        & 20.8         & 18.9         & 35.6          & 34.8         & 67.8          & 67.2         \\ \cline{4-11} 
		& \multicolumn{2}{c|}{}                          & 0.3        & 0.3        & 37.0         & 34.6         & 61.6          & 61.7         & 92.4          & 93.0           \\ \cline{4-11} 
		& \multicolumn{2}{c|}{}                          & 0.6        & 0.6        & 92.8         & 92.3         & 99.8          & 99.7         & 100           & 100          \\ \cline{4-11} 
		& \multicolumn{2}{c|}{}                          & 0.9        & 0.9        & 99.9         & 99.9         & 100           & 100          & 100           & 100          \\ \cline{4-11} 
		& \multicolumn{2}{c|}{}                          & 0.3        & 0.6        & 77.8         & 75.6         & 95.7          & 95.6         & 100           & 100          \\ \cline{4-11} 
		& \multicolumn{2}{c|}{}                          & 0.3        & 0.9        & 95.7         & 94.9         & 100           & 100          & 100           & 100          \\ \cline{2-11} 
		& \multicolumn{2}{c|}{\multirow{6}{*}{(1,5)}}    & 0          & 0.3        & 11.9         & 10.4         & 14.0          & 13.3         & 24.7          & 23.6         \\ \cline{4-11} 
		& \multicolumn{2}{c|}{}                          & 0.3        & 0.3        & 15.5         & 14.1         & 23.0          & 22.2         & 44.5          & 44.9         \\ \cline{4-11} 
		& \multicolumn{2}{c|}{}                          & 0.6        & 0.6        & 44.1         & 43.2         & 74.6          & 74.9         & 96.5          & 96.5         \\ \cline{4-11} 
		& \multicolumn{2}{c|}{}                          & 0.9        & 0.9        & 77.9         & 79.2         & 98.3          & 98.7         & 100           & 100          \\ \cline{4-11} 
		& \multicolumn{2}{c|}{}                          & 0.3        & 0.6        & 28.4         & 27.6         & 55.4          & 55.1         & 85.2          & 85.2         \\ \cline{4-11} 
		& \multicolumn{2}{c|}{}                          & 0.3        & 0.9        & 55.4         & 52.6         & 84.0          & 83.6         & 99.4          & 99.4         \\ \hline
		\multirow{12}{*}{\begin{tabular}[c]{@{}c@{}}Multivariate\\ Cauchy\end{tabular}}                  & \multicolumn{2}{c|}{\multirow{6}{*}{(1,1)}}    & 0          & 0.3        & 11.2         & 10.0         & 17.2          & 15.9         & 26.0          & 25.2         \\ \cline{4-11} 
		& \multicolumn{2}{c|}{}                          & 0.3        & 0.3        & 16.7         & 14.4         & 25.6          & 24.2         & 46.8          & 45.6         \\ \cline{4-11} 
		& \multicolumn{2}{c|}{}                          & 0.6        & 0.6        & 46.7         & 45.3         & 74.1          & 74.1         & 96.7          & 97.1         \\ \cline{4-11} 
		& \multicolumn{2}{c|}{}                          & 0.9        & 0.9        & 74.5         & 75.0         & 97.1          & 97.4         & 100           & 100          \\ \cline{4-11} 
		& \multicolumn{2}{c|}{}                          & 0.3        & 0.6        & 29.4         & 27.6         & 56.0          & 54.2         & 84.9          & 85.6         \\ \cline{4-11} 
		& \multicolumn{2}{c|}{}                          & 0.3        & 0.9        & 53.8         & 52.0         & 82.2          & 82.4         & 98.8          & 98.9         \\ \cline{2-11} 
		& \multicolumn{2}{c|}{\multirow{6}{*}{(1,5)}}    & 0          & 0.3        & 9.0          & 7.4          & 8.4           & 7.7          & 11.5          & 11.2         \\ \cline{4-11} 
		& \multicolumn{2}{c|}{}                          & 0.3        & 0.3        & 10.4         & 10.0         & 10.4          & 10.4         & 19.9          & 19.8         \\ \cline{4-11} 
		& \multicolumn{2}{c|}{}                          & 0.6        & 0.6        & 21.9         & 19.4         & 34.2          & 33.4         & 63.0          & 63.0         \\ \cline{4-11} 
		& \multicolumn{2}{c|}{}                          & 0.9        & 0.9        & 38.7         & 36.8         & 68.3          & 67.6         & 93.3          & 93.5         \\ \cline{4-11} 
		& \multicolumn{2}{c|}{}                          & 0.3        & 0.6        & 15.3         & 14.3         & 22.9          & 21.9         & 41.4          & 40.4         \\ \cline{4-11} 
		& \multicolumn{2}{c|}{}                          & 0.3        & 0.9        & 24.9         & 23.6         & 41.8          & 41.4         & 72.8          & 73.9         \\ \hline
	\end{tabular}
\label{Table:flex-design1-power}
\end{table}

\begin{table}[!htb]
	\caption{Achieved power ($\times100$) for data generated from Multivariate Normal, Multivariate Log-Normal and Multivariate Cauchy distribution with general missing pattern for $d=2$. Here, $n$ is the total sample size in Design 2; $Q_n$ is the Wald-type statistic proposed in (\ref{Q_N}); $F_n$ is ANOVA-type statistic for small sample approximation proposed in (\ref{ANOVA}). The correlation coefficients are $(\rho_1,\rho_2,\rho_{12})=(0.1,0.1,0.1)$. The nominal Type-I error rate is $\alpha=0.05$.}
	\centering
	\begin{tabular}{|c|c|c|c|c|c|c|c|c|c|c|c|c|c|}
		\hline
		 \multicolumn{6}{|c|}{\textbf{Design 2}}           & \multicolumn{2}{c|}{$a=0.2$} & \multicolumn{2}{c|}{$a=0.4$} & \multicolumn{2}{c|}{$a=0.6$} & \multicolumn{2}{c|}{$a=0.8$} \\ \hline
		$n$                   & Distribution                                                                                  & \multicolumn{2}{c|}{$(\sigma_1^2,\sigma_2^2)$} & $\delta_1$ & $\delta_2$ & $Q_n$         & $F_n$        & $Q_n$         & $F_n$        & $Q_n$         & $F_n$        & $Q_n$         & $F_n$        \\ \hline
		\multirow{36}{*}{210} & \multirow{12}{*}{\begin{tabular}[c]{@{}c@{}}Discretized\\ Multivariate\\ Normal\end{tabular}} & \multicolumn{2}{c|}{\multirow{6}{*}{(1,1)}}    & 0          & 0.3        & 56.2          & 55.1         & 59.8          & 59           & 68.2          & 67.2         & 73.6          & 73           \\ \cline{5-14} 
		&                                                                                               & \multicolumn{2}{c|}{}                          & 0.3        & 0.3        & 83.6          & 84.1         & 89.6          & 89           & 92.5          & 92.4         & 95.9          & 95.7         \\ \cline{5-14} 
		&                                                                                               & \multicolumn{2}{c|}{}                          & 0.6        & 0.6        & 100           & 100          & 100           & 100          & 100           & 100          & 100           & 100          \\ \cline{5-14} 
		&                                                                                               & \multicolumn{2}{c|}{}                          & 0.9        & 0.9        & 100           & 100          & 100           & 100          & 100           & 100          & 100           & 100          \\ \cline{5-14} 
		&                                                                                               & \multicolumn{2}{c|}{}                          & 0.3        & 0.6        & 99.8          & 99.8         & 99.9          & 99.9         & 100           & 100          & 100           & 100          \\ \cline{5-14} 
		&                                                                                               & \multicolumn{2}{c|}{}                          & 0.3        & 0.9        & 100           & 100          & 100           & 100          & 100           & 100          & 100           & 100          \\ \cline{3-14} 
		&                                                                                               & \multicolumn{2}{c|}{\multirow{6}{*}{(1,5)}}    & 0          & 0.3        & 20.4          & 20.3         & 24.6          & 24.1         & 27.2          & 26.6         & 26.7          & 26.9         \\ \cline{5-14} 
		&                                                                                               & \multicolumn{2}{c|}{}                          & 0.3        & 0.3        & 37.2          & 36.6         & 43.5          & 43.5         & 47.5          & 47.8         & 49.2          & 49.3         \\ \cline{5-14} 
		&                                                                                               & \multicolumn{2}{c|}{}                          & 0.6        & 0.6        & 91.1          & 91.4         & 94.7          & 95.5         & 97.2          & 97.5         & 98.4          & 98.6         \\ \cline{5-14} 
		&                                                                                               & \multicolumn{2}{c|}{}                          & 0.9        & 0.9        & 100           & 100          & 100           & 100          & 100           & 100          & 100           & 100          \\ \cline{5-14} 
		&                                                                                               & \multicolumn{2}{c|}{}                          & 0.3        & 0.6        & 74.9          & 74.6         & 83.1          & 82.8         & 86            & 85.7         & 89.8          & 90.1         \\ \cline{5-14} 
		&                                                                                               & \multicolumn{2}{c|}{}                          & 0.3        & 0.9        & 96.2          & 95.9         & 98.7          & 98.5         & 99.2          & 99.1         & 99.3          & 99.2         \\ \cline{2-14} 
		& \multirow{12}{*}{\begin{tabular}[c]{@{}c@{}}Multivariate\\ Cauchy\end{tabular}}               & \multicolumn{2}{c|}{\multirow{6}{*}{(1,1)}}    & 0          & 0.3        & 21.1          & 20           & 24.9          & 24.1         & 27.9          & 27.3         & 29.8          & 29.7         \\ \cline{5-14} 
		&                                                                                               & \multicolumn{2}{c|}{}                          & 0.3        & 0.3        & 38.3          & 37.9         & 43.8          & 43.5         & 49.8          & 49.4         & 55.6          & 56.1         \\ \cline{5-14} 
		&                                                                                               & \multicolumn{2}{c|}{}                          & 0.6        & 0.6        & 91.7          & 91.8         & 94.4          & 94.6         & 96.9          & 97.7         & 98.1          & 98           \\ \cline{5-14} 
		&                                                                                               & \multicolumn{2}{c|}{}                          & 0.9        & 0.9        & 99.9          & 99.9         & 100           & 100          & 100           & 100          & 100           & 100          \\ \cline{5-14} 
		&                                                                                               & \multicolumn{2}{c|}{}                          & 0.3        & 0.6        & 74.4          & 74.1         & 83.6          & 84           & 85            & 86           & 88.4          & 88.7         \\ \cline{5-14} 
		&                                                                                               & \multicolumn{2}{c|}{}                          & 0.3        & 0.9        & 95.7          & 95.4         & 97.2          & 97.4         & 98.8          & 99.1         & 99.7          & 99.6         \\ \cline{3-14} 
		&                                                                                               & \multicolumn{2}{c|}{\multirow{6}{*}{(1,5)}}    & 0          & 0.3        & 9.1           & 9            & 12.5          & 11.2         & 13.3          & 13.1         & 14            & 13.7         \\ \cline{5-14} 
		&                                                                                               & \multicolumn{2}{c|}{}                          & 0.3        & 0.3        & 16            & 16.2         & 19.1          & 18.4         & 19.1          & 19.6         & 21            & 21.4         \\ \cline{5-14} 
		&                                                                                               & \multicolumn{2}{c|}{}                          & 0.6        & 0.6        & 51.2          & 50.6         & 55.9          & 56.6         & 64.7          & 65.9         & 72            & 72           \\ \cline{5-14} 
		&                                                                                               & \multicolumn{2}{c|}{}                          & 0.9        & 0.9        & 84.5          & 85.5         & 89.5          & 89.7         & 92.6          & 92.7         & 97.2          & 97.4         \\ \cline{5-14} 
		&                                                                                               & \multicolumn{2}{c|}{}                          & 0.3        & 0.6        & 35.1          & 34.7         & 38.9          & 38.1         & 47.6          & 46.8         & 48.2          & 48           \\ \cline{5-14} 
		&                                                                                               & \multicolumn{2}{c|}{}                          & 0.3        & 0.9        & 62.3          & 61.7         & 67.5          & 67.1         & 72.9          & 73.3         & 78.5          & 78.1         \\ \hline
	\end{tabular}
\label{Table:flex-design2-power}
\end{table}

\begin{table}[!htb]
	\caption{Achieved power ($\times100$) for data generated from Multivariate Normal, Multivariate Log-Normal and Multivariate Cauchy distribution with general missing pattern for $d=2$. Here, $n_1$ is the complete sample size in Design 3; $Q_n$ is the Wald-type statistic proposed in (\ref{Q_N}); $F_n$ is ANOVA-type statistic for small sample approximation proposed in (\ref{ANOVA}). The correlation coefficients are $(\rho_1,\rho_2,\rho_{12})=(0.1,0.1,0.1)$. The nominal Type-I error rate is $\alpha=0.05$.}
	\centering
	\begin{tabular}{|c|c|l|c|c|c|c|c|c|c|c|}
		\hline
		\multicolumn{5}{|c|}{\textbf{Design 3}}                                                                                                                                  & \multicolumn{2}{c|}{$n_1=5$} & \multicolumn{2}{c|}{$n_1=10$} & \multicolumn{2}{c|}{$n_1=20$} \\ \hline
		Distribution                                                                                  & \multicolumn{2}{c|}{$(\sigma_1^2,\sigma_2^2)$} & $\delta_1$ & $\delta_2$ & $Q_n$         & $F_n$        & $Q_n$         & $F_n$         & $Q_n$         & $F_n$         \\ \hline
		\multirow{12}{*}{\begin{tabular}[c]{@{}c@{}}Discretized\\ Multivariate\\ Normal\end{tabular}} & \multicolumn{2}{c|}{\multirow{6}{*}{(1,1)}}    & 0          & 0.3        & 99.8          & 99.8         & 100           & 100           & 100           & 100           \\ \cline{4-11} 
		& \multicolumn{2}{c|}{}                          & 0.3        & 0.3        & 100           & 100          & 100           & 100           & 100           & 100           \\ \cline{4-11} 
		& \multicolumn{2}{c|}{}                          & 0.6        & 0.6        & 100           & 100          & 100           & 100           & 100           & 100           \\ \cline{4-11} 
		& \multicolumn{2}{c|}{}                          & 0.9        & 0.9        & 100           & 100          & 100           & 100           & 100           & 100           \\ \cline{4-11} 
		& \multicolumn{2}{c|}{}                          & 0.3        & 0.6        & 100           & 100          & 100           & 100           & 100           & 100           \\ \cline{4-11} 
		& \multicolumn{2}{c|}{}                          & 0.3        & 0.9        & 100           & 100          & 100           & 100           & 100           & 100           \\ \cline{2-11} 
		& \multicolumn{2}{c|}{\multirow{6}{*}{(1,5)}}    & 0          & 0.3        & 78.8          & 78.8         & 78.2          & 78            & 81.1          & 80.9          \\ \cline{4-11} 
		& \multicolumn{2}{c|}{}                          & 0.3        & 0.3        & 97.6          & 97.5         & 98.2          & 98.4          & 98.4          & 98.5          \\ \cline{4-11} 
		& \multicolumn{2}{c|}{}                          & 0.6        & 0.6        & 100           & 100          & 100           & 100           & 100           & 100           \\ \cline{4-11} 
		& \multicolumn{2}{c|}{}                          & 0.9        & 0.9        & 100           & 100          & 100           & 100           & 100           & 100           \\ \cline{4-11} 
		& \multicolumn{2}{c|}{}                          & 0.3        & 0.6        & 100           & 100          & 100           & 100           & 100           & 100           \\ \cline{4-11} 
		& \multicolumn{2}{c|}{}                          & 0.3        & 0.9        & 100           & 100          & 100           & 100           & 100           & 100           \\ \hline
		\multirow{12}{*}{\begin{tabular}[c]{@{}c@{}}Multivariate\\ Cauchy\end{tabular}}               & \multicolumn{2}{c|}{\multirow{6}{*}{(1,1)}}    & 0          & 0.3        & 82            & 82           & 82.9          & 82.5          & 84.3          & 83.8          \\ \cline{4-11} 
		& \multicolumn{2}{c|}{}                          & 0.3        & 0.3        & 99            & 98.9         & 97.9          & 98            & 97.9          & 98            \\ \cline{4-11} 
		& \multicolumn{2}{c|}{}                          & 0.6        & 0.6        & 100           & 100          & 100           & 100           & 100           & 100           \\ \cline{4-11} 
		& \multicolumn{2}{c|}{}                          & 0.9        & 0.9        & 100           & 100          & 100           & 100           & 100           & 100           \\ \cline{4-11} 
		& \multicolumn{2}{c|}{}                          & 0.3        & 0.6        & 100           & 100          & 100           & 100           & 100           & 100           \\ \cline{4-11} 
		& \multicolumn{2}{c|}{}                          & 0.3        & 0.9        & 100           & 100          & 100           & 100           & 100           & 100           \\ \cline{2-11} 
		& \multicolumn{2}{c|}{\multirow{6}{*}{(1,5)}}    & 0          & 0.3        & 37.2          & 37.6         & 37.4          & 37.4          & 38            & 37.9          \\ \cline{4-11} 
		& \multicolumn{2}{c|}{}                          & 0.3        & 0.3        & 67            & 67.1         & 69.6          & 69.6          & 67.5          & 67.4          \\ \cline{4-11} 
		& \multicolumn{2}{c|}{}                          & 0.6        & 0.6        & 99.9          & 99.9         & 100           & 100           & 99.9          & 99.9          \\ \cline{4-11} 
		& \multicolumn{2}{c|}{}                          & 0.9        & 0.9        & 100           & 100          & 100           & 100           & 100           & 100           \\ \cline{4-11} 
		& \multicolumn{2}{c|}{}                          & 0.3        & 0.6        & 97.8          & 97.7         & 96.9          & 97.2          & 97.2          & 97.2          \\ \cline{4-11} 
		& \multicolumn{2}{c|}{}                          & 0.3        & 0.9        & 99.9          & 99.9         & 100           & 100           & 100           & 100           \\ \hline
	\end{tabular}
\label{Table:flex-design3-power}
\end{table}

\section{Real Data Analysis}
\label{sec:RealData}
In this section, we analyze the ARTIS data described in Section \ref{sec:Motivation}. The data-analytic objective is to test the existence of air-filter intervention effect on the domain variables. Our main goal is to illustrate the application of the methods developed in this paper with a  real data and compare the results with the alternative methods presented in Section \ref{sec:SimuRes}. For the ARTIS data, the interpretation of significant intervention effect is that at least one of the domain variables has nonparametric effect size not equal to 0.5, i.e. data collected after the air-filter intervention tend to be either greater or smaller than those collected before in at least one of the domains. More generally, we want to know whether the air-filter intervention improves quality of life for children with asthma in terms of activity limitation, emotional function or symptoms domains.

As mentioned in Section \ref{sec:Motivation}, one child per family and one visit per intervention period is selected randomly to make the ARTIS data fit the data scheme assumed in this paper. In the selected random sample, $n_c=33$ children have paired data before and after the intervention on all the three domains, $n_1=8$ children have data on all three domains only in the pre-intervention period and $n_2=1$ child has data on all three domains only in the post-intervention period. The estimates of $\bm{p}$ are shown in Table \ref{Table:TestResults} for the methods derived in this paper (All) and, also, for Incomplete and Complete methods. Note that in the ARTIS data, ANOVA-type statistic is more reliable compared with Wald-type statistic due to the small sample size.

From Table \ref{Table:TestResults}, at $\alpha=0.05$, neither Wald-type nor ANOVA-type statistics detect significant intervention effect for all the three methods. Therefore, the air-filter intervention do not have significant tendency to result in larger PAQLQ scores for the three domains. In other words, the intervention does not improve the quality of life for children with asthma in homes using wood-burning stoves. Note that, since the total sample size for Incomplete method is $n_1+n_2=9$ and there is only one sample in the post-intervention group, the results from this method are rather biased and, therefore, its effect size estimates are not reliable. The Imputation method does not provide effect size estimate. Thus, the only conclusion that can be drawn is that there is no significant difference between mean vectors of the domain variables between pre-intervention and post-intervention periods, which is consistent with the nonparametric test results.

\begin{table}[!htb]
	\centering
	\caption{Test statistics and p-values for Domain Variables in ARTIS data. Here, $Q_n$ and $F_n$ represent Wald-type and ANOVA-type statistics proposed in (\ref{Q_N}) and (\ref{ANOVA}); $Q_n^{(1)}$ and $F_n^{(1)}$ are Wald-type and ANOVA-type tests in the Incomplete method; $Q_n^{(2)}$ and $F_n^{(2)}$ are Wald-type and ANOVA-type tests in the Complete method. The nominal Type-I error rate is $\alpha=0.05$.}
	\footnotesize
	\begin{tabular}{|c|c|c|c|c|c|c|c|c|c|c|}
		\hline
		& \multicolumn{3}{c|}{All} & \multicolumn{3}{c|}{Incomplete} & \multicolumn{3}{c|}{Complete} & Multiple Impute \\ \hline
		& $\widehat{\textbf{p}}$ & $Q_n$ & $F_n$ & $\widehat{\textbf{p}}$ & $Q_n^{(1)}$ & $F_n^{(1)}$ & $\widehat{\textbf{p}}$ & $Q_n^{(2)}$ & $F_n^{(2)}$ & p-value \\ \hline
		\begin{tabular}[c]{@{}c@{}}Activity\\ Limitation\end{tabular} & 0.587 & \multirow{3}{*}{0.363} & \multirow{3}{*}{0.134} & 0.813 & \multirow{3}{*}{0.807} & \multirow{3}{*}{0.813} & 0.564 & \multirow{3}{*}{0.3382} & \multirow{3}{*}{0.136} & \multirow{3}{*}{0.368} \\ \cline{1-2} \cline{5-5} \cline{8-8}
		\begin{tabular}[c]{@{}c@{}}Emotional\\ Function\end{tabular} & 0.583 &  &  & 0.870 &  &  & 0.581 &  &  &  \\ \cline{1-2} \cline{5-5} \cline{8-8}
		Symptoms & 0.564 &  &  & 0.750 &  &  & 0.567 &  &  &  \\ \hline
	\end{tabular}
	\label{Table:TestResults}
\end{table}

\section{Discussion}
\label{sec:Discussion}

In many studies, multivariate data from subjects that belong to the same or different treatment groups are collected. In this paper, we have proposed methods that can be used to compare treatment groups for this type of data with general missing patterns. Commonly used approaches include removing incomplete samples, i.e. keep complete cases only, and imputing incomplete cases from the existing data. These strategies are not effective in the sense that they either ignore valuable information or introduce imputation errors. With simulation studies, we have shown that all of these alternatives methods are not efficient in preserving the preassigned Type-I error rate or achieving reasonable power. Therefore, they cannot be recommended for general application. 

The present paper aims at inferential procedures with the fewest assumptions so that ordinal or skewed data are accommodated in a seamless way. In that sense, we derived a fully nonparametric method for estimating and testing the nonparametric effect size applicable for multivariate data. Our nonparametric effect size estimators also allow comparisons among treatment groups on each response variable. In other words, unlike global tests, the proposed procedures can provide more specific information.

With our method, the marginal distribution of each response variable is estimated independently by weighing the corresponding complete and incomplete data. This strategy leads to procedures implementable using ranking routines and reduces the code complexity for calculating effect sizes and covariance matrices. However, other estimation of marginal distributions, such as taking correlation among the variables and between the treatment groups into consideration in the weighting scheme, is likely to be more accurate. Furthermore, although multiple response variables are allowed in this paper, only the two group case is considered. Also, clustered data are not allowed for any of the response variables. In summary, extension to multiple groups and more elaborate estimators of nonparametric effect to accommodate complex data structures will be of interest. For instance, the ARTIS data set involves twenty-three quality of life scores measured in three groups, which can be regarded as multivariate clustered data in factorial design with complete and incomplete clusters. It should be pointed out that the missing data patterns along with dependence structures among response variables and clusters may result in a rather complex covariance matrix in the asymptotic theory. We plan to investigate these problems in future researches. 

\section*{Reference}
\newpage
\bibliographystyle{elsarticle-num-names}
\bibliography{reference}

\newpage

\section{Appendix}
\label{sec:Appendix}

\subsection{Proof of Theorem \ref{thm:consistency}}
\label{Appendix:Proof4.3}
To show $\parallel\widehat{\bm{V}}-\bm{V}\parallel_2^{2}\rightarrow0$, it suffices to show the $L_2$-consistency for each element of $\widehat{\bm{V}}-\bm{V}$, i.e. $\parallel\widehat{v}^{(l,r)}-v^{(l,r)}\parallel_2^2\to0$. Here, $v^{(l,r)}$ is given in (\ref{v_n}) and $\widehat{v}^{(l,r)}=\widehat{v}_{c}^{(l,r)}+\widehat{v}_{1}^{(l,r)}+\widehat{v}_{2}^{(l,r)}$, where $\widehat{v}_{j}^{(l,r)}$ is the $(l,r)^{th}$ entry of $\widehat{\bm{V}}_{j},j\in\{c,1,2\}$ as given in (\ref{V-12}) and (\ref{V-3}). First, rewrite
\begin{equation}
	\widehat{v}_{g}^{(l,r)}=n\frac{n_g}{m_g^2(n_g-1)}\sum_{k=1}^{n_g}(\widehat{Y}_{gk}^{(i)(l)}-\widehat{\overline{Y}}_{g\cdot}^{(i)(l)})(\widehat{Y}_{gk}^{(i)(r)}-\widehat{\overline{Y}}_{g\cdot}^{(i)(r)}),\quad g=1,2\quad\textrm{and}
\end{equation}
\begin{equation}
	\widehat{v}_{c}^{(l,r)}=\frac{n}{n_c(n_c-1)}\sum_{k=1}^{n_c}(\widehat{Z}_k^{(c)(l)}-\widehat{\overline{Z}_\cdot}^{(c)(l)})(\widehat{Z}_k^{(c)(r)}-\widehat{\overline{Z}_\cdot}^{(c)(r)}).
\end{equation}
By triangle inequality, it follows
\begin{equation*}
	\parallel\widehat{v}^{(l,r)}-v^{(l,r)}\parallel_2^2\le\parallel\widehat{v}^{(l,r)}-\tilde{v}^{(l,r)}\parallel_2^2+\parallel\tilde{v}^{(l,r)}-v^{(l,r)}\parallel_2^2.
\end{equation*}
Now, the proof will be complete by showing $\parallel\widehat{v}^{(l,r)}-\tilde{v}^{(l,r)}\parallel_2^2\to0$ and $\parallel\tilde{v}^{(l,r)}-v^{(l,r)}\parallel_2^2\to0$.

To show $\parallel\tilde{v}^{(l,r)}-v^{(l,r)}\parallel_2^2\to0$, it suffices to show $\parallel\widehat{v}_{j}^{(l,r)}-\tilde{v}_{j}^{(l,r)}\parallel_2^2\to0$, $j\in\{c,1,2\}$ by triangle equality. Proof of the special case when $l=r\in\{1,\cdots,d\}$, $j\in\{c,1,2\}$ follows from Theorem 5.1 in \cite{matched-pair-2012}. For other cases where $l\ne r$, we give an example proof for $j=c$.
Observe that
\begin{align*}
	&\quad\parallel\tilde{v}_{c}^{(l,r)}-v_{c}^{(l,r)}\parallel_2^2\displaybreak[0]\\
	&=E\bigg(\frac{n}{n_c(n_c-1)}\sum_{k=1}^{n_c}\big[\frac{n_c}{m_2}Y_{2k}^{(c)(l)}-\frac{n_c}{m_1}Y_{1k}^{(c)(l)}-\big(\frac{n_c}{m_2}\overline{Y}_{2\cdot}^{(c)(l)}-\frac{n_c}{m_1}\overline{Y}_{1\cdot}^{(c)(l)}\big)\big]\times\displaybreak[0]\\
	&\qquad\qquad\qquad\qquad\quad\big[\frac{n_c}{m_2}Y_{2k}^{(c)(r)}-\frac{n_c}{m_1}Y_{1k}^{(c)(r)}-\big(\frac{n_c}{m_2}\overline{Y}_{2\cdot}^{(c)(r)}-\frac{n_c}{m_1}\overline{Y}_{1\cdot}^{(c)(r)}\big)\big]\displaybreak[0]\\
	&\qquad\quad-\frac{n}{n_c}\textrm{Cov}\big(\frac{n_c}{m_2}Y_{2k}^{(c)(l)}-\frac{n_c}{m_1}Y_{1k}^{(c)(l)},\frac{n_c}{m_2}Y_{2k}^{(c)(r)}-\frac{n_c}{m_1}Y_{1k}^{(c)(r)}\big)\bigg)^2\displaybreak[0]\\
	&=E\bigg(\sum_{g=1}^{2}\sum_{g'=1}^{2}(1-2\ell_{g\ne g'})\frac{n}{n_c(n_c-1)}\sum_{k=1}^{n_c}\big(\frac{n_c}{m_g}Y_{gk}^{(c)(l)}-\frac{n_c}{m_g}\overline{Y}_{g\cdot}^{(c)(l)}\big)\big(\frac{n_c}{m_{g'}}Y_{g'k}^{(c)(r)}-\frac{n_c}{m_{g'}}\overline{Y}_{g'\cdot}^{(c)(r)}\big)\displaybreak[0]\\
	&\qquad-\sum_{g=1}^{2}\sum_{g'=1}^{2}(1-2\ell_{g\ne g'})\frac{n}{n_c}\frac{1}{n_c}\sum_{k=1}^{n_c}\textrm{Cov}\big(\frac{n_c}{m_g}Y_{gk}^{(c)(l)}-\frac{n_c}{m_{g'}}Y_{g'k}^{(c)(r)}\big)\bigg)^2\displaybreak[0]\\
	&=E\bigg(\sum_{g=1}^{2}\sum_{g'=1}^{2}(1-2\ell_{g\ne g'})\frac{n_c}{m_g}\frac{n_c}{m_{g'}}\big[\frac{n}{n_c(n_c-1)}\sum_{k=1}^{n_c}(Y_{gk}^{(c)(l)}-\overline{Y}_{g\cdot}^{(c)(l)})(Y_{g'k}^{(c)(r)}-\overline{Y}_{g'\cdot}^{(c)(r)})\displaybreak[0]\\
	&\qquad\quad-\frac{n}{n_c^2}\sum_{k=1}^{n_c}\textrm{Cov}(Y_{gk}^{(c)(l)},Y_{g'k}^{(c)(r)})\big]\bigg)^2\displaybreak[0]\\
	&\le4\sum_{g=1}^{2}\sum_{g'=1}^{2}(\frac{n_c}{m_g}\frac{n_c}{m_{g'}})^2E\bigg(\frac{n}{n_c(n_c-1)}\sum_{k=1}^{n_c}(Y_{gk}^{(c)(l)}-\overline{Y}_{g\cdot}^{(c)(l)})(Y_{g'k}^{(c)(r)}-\overline{Y}_{g'\cdot}^{(c)(r)})\displaybreak[0]\\
	&\qquad-\frac{n}{n_c^2}\sum_{k=1}^{n_c}\textrm{Cov}(Y_{gk}^{(c)(l)},Y_{g'k}^{(c)(r)})\bigg)^2.
\end{align*}
According to the proof of Theorem 3.5 in \cite{werner2006nichtparametrische}, we have
\begin{equation*}
	E\bigg(\frac{n}{n_c(n_c-1)}\sum_{k=1}^{n_c}(Y_{gk}^{(c)(l)}-\overline{Y}_{g\cdot}^{(c)(l)})(Y_{g'k}^{(c)(r)}-\overline{Y}_{g'\cdot}^{(c)(r)})-\frac{n}{n_c^2}\sum_{k=1}^{n_c}\textrm{Cov}(Y_{gk}^{(c)(l)},Y_{g'k}^{(c)(r)})\bigg)^2=O(\frac{1}{n_c}).
\end{equation*}
This implies
\begin{equation*}
	\parallel\tilde{v}_{c}^{(l,r)}-v_{c}^{(l,r)}\parallel_2^2 \le 4\sum_{g=1}^{2}\sum_{g'=1}^{2}(\frac{n_c}{m_g}\frac{n_c}{m_{g'}})^2O(\frac{1}{n_c})=O(\frac{1}{m_1}+\frac{1}{m_2}).
\end{equation*}
Next, we show $\parallel\widehat{v}_{c}^{(l,r)}-\tilde{v}_{c}^{(l,r)}\parallel_2^2\to0$. Notice that
\begin{align*}
	&\quad\parallel\widehat{v}_{c}^{(l,r)}-\tilde{v}_{c}^{(l,r)}\parallel_2^2\\
	&=E(\widehat{v}_{c}^{(l,r)}-\tilde{v}_{c}^{(l,r)})^2\\
	&=(\frac{n}{n_c-1})^2E\bigg(\frac{1}{n_c}\sum_{k=1}^{n_c}\underset{\Delta_k}{\underbrace{\big[(\widehat{Z}_k^{(c)(l)}-\widehat{\overline{Z}}_{\cdot}^{(c)(l)})(\widehat{Z}_k^{(c)(r)}-\widehat{\overline{Z}}_{\cdot}^{(c)(r)})-(Z_k^{(c)(l)}-\overline{Z}_{\cdot}^{(c)(l)})(Z_k^{(c)(r)}-\overline{Z}_{\cdot}^{(c)(r)})\big]}}\bigg)^2\\
	&=(\frac{n}{n_c-1})^2E(\frac{1}{n_c}\sum_{k=1}^{n_c}\Delta_k)^2.
\end{align*}
Let $A=\frac{1}{n_c}\sum_{k=1}^{n_c}\Delta_k$. It follows that $\quad\parallel\widehat{v}_{c}^{(l,r)}-\tilde{v}_{c}^{(l,r)}\parallel_2^2=(\frac{n}{n_c-1})^2E(A^2)$. By Jensen's Inequality and $C_r$-inequality,
\begin{equation}
	\label{Inequality:ASquare}
	\begin{split}
		E(A^2)&\le\frac{1}{n_c}\sum_{k=1}^{n_c}E(\Delta_k^2)\\
		&\le2E\bigg((Z_1^{(c)(l)}-\overline{Z}_{\cdot}^{(c)(l)})^2\big[(\widehat{Z}_1^{(c)(r)}-\widehat{\overline{Z}}_{\cdot}^{(c)(r)})-(Z_1^{(c)(r)}-\overline{Z}_{\cdot}^{(c)(r)})\big]^2\bigg)\\
		&\quad+2E\bigg((\widehat{Z}_1^{(c)(r)}-\widehat{\overline{Z}}_{\cdot}^{(c)(r)})^2\big[(\widehat{Z}_1^{(c)(l)}-\widehat{\overline{Z}}_{\cdot}^{(c)(l)})-(Z_1^{(c)(l)}-\overline{Z}_{\cdot}^{(c)(l)})\big]^2\bigg).
	\end{split}
\end{equation}
Note that the random variables $Y_{gk}^{(c)(l)}$ and $\widehat{Y}_{gk}^{(c)(l)}$ are uniformly bounded by 1 so that $|Y_{gk}^{(c)(l)}-\overline{Y}_{g\cdot}^{(c)(l)}|\le1$ and $|\widehat{Y}_{gk}^{(c)(l)}-\widehat{\overline{Y}}_{g\cdot}^{(c)(l)}|\le1$. Then, by $C_r$-inequality,  it follows that
\begin{equation}
	\label{Inequality:ZSquare}
	\begin{split}
		E(Z_1^{(c)(l)}-\overline{Z}_{\cdot}^{(c)(l)})^2&=E\big[\theta_2(Y_{2k}^{(c)(l)}-\overline{Y}_{2\cdot}^{(c)(l)})-\theta_1(Y_{1k}^{(c)(l)}-\overline{Y}_{1\cdot}^{(c)(l)})\big]^2\\
		&\le2E\big[\theta_2^2(Y_{2k}^{(c)(l)}-\overline{Y}_{2\cdot}^{(c)(l)})^2+\theta_1^2(Y_{1k}^{(c)(l)}-\overline{Y}_{1\cdot}^{(c)(l)})^2\big]\\
		&\le2(\theta_2^2+\theta_1^2)\\
		&\le4.
	\end{split}
\end{equation}
Further, 
\begin{align*}
	&\quad E\bigg((Z_1^{(c)(l)}-\overline{Z}_{\cdot}^{(c)(l)})^2\big[(\widehat{Z}_1^{(c)(r)}-\widehat{\overline{Z}}_{\cdot}^{(c)(r)})-(Z_1^{(c)(r)}-\overline{Z}_{\cdot}^{(c)(r)})\big]^2\bigg)\\
	&\le16E\big[(\widehat{Z}_1^{(c)(r)}-\widehat{\overline{Z}}_{\cdot}^{(c)(r)})-(Z_1^{(c)(r)}-\overline{Z}_{\cdot}^{(c)(r)})\big]^2\\
	&=16E\big[(\widehat{Z}_1^{(c)(r)}-Z_1^{(c)(r)})-(\widehat{\overline{Z}}_{\cdot}^{(c)(r)}-\overline{Z}_{\cdot}^{(c)(r)})]^2\\
	&\le32E(\widehat{Z}_1^{(c)(r)}-Z_1^{(c)(r)})^2+32E(\widehat{\overline{Z}}_{\cdot}^{(c)(r)}-\overline{Z}_{\cdot}^{(c)(r)})^2\\
	&\le64E(\widehat{Z}_1^{(c)(r)}-Z_1^{(c)(r)})^2,
\end{align*}
where the last step follows by Jensen's inequality. According to proof of Theorem 5.1 in \cite{matched-pair-2012}, $E\big(\widehat{F}_g^{(r)}(X_{g'k}^{(c)(r)})-F_g^{(r)}(X_{g'k}^{(c)(r)})\big)^2=O(\frac{1}{m_g})$, $g,g'\in\{1,2\}$ and $g\ne g'$. Then by $C_r$-inequality it follows that
\begin{equation}
	\label{Inequality:ZHat}
	\begin{split}
		E(\widehat{Z}_1^{(c)(r)}-Z_1^{(c)(r)})^2&=E\big(\theta_2(\widehat{Y}_{2k}^{(c)(r)}-Y_{2k}^{(c)(r)})-\theta_1(\widehat{Y}_{1k}^{(c)(r)}-Y_{1k}^{(c)(r)})\big)^2\\
		&\le2\theta_2^2E(\widehat{Y}_{2k}^{(c)(r)}-Y_{2k}^{(c)(r)})^2+2\theta_1^2E(\widehat{Y}_{1k}^{(c)(r)}-Y_{1k}^{(c)(r)})^2\\
		&=2\theta_2^2E\big(\widehat{F}_1^{(r)}(X_{2k}^{(c)(r)})-F_1^{(r)}(X_{2k}^{(c)(r)})\big)^2\\
		&\quad+2\theta_1^2E\big(\widehat{F}_2^{(r)}(X_{1k}^{(c)(r)})-F_2^{(r)}(X_{1k}^{(c)(r)})\big)^2\\
		&=O(\frac{1}{m_1}+\frac{1}{m_2}).
	\end{split}
\end{equation}
Finally, applying (\ref{Inequality:ASquare}), (\ref{Inequality:ZSquare}) and (\ref{Inequality:ZHat}), we get $E(A^2)=O(\frac{1}{m_1}+\frac{1}{m_2})$. Together with Assumption \ref{assumption-3.1}, this completes the proof.

\subsection{Estimator of Covariance Matrix in Section \ref{sec:flex-missing}}
We will estimate covariance matrix $\widehat{\bm{V}}$ by estimating its diagonal elements and off-diagonal elements separately. Let us start with the off-diagonal elements. By independence between the subjects (clusters), it follows for $l\ne r$
\[\begin{split}
	C_1&=\textrm{Cov}(\sum_{k\in S^{(c,l)}}Z_k^{(c)(l)},\sum_{k\in S^{(c,r)}}Z_k^{(c)(r)})\\
	&=\textrm{Cov}(\sum_{k\in S^{(c,l)}\cap S^{(c,r)}}Z_k^{(c)(l)},\sum_{k\in S^{(c,l)}\cap S^{(c,r)}}Z_k^{(c)(r)})\\
	&=\sum_{k\in S^{(c,l)}\cap S^{(c,r)}}n_k\cdot\textrm{Cov}(Z_1^{(c)(l)},Z_1^{(c)(r)})\\
	&:=e_1\cdot\textrm{Cov}(Z_1^{(c)(l)},Z_1^{(c)(r)}).
\end{split}\]
Similarly, we have
\[C_2=\sum_{k\in S^{(c,l)}\cap S^{(2,r)}}n_k\cdot\textrm{Cov}(Z_1^{(c)(l)},Y_{21}^{(i)(r)}):=e_2\cdot\textrm{Cov}(Z_1^{(c)(l)},Y_{21}^{(i)(r)}),\]
\[C_3=\sum_{k\in S^{(c,l)}\cap S^{(1,r)}}n_k\cdot\textrm{Cov}(Z_1^{(c)(l)},Y_{11}^{(i)(r)}):=e_3\cdot\textrm{Cov}(Z_1^{(c)(l)},Y_{11}^{(i)(r)}),
\]
\[C_4=\sum_{k\in S^{(2,l)}\cap S^{(c,r)}}n_k\cdot\textrm{Cov}(Y_{21}^{(i)(l)},Z_1^{(c)(r)}):=e_4\cdot\textrm{Cov}(Y_{21}^{(i)(l)},Z_1^{(c)(r)}),\]
\[C_5=\sum_{k\in S^{(2,l)}\cap S^{(2,r)}}n_k\cdot\textrm{Cov}(Y_{21}^{(i)(l)},Y_{21}^{(i)(r)}):=e_5\cdot\textrm{Cov}(Y_{21}^{(i)(l)},Y_{21}^{(i)(r)}),\]
\[C_6=\sum_{k\in S^{(2,l)}\cap S^{(1,r)}}n_k\cdot\textrm{Cov}(Y_{21}^{(i)(l)},Y_{11}^{(i)(r)}):=e_6\cdot\textrm{Cov}(Y_{21}^{(i)(l)},Y_{11}^{(i)(r)}),\]
\[C_7=\sum_{k\in S^{(1,l)}\cap S^{(c,r)}}n_k\cdot\textrm{Cov}(Y_{11}^{(i)(l)},Z_1^{(c)(r)}):=e_7\cdot\textrm{Cov}(Y_{11}^{(i)(l)},Z_1^{(c)(r)}),\]
\[C_8=\sum_{k\in S^{(1,l)}\cap S^{(2,r)}}n_k\cdot\textrm{Cov}(Y_{11}^{(i)(l)},Y_{21}^{(i)(r)}):=e_8\cdot\textrm{Cov}(Y_{11}^{(i)(l)},Y_{21}^{(i)(r)}),\]
and
\[C_9=\sum_{k\in S^{(1,l)}\cap S^{(1,r)}}n_k\cdot\textrm{Cov}(Y_{11}^{(i)(l)},Y_{11}^{(i)(r)}):=e_9\cdot\textrm{Cov}(Y_{11}^{(i)(l)},Y_{11}^{(i)(r)}),\]
{where $e_j$ is the number of subjects that are involved in the corresponding covariance term $C_j$, $j=1,\cdots,9$.} If $Y_{gk}^{(A)(l)}$, $g\in\{1,2\}$, $A\in\{c,i\}$, $l\in\{1,\cdots,d\}$ were observable, the natural estimators $\tilde{C}_j$ of $C_j$, $j=1,\cdots,9$ would be
\[\tilde{C}_{1}=\frac{e_1}{e_1-1}\sum_{k=1}^{e_1}(Z_k^{(c)(l)}-\overline{Z}_{\cdot}^{(c)(l)})(Z_k^{(c)(r)}-\overline{Z}_{\cdot}^{(c)(r)}),\]
\[\tilde{C}_2=\frac{e_2}{e_2-1}\sum_{k=1}^{e_2}(Z_k^{(c)(l)}-\overline{Z}_\cdot^{(c)(l)})(Y_{2k}^{(i)(r)}-\overline{Y}_{2\cdot}^{(i)(r)}),\]
\[\tilde{C}_3=\frac{e_3}{e_3-1}\sum_{k=1}^{e_3}(Z_k^{(c)(l)}-\overline{Z}_\cdot^{(c)(l)})(Y_{1k}^{(i)(r)}-\overline{Y}_{1\cdot}^{(i)(r)}),\]
\[\tilde{C}_4=\frac{e_4}{e_4-1}\sum_{k=1}^{e_4}(Y_{2k}^{(i)(l)}-\overline{Y}_{2\cdot}^{(i)(l)})(Z_k^{(c)(r)}-\overline{Z}_\cdot^{(c)(r)}),\]
\[\]
\[\tilde{C}_5=\frac{e_5}{e_5-1}\sum_{k=1}^{e_5}(Y_{2k}^{(i)(l)}-\overline{Y}_{2\cdot}^{(i)(l)})(Y_{2k}^{(i)(r)}-\overline{Y}_{2\cdot}^{(i)(r)}),\]
\[\tilde{C}_6=\frac{e_6}{e_6-1}\sum_{k=1}^{e_6}(Y_{2k}^{(i)(l)}-\overline{Y}_{2\cdot}^{(i)(l)})(Y_{1k}^{(i)(r)}-\overline{Y}_{1\cdot}^{(i)(r)}),\]
\[\tilde{C}_7=\frac{e_7}{e_7-1}\sum_{k=1}^{e_7}(Y_{1k}^{(i)(l)}-\overline{Y}_{1\cdot}^{(i)(l)})(Z_k^{(c)(r)}-\overline{Z}_\cdot^{(c)(r)}),\]
\[\tilde{C}_8=\frac{e_8}{e_8-1}\sum_{k=1}^{e_8}(Y_{1k}^{(i)(l)}-\overline{Y}_{1\cdot}^{(i)(l)})(Y_{2k}^{(i)(r)}-\overline{Y}_{2\cdot}^{(i)(r)}),\]
and
\[\tilde{C}_9=\frac{e_9}{e_9-1}\sum_{k=1}^{e_9}(Y_{1k}^{(i)(l)}-\overline{Y}_{1\cdot}^{(i)(l)})(Y_{1k}^{(i)(r)}-\overline{Y}_{1\cdot}^{(i)(r)}).\]
Then we can achieve estimators $\widehat{C}_j$ of $C_j$ by replacing the unobservable random variables $Y_{gk}^{(A)(l)}$ and $\overline{Y}_{g\cdot}^{(A)(l)}$ with their empirical counterparts $\widehat{Y}_{gk}^{(A)(l)}$ and $\widehat{\overline{Y}}_{g\cdot}^{(A)(l)}$ in $\tilde{C}_j,j=1,\cdots,9$. Estimators of diagonal elements ($l=r$) can be obtained analogously. In fact, in this case, due to independence among subjects, it is easy to verify that 
\[\tilde{C}_2=\tilde{C}_3=\tilde{C}_4=\tilde{C}_6=\tilde{C}_7=\tilde{C}_8=0,\]
and for the $l^{th}$ component, natural estimators of $C_1$, $C_5$ and $C_9$ are
\[\tilde{C}_1=\frac{n_c^{(l)}}{n_c^{(l)}-1}\sum_{k\in S^{(c,l)}}(Z_k^{(c)(l)}-\overline{Z}_\cdot^{(c)(l)})^2,\quad\tilde{C}_5=\frac{n_{2}^{(l)}}{n_{2}^{(l)}-1}\sum_{k\in S^{(2,l)}}(Y_{2k}^{(i)(l)}-\overline{Y}_{2\cdot}^{(i)(l)})^2\]
\[\textrm{and}\quad\tilde{C}_9=\frac{n_{1}^{(l)}}{n_{1}^{(l)}-1}\sum_{k\in S^{(1,l)}}(Y_{1k}^{(i)(l)}-\overline{Y}_{1\cdot}^{(i)(l)})^2.\]
Then we replace unobservable random variables with their empirical counterparts to get estimators of the diagonal elements. 

The main difference between covariance matrices for data sets with simple and general missing pattern results from their off-diagonal elements. In Appendix \ref{Appendix:Proof4.3}, we see that elements of covariance matrix for data sets with simple missing pattern can be estimated by summing up estimates of the three covariance decompositions, which are then proved to be consistent. However, for the general missing pattern, the off-diagonal elements consist of 9 covariance decompositions. Since consistency of these covariance decompositions can be established in a manner similar to the proof in Appendix \ref{Appendix:Proof4.3}, the consistency of the covariance matrix follows.

\end{document}